\documentclass[aps,prx,twocolumn,10pt,superscriptaddress]{revtex4-1}

\usepackage{mathrsfs}
\usepackage{bbm}
\usepackage{amsmath}
\usepackage{amssymb}
\usepackage{amsthm}
\usepackage{graphicx}
\usepackage{MnSymbol}
\usepackage{mathtools}
\usepackage{stmaryrd}
\usepackage{enumitem}
\usepackage{bbold}
\usepackage[italicdiff]{physics}
\usepackage{chemformula}
\usepackage{dsfont}
\usepackage{float}

\makeatletter
\usepackage{hyperref}
\usepackage{xcolor}
\definecolor{supcol}{RGB}{10,50,180}
\definecolor{eqcol}{RGB}{220,10,100}
\hypersetup{
	colorlinks,
	citecolor=supcol,
	linkcolor=eqcol,
	urlcolor=supcol
}

\allowdisplaybreaks

\newtheorem{theorem}{Theorem}

\newtheorem{proposition}[theorem]{Proposition}

\newcommand{\mds}{\mathds}
\newcommand{\mca}{\mathcal}

\newcommand{\msf}{\mathsf}
\newcommand{\mfr}{\mathfrak}

\newcommand{\iset}[1]{\llbracket #1\rrbracket}

\DeclareMathOperator{\mvar}{var}

\begin{document}
\title{Time-Cost-Error Trade-Off Relation in Thermodynamics: The Third Law and Beyond}

\author{Tan Van Vu}
\email{tan.vu@yukawa.kyoto-u.ac.jp}
\affiliation{Center for Gravitational Physics and Quantum Information, Yukawa Institute for Theoretical Physics, Kyoto University, Kitashirakawa Oiwakecho, Sakyo-ku, Kyoto 606-8502, Japan}

\author{Keiji Saito}
\email{keiji.saitoh@scphys.kyoto-u.ac.jp}
\affiliation{Department of Physics, Kyoto University, Kyoto 606-8502, Japan}

\date{\today}

\begin{abstract}
Elucidating fundamental limitations inherent in physical systems is a central subject in physics. For important thermodynamic operations such as information erasure, cooling, and copying, resources like time and energetic cost must be expended to achieve the desired outcome within a predetermined error margin. In the context of cooling, the unattainability principle of the third law of thermodynamics asserts that infinite ``resources'' are needed to reach absolute zero. However, the precise identification of relevant resources and how they jointly constrain achievable error remains unclear within the frameworks of stochastic and quantum thermodynamics. In this work, we introduce the concept of separated states, which consist of fully unoccupied and occupied states, and formulate the corresponding thermokinetic cost and error, thereby establishing a unifying framework for a broad class of thermodynamic operations. We then uncover a three-way trade-off relation between time, cost, and error for thermodynamic operations aimed at creating separated states, simply expressed as $\tau\mathcal{C}\varepsilon_{\tau}\ge 1-\eta$. This fundamental relation is applicable to diverse thermodynamic operations, including information erasure, cooling, and copying. It provides a profound quantification of the unattainability principle in the third law of thermodynamics in a general form. Building upon this relation, we explore the quantitative limitations governing cooling operations, the preparation of separated states, and a no-go theorem for exact classical copying. Furthermore, we extend these findings to the quantum regime, encompassing both Markovian and non-Markovian dynamics. Specifically, within Lindblad dynamics, we derive a similar three-way trade-off relation that quantifies the cost of achieving a pure state with a given error. The generalization to general quantum dynamics involving a system coupled to a finite bath implies that the dissipative cost becomes infinite as the quantum system is exactly cooled down to the ground state or perfectly reset to a pure state, thereby resolving an open question regarding the thermodynamic cost of information erasure.
\end{abstract}

\maketitle

\section{Introduction}
\subsection{Background}
The concepts of time, cost, and error play foundational roles across various scientific disciplines. Intuitively, implementing a high-accuracy process at a finite speed unavoidably incurs a cost. This no-free-lunch principle is ubiquitous in nature, ranging from intricate operations in biological systems to fundamental processes in physics. For instance, sensory adaptation systems, which typically employ feedback control mechanisms in noisy environments, must finely balance accuracy, speed, and energy expenditure to function effectively \cite{Lan.2012.NP}. Similarly, it was experimentally demonstrated that fast and accurate computation is unavoidably accompanied by high dissipation \cite{Brut.2012.N,Jun.2014.PRL,Hong.2016.SA,Yan.2018.PRL,Dago.2021.PRL,Scandi.2022.PRL}. Quantifying the trade-offs among these incompatible quantities is crucial not only for understanding the inherent limitations of both natural and engineered systems but also for the development of technologies that operate at the edge of these fundamental constraints.

Thermodynamics, a well-established field with a long history, primarily studies heat, work, energy transformation, and their relationships with the physical properties of both equilibrium and nonequilibrium systems. Over the past two decades, its scope has broadened from macroscopic to microscopic systems operating within finite times, leading to the development of stochastic and quantum thermodynamics \cite{Sekimoto.2010,Seifert.2012.RPP,Vinjanampathy.2016.CP,Goold.2016.JPA,Deffner.2019}. These advanced frameworks provide a theoretical basis for exploring novel trade-off relations inherent in fluctuating systems from an energetic perspective. A prominent example is the thermodynamic uncertainty relation \cite{Barato.2015.PRL,Gingrich.2016.PRL,Horowitz.2017.PRE,Dechant.2018.JSM,Hasegawa.2019.PRL,Timpanaro.2019.PRL,Koyuk.2020.PRL,Miller.2021.PRL.TUR,Vu.2022.PRL.TUR,Horowitz.2020.NP}, which demonstrates a trade-off between cost and precision, asserting that significant dissipation is necessary to achieve high accuracy of observables. This relation yields various implications in nonequilibrium systems, including power-efficiency trade-offs in heat engines \cite{Pietzonka.2018.PRL}, bounds on diffusion extents \cite{Hartich.2021.PRL}, and useful methods for inference of dissipation \cite{Li.2019.NC,Manikandan.2020.PRL,Vu.2020.PRE,Otsubo.2020.PRE,Dechant.2021.PRX}. Another pivotal result is the thermodynamic speed limit \cite{Aurell.2012.JSP,Shiraishi.2018.PRL,Vo.2020.PRE,Ito.2020.PRX,Vu.2021.PRL,Yoshimura.2021.PRL,Salazar.2022.PRE,Vu.2023.PRL.TSL,Kwon.2024.PRE,Delvenne.2024.PRE}, which expresses a trade-off between time and cost in the transformation of a system's state, implying that greater dissipation is required for faster transformations. The speed limit concept also sets a finite-time bound for the thermodynamic cost of information erasure, extending beyond the conventional Landauer principle \cite{Landauer.1961.JRD,Goold.2015.PRL,Proesmans.2020.PRL,Zhen.2021.PRL,Vu.2022.PRL,Lee.2022.PRL,Vu.2023.PRX,Rolandi.2023.PRL}. Other intriguing findings include time-dissipation \cite{Falasco.2020.PRL,Neri.2022.SP} and time-information uncertainty relations \cite{Nicholson.2020.NP,Pintos.2022.PRX}, which characterize the trade-offs between time and cost in changing observables.

\begin{figure}[t]
\centering
\includegraphics[width=1.0\linewidth]{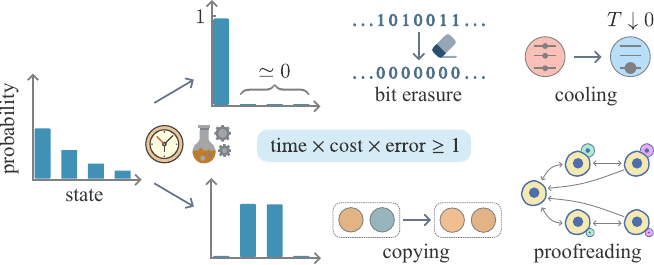}
\protect\caption{In thermodynamic processes that reduce the observation probability to zero, a trade-off relation between time, cost, and error naturally arises. Examples include: (i) erasing information from bits (where the probability of all microstates corresponding to the logical state `$1$’ is reduced to zero), (ii) cooling toward the ground state (where the probability of excited states is reduced to zero), (iii) copying (where the probability of states differing between the source and target is reduced to zero), and (iv) kinetic proofreading (where the probability of pathways producing incorrect products is reduced to zero).}\label{fig:Cover}
\end{figure}

In this study, we focus on a general class of thermodynamic operations that require high accuracy and small thermodynamic cost within finite time. A common feature of many important thermodynamic operations is the aim to achieve one or more desired states while ideally reducing the probabilities of all other states to zero (see Fig.~\ref{fig:Cover} for illustration).
This class of processes holds significance not only in physics but also in biology. A representative example is the task of bit erasure \cite{Bennett.1982.IJTP,Rio.2011.N,Dillenschneider.2009.PRL,Munson.2025.PRXQ}. In this task, the goal is to reduce the probability of observing the bit in the logical state `$1$' to zero, ensuring that the bit is reliably reset to the logical state `$0$', regardless of its initial state. Preparing pure states or cooling systems to ground states are also crucial thermodynamic operations \cite{Schulman.2005.PRL,Chan.2011.N,Teufel.2011.N,Peterson.2016.PRL,Guo.2019.PRL,Soldati.2022.PRL,Taranto.2025.PRL}. Cooling a system to its ground state involves reducing the probability of finding the system in any higher energy eigenstates to zero. In the biological context, a notable example is the role of proofreading mechanism \cite{Hopfield.1974.PNAS,Bennett.1979.B,Sartori.2015.PRX} in information processing within living systems. This mechanisms incorporates irreversible intermediate steps that enhance discrimination between two possible reaction pathways, effectively reducing the likelihood of following the incorrect pathway that leads to erroneous products. 

For the processes introduced above, a trade-off between error and resource usage is intuitively expected. Specifically, the achievable error in the probability is always constrained by the available resources of time and cost. For the problem of preparing pure states, the no-go theorem has been established \cite{Franco.2013.SR,Wu.2013.SR,Ticozzi.2014.SR}, asserting that it is impossible to create a pure state from a mixed one using a unitary transformation applied to the system and a finite-dimensional reservoir. Furthermore, quantitative relations that link the achievable error to the volume and energy bandwidth of the reservoir have been explored in the literature \cite{Masanes.2017.NC,Wilming.2017.PRX,Scharlau.2018.Q,Clivaz.2019.PRL}. Despite these advancements, previous studies have largely focused on specific thermodynamic operations, such as cooling, without providing a unified framework for the general thermodynamic constraints---such as trade-offs between time, cost, and error---that underlie a broad class of important processes, including information erasure, cooling, copying, and kinetic proofreading. This leaves several fundamental questions open: What are the relevant resources in these tasks? How do these resources jointly constrain the achievable error? Moreover, a definitive relation that clearly elucidates the trade-off among three fundamental yet incompatible quantities---{\it time}, {\it cost}, and {\it error}---remains lacking.

\subsection{Summary of results}
In this work, we resolve this open problem within a unified framework across both classical and quantum domains. To obtain the general thermodynamic constraints, we introduce the concept of {\it separated states}, which consist of two classes of states: desired states that the operation aims to achieve and undesired states, for which the probabilities should ideally be zero after the operation. Building on this concept, we quantitatively elucidate a three-way trade-off relation among the quantities of time, cost, and error, demonstrating that infinite resources of time and cost are required to reduce the error to zero [Eq.~\eqref{eq:main.res}]. Its explicit form reads
\begin{equation}
	\tau\mca{C}\varepsilon_\tau\ge 1-\eta~(\text{or}~\varepsilon_0\tau\mca{C}\ge\eta^{-1}-1),
\end{equation}
where $\tau$ and $\mca{C}$ denote time and cost, respectively, $\varepsilon_\tau$ ($\varepsilon_0$) is the error at the final (initial) time, and $\eta=\varepsilon_\tau/\varepsilon_0$ is the relative error (Sec.~\ref{sec:tradeoff}). This transition from conceptual intuition to rigorous mathematical formulation not only clarifies theoretical ambiguities but also paves the way for several essential applications in nonequilibrium thermodynamics, as explained below.

(i) As the first application of this trade-off relation, we establish a thermodynamic bound that embodies the third law of thermodynamics in the form of the unattainability principle for Markovian dynamics (Sec.~\ref{sec:uptl}). This bound determines a lower limit on the error in executing the task of preparing ground states or pure states in terms of time and cost. Additionally, as a direct consequence, we obtain a lower bound on the achievable temperature when cooling the system toward the ground state. These bounds provide a quantitatively comprehensive understanding of the third law of thermodynamics. 

(ii) Next, expanding beyond the traditional concepts of ground states and pure states, we investigate the problem of preparing separated states, where the probability is restricted to a specific set of desired states (Sec.~\ref{sec:lpss}). By applying the trade-off relation, we establish a no-go theorem for the preparation of these states and derive a thermodynamic bound on the achievable error in terms of time and cost. 

(iii) In the third application, we address the copying problem, where the state of a source is replicated into the state of a target using a machine (Sec.~\ref{sec:copying}). More precisely, we formulate a no-go theorem for exact classical copying, covering a broader range of scenarios than the conventional cloning problem. Furthermore, we demonstrate the hardness of exact copying from an information-theoretic perspective, highlighting how finite resources fundamentally constrain the attainment of the maximum mutual information between the source and the target.

(iv) Last, we demonstrate the universality of our findings by generalizing them to quantum cases, applicable to both Markovian dynamics and generic quantum systems coupled to finite-size reservoirs (Sec.~\ref{sec:quantum.gen}). Remarkably, they resolve an open question regarding the thermodynamic cost of erasing information: the cost diverges in the limit of perfect erasure.

\subsection{Relevant literature}
We briefly review the relevant literature and highlight how our results differ from existing works. While the third law of thermodynamics has been extensively discussed in the context of classical thermodynamics for macroscopic systems, our focus here is on microscopic dynamics, which is more relevant to our study.

In quantum settings, the third law has been investigated using various models. In Refs.~\cite{Levy.2012.PRE,Kolar.2012.PRL}, the unattainability principle of zero temperature was studied through a quantum refrigerator model composed of two-level systems operating in a stationary state. The cooling dynamics of a cold bath, weakly coupled to the refrigerator, was characterized by the exponent $\zeta$ in the relation $dT_c/dt \propto -T_c^\zeta$ as the bath temperature approaches absolute zero ($T_c\to 0$). The validity of the third law can be assessed by checking whether $\zeta \ge 1$.

Another direction involves the study of approximate cooling, where the goal is to bring the target system close to its ground state. In a general setup where the target is coupled to a finite-dimensional thermal bath and the composite system undergoes a unitary evolution, Ref.~\cite{Reeb.2014.NJP} demonstrated that the cooling error is constrained by the energy bandwidth of the bath Hamiltonian.

Further development was made in Ref.~\cite{Wilming.2017.PRX}, which analyzed cooling near absolute zero within the framework of catalytic thermal operations. In that work, a composite system comprising the target, a resource, a thermal bath, and a catalyst undergoes an energy-conserving unitary transformation. At the final time, the target is expected to be close to its ground state, while the catalyst approximately returns to its initial state. Assuming the target initially thermalizes with its environment, necessary and sufficient conditions for approximate cooling were derived, formulated in terms of a function $V_\beta(\varrho;H)$ called the ``vacancy.'' The third law is then interpreted through the divergence of this function, which implies that an infinite amount of ``resources'' is required for exact ground-state preparation.

In a related yet distinct setting, Ref.~\cite{Masanes.2017.NC} considered a unitary evolution involving a target, a thermal bath, and a work storage system, where the global unitary preserves the total energy. In this setup, the cooling error was shown to be bounded by both the bath volume and the worst-case work fluctuations. When the bath Hamiltonian satisfies a locality condition, the Lieb-Robinson bound \cite{Lieb.1972.CMP} can be employed to relate the cooling error to the finite operational time; in this case, the third law emerges as a trade-off between the cooling error and the required time.

These studies collectively provide significant insights into the theoretical constraints imposed by the third law of thermodynamics on cooling processes near absolute zero. However, it remains unclear how thermodynamic and kinetic costs, such as entropy production and dynamical activity, constrain cooling limitations together with the resource of time. Our study addresses this gap by, for the first time, elucidating the interplay between time and these costs within the frameworks of stochastic and quantum thermodynamics. We consider a general setup where the target system is subject to a time-dependent protocol, and the operation does not necessarily conserve the total energy. Our work advances the understanding of the third law of thermodynamics in three key aspects:
\begin{enumerate}
	\item Generalization of the third law: We extend the third law (i.e., a constraint on cooling to absolute zero) to a more general concept of creating separated states, encompassing a wide range of fundamental thermodynamic operations such as information erasure and copying.
	\item Identification of resources: We identify the relevant resources and demonstrate that time is not the sole determining factor; other contributions, such as dissipation and frenesy, also play essential roles.
	\item Quantification of trade-offs: We rigorously formulate the third law as a \emph{thermodynamic} trade-off between time, cost, and error within the frameworks of stochastic and quantum thermodynamics.
\end{enumerate}
The remainder of the paper is organized as follows. Section \ref{sec:setup} introduces the concept of separated states, dynamics, and the notions of error and cost used in this study. In Sec.~\ref{sec:results}, we present our main results, including the central relation and its applications in nonequilibrium thermodynamics. Section \ref{sec:demon} provides numerical demonstrations of our findings. Finally, Sec.~\ref{sec:sum.out} gives a summary and outlook. Detailed analytical calculations and derivations are presented in the Appendices.

\section{Setup}\label{sec:setup}
\subsection{Separated states}
We introduce the concept of separated states, which encompass relevant states in nonequilibrium thermodynamic operations and play a key object in this study.
A $d$-dimensional distribution $\ket{p}=[p_1,\dots,p_d]^\top$ is considered separated with respect to a nonempty set $\mfr{u}\subset\{1,\dots,d\}\eqqcolon\iset{d}$ if no states in $\mfr{u}$ are probable; that is,
\begin{equation}
	p_{\mfr{u}}\coloneqq\sum_{n\in \mfr{u}}p_n=0.
\end{equation}
Equivalently, the probability supported by the complement $\mfr{d}=\iset{d}\setminus\mfr{u}$ is equal to $1$ (i.e., $p_{\mfr{d}}\coloneqq\sum_{n\in \mfr{d}}p_n=1$).
The sets $\mfr{u}$ and $\mfr{d}$ represent the undesired and desired states, respectively, indicating the states we want the system to avoid or occupy (see Fig.~\ref{fig:SepStateMarkov}).
For instance, in the bit erasure operation depicted in Fig.~\ref{fig:Cover}, the desired state is the reset state `$00\dots 0$', while all other states are considered undesired. Similarly, in the cooling operation, the desired state is the ground state, with all excited states being undesired. It is important to note that the notion of separated distributions can include multiple desired states (i.e., $|\mfr{d}|\ge 2$). For example, in the copying process, the desired states are those where the bit and memory have the same values, while all other states are undesired.

\begin{figure}[t!]
\centering
\includegraphics[width=0.9\linewidth]{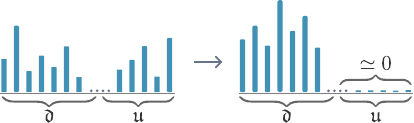}
\protect\caption{The schematic illustrates the preparation of a separated state in Markov jump processes. After finite-time operations, the final distribution is approximately separated, i.e., the total probability of finding the system in the undesired states belonging to the set $\mfr{u}$ is nearly zero.}\label{fig:SepStateMarkov}
\end{figure}

\subsection{Dynamics, error, and cost}
We consider a thermodynamic process of a discrete-state system, which has $d$ internal states.
The term ``internal states'' is context-dependent and can refer to a variety of configurations: for instance, mesostates in coarse-grained dynamics, energy eigenstates in quantum jump processes, configurations of chemical species in reaction networks, or discrete-system states in other physical scenarios.
The system is coupled to thermal reservoirs and may be subject to external controls, thus being driven out of equilibriums.
The system can be characterized by a time-dependent probability distribution $\ket{p_t}=[p_1(t),\dots,p_d(t)]^\top$, where $p_n(t)$ denotes the probability of finding the system in state $n$ at time $t$.
Assuming the system's dynamics is a Markov jump process, the time evolution of this probability distribution is described by the master equation,
\begin{equation}
	\ket{\dot{p}_t}=\msf{W}_t\ket{p_t},\label{eq:mas.eq}
\end{equation}
where the dot $\cdot$ denotes the time derivative and $\msf{W}_t=[w_{mn}(t)]$ is the time-dependent transition rate matrix satisfying the normalization conditions $\sum_{m=1}^dw_{mn}(t)=0$ for any $n\in\iset{d}$.
To guarantee the thermodynamical consistency, we assume the local detailed balance \cite{Seifert.2012.RPP}; that is, the log of the ratio of transition rates is related to the entropy change in the environment.
While we consider classical systems here, the generalization to quantum cases is straightforward and will be discussed in a subsequent section.

Now we introduce the notions of error and cost in the thermodynamic process.
The definition of error should depend on the task's purpose.
In this study, we focus on a general class of thermodynamic tasks wherein one tries to reduce the probability of observing the system in states belonging to a set $\mfr{u}\subset\iset{d}$ to zero.
In other words, the objective is to create a separated state with respect to the set $\mfr{u}$ of undesired states (Fig.~\ref{fig:SepStateMarkov}).
As demonstrated later, prominent examples include the cooling problem in the third law of thermodynamics, preparation of separated states, classical copying, and kinetic proofreading (Fig.~\ref{fig:Cover}).
It is noteworthy that there is an infinite number of ways to define the error even for a specific problem.
Nevertheless, to achieve a meaningful trade-off relation, the error should possibly take values in the range of $[0,+\infty)$.
More precisely, the error should be zero when the task is perfectly accomplished and be positive when no action and cost is involved.
Keeping this in mind, we define the error at time $t$ in the following manner:
\begin{equation}\label{eq:error.def}
	\varepsilon_t\coloneqq-[\ln p_{\mfr{u}}(t)]^{-1},
\end{equation}
where $p_{\mfr{u}}(t)=\sum_{n\in\mfr{u}}p_n(t)$ is the probability of finding the system in undesired states of the set $\mfr{u}$.
In information theory, $-\ln p$ for a given probability $p$ is known as ``surprisal,'' representing the amount of information gained when a specific event occurs, and is directly related to Shannon entropy. Hence, $\varepsilon_t$ can be interpreted as the reciprocal of surprisal, meaning that reducing $\varepsilon_t$ corresponds to increasing surprisal.
Intuitively, $\varepsilon_t$ quantifies how close the total probability supported by the $\mfr{u}$ space is to zero at time $t$.
By this definition, one can observe that $\varepsilon_0$ can be infinite if $p_{\mfr{u}}(0)=1$, and $\varepsilon_\tau$ approaches zero as the process becomes perfect at the final time [i.e., when $p_{\mfr{u}}(\tau)\to 0$].

The cost associated with any process can be mainly divided into two contributions: kinetic and thermodynamic.
A kinetic contribution is pertinent to how strong the system's activity can be and should be present even when the system is in equilibrium.
One of the feasible ways to properly capture such contributions is through the maximum escape rate from states, leading us to define the quantity
\begin{equation}\label{eq:kine.cont}
	\omega_t\coloneqq\max_{n\in\mfr{u}}\sum_{m\in\mfr{d}}w_{mn}(t).
\end{equation}
More precisely, $\omega_t$ is the maximum escape rate from one undesired state in the $\mfr{u}$ space to the $\mfr{d}$ space, quantifying the ability of changing from undesired states to desired states at time $t$ [Fig.~\ref{fig:Cost}(a)].
The kinetic nature of this quantity is further highlighted by the fact that $\omega_t$ corresponds to the maximal jump frequency from the $\mfr{u}$ space to the $\mfr{d}$ space, maximized over all possible system distributions.
On the other hand, a thermodynamic contribution reflects how far the system is driven from equilibriums and is related to irreversible dissipation.
Entropy production, which quantifies the degree of thermodynamic irreversibility, is commonly invoked for the examination of thermodynamic costs in nonequilibrium systems \cite{Landi.2021.RMP}.
In the framework of stochastic thermodynamics \cite{Seifert.2012.RPP}, the entropy production rate can be expressed as the sum of the product of the probability currents and the thermodynamic forces, given by
\begin{equation}\label{eq:thermo.cont}
	\sigma_t\coloneqq\sum_{m\neq n}w_{mn}(t)p_n(t)\ln\frac{w_{mn}(t)p_n(t)}{w_{nm}(t)p_m(t)}.
\end{equation}
The jump frequency between desired states in the set $\mfr{d}$ and other undesired states belonging to $\mfr{u}$ is relevant for reducing the probability $p_{\mfr{u}}(t)$ and can be calculated as
\begin{equation}\label{eq:act.cont}
	a_t\coloneqq\sum_{n\in \mfr{u},m\in \mfr{d}}[w_{mn}(t)p_n(t)+w_{nm}(t)p_m(t)].
\end{equation}
This quantity is also known as dynamical activity in the literature \cite{Maes.2020.PR} and imposes important constraints on speed and precision of dynamics \cite{Shiraishi.2018.PRL,Garrahan.2017.PRE,Terlizzi.2019.JPA,Hasegawa.2020.PRL,Vo.2022.JPA}.

\begin{figure}[t!]
\centering
\includegraphics[width=1.0\linewidth]{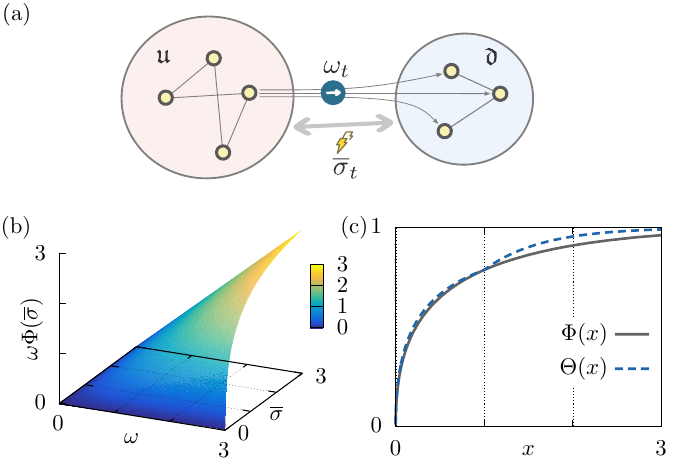}
\protect\caption{(a) Schematic illustration of the kinetic and thermodynamic contributions $\omega_t$ and $\overline{\sigma}_t$. (b) The behavior of the cost term as the function of the kinetic ($\omega$) and thermodynamic ($\overline{\sigma}$) contributions. The cost term vanishes when either contribution is zero and increases as a function of the other contribution when one is fixed. (c) Visualization of the function $\Phi(x)$ and its upper bound $\Theta(x)$, which are both monotonically increasing functions.}\label{fig:Cost}
\end{figure}

Using the kinetic and thermodynamic contributions introduced above, we define the thermokinetic cost associated with the thermodynamic process as
\begin{equation}\label{eq:cost.def}
	\mca{C}\coloneqq\omega\Phi(\overline{\sigma}),
\end{equation}
where
\begin{equation}\label{eq:ther.cost.def}
	\omega\coloneqq\max_{0\le t\le\tau}\omega_t,~\overline{\sigma}\coloneqq\ev{\overline{\sigma}_t}_\tau,~\overline{\sigma}_t\coloneqq\frac{\sigma_t}{a_t}.
\end{equation}
Here, $\ev{x_t}_\tau$ denotes the time average of an arbitrary time-dependent quantity $x_t$, defined as $\ev{x_t}_\tau\coloneqq\tau^{-1}\int_0^\tau\dd{t}x_t$. $\Phi$ is a monotonically increasing function, defined as $\Phi\coloneqq 1-\phi^{-1}\in[0,1)$, where $\phi^{-1}$ denotes the inverse function of $\phi(x)\coloneqq(x+1)^{-1}(x-1)\ln x$. Note that $\sigma_t$ is the total entropy production rate of the original dynamics \eqref{eq:mas.eq} and the dynamical activity $a_t$ represents the average number of jumps per unit time between $\mfr{u}$ and $\mfr{d}$ spaces. Therefore, $\overline{\sigma}$ quantifies the time average of entropy production per transition between these two spaces, being crucial in quantifying the thermodynamic cost. The cost $\mca{C}$ is defined as a dimensionless quantity with respect to time, as our aim is to clearly distinguish between the resources of time and thermokinetic contributions. Near equilibrium, where $\overline{\sigma}$ is small, $\Phi$ behaves approximately as $\Phi(\overline{\sigma})\approx\sqrt{2\overline{\sigma}}$. In general, $\Phi(x)\le\Theta(x)$, where $\Theta(x) \coloneqq 1 - e^{-c \max(\sqrt{x}, x)}$ is a monotonically increasing function with $c \approx 1.543$ \cite{fnt1} (see Proposition \ref{prop4} for the proof). The behavior of the cost term is quantitatively illustrated in Figs.~\ref{fig:Cost}(b) and \ref{fig:Cost}(c). In the following, we delve into the essence of $\mca{C}$ and provide further insights into its significance.

First, we explain why both thermodynamic and kinetic contributions are essential in defining the cost $\mca{C}$. It has been shown that probability distributions can be transformed with arbitrarily small entropy production \cite{Vu.2021.PRL2,Remlein.2021.PRE,Dechant.2022.JPA}; therefore, entropy production alone is insufficient to characterize state transformations, necessitating the inclusion of kinetic contributions \cite{Vu.2023.PRX}. Moreover, reaching new states within finite time inevitably drives the system out of equilibrium, resulting in unavoidable entropy production that contributes to the transformation cost. These facts clearly indicate that both contributions are crucial in characterizing the cost. Such an integration of thermodynamic and kinetic contributions also appears in other contexts. For instance, the relative fluctuation of currents was shown to be constrained by both thermodynamic and kinetic factors \cite{Horowitz.2020.NP,Vo.2022.JPA}.

Second, we examine what the cost $\mca{C}$ captures through its thermodynamic and kinetic contributions. As evident from its definition, $\mca{C}$ encapsulates the combined effects of frenesy and dissipation required to drive the system out of equilibrium and reach the desired final distribution within a finite time, expressed via the product $\omega\Phi(\overline{\sigma})$. This formulation clearly highlights the \emph{trade-off} interplay between these two factors: a reduction in one necessitates a compensatory increase in the other. In other words, both contributions cannot be simultaneously small when a substantial cost is required to achieve low error. Notably, such a nontrivial interplay between kinetic and thermodynamic factors also emerges in the context of fluctuation suppression \cite{Vo.2022.JPA}. Furthermore, when either $\omega$ or $\overline{\sigma}$ is fixed, the minimal value of the remaining factor needed to accomplish the task can be determined. From a quantitative standpoint, $\mca{C}$ embodies the intrinsic nature of cost---any increase in either $\omega$ or $\overline{\sigma}$ inevitably leads to a higher overall cost.
It is worth noting that the role of kinetics is both broad and significant: variations in the transition rates influence not only the maximum escape rate but also the entropy production per transition in a nontrivial manner. As a result, the kinetic and thermodynamic contributions are, in general, interdependent rather than independent. Nevertheless, since $\Phi \leq 1$, the kinetic contribution plays a decisive role: the divergence of the cost $\mca{C}$ necessarily entails the divergence of $\omega$.
To gain deeper insights, we establish the following hierarchical relation: 
\begin{equation}
	\mca{C}_1\le\mca{C}\le \mca{C}_2,
\end{equation}
where the lower and upper bounds of the cost are defined as
\begin{align}
	\mca{C}_1&\coloneqq \ev{\omega_t\Phi\qty(\overline{\sigma}_t)}_\tau,\\
	\mca{C}_2&\coloneqq \omega.
\end{align}
Notably, each term in the hierarchy ($\mca{C}_1 \le \mca{C} \le \mca{C}_2$) can serve as a quantifier of cost. While employing $\mca{C}_2$ leads to a looser bound, it offers a clearer interpretation, as $\mca{C}_2$ solely reflects the maximum escape rate, entirely neglecting thermodynamic contributions. In contrast, $\mca{C}_1$ incorporates both kinetic and thermodynamic factors, making it qualitatively similar to $\mca{C}$. Quantitatively, $\mca{C}_1$ represents the time-averaged instantaneous cost, thereby yielding a tighter bound. For precise analysis, $\mca{C}_1$ might be preferable to ensure tightness while accurately capturing the nature of the cost. Nonetheless, in practical applications, $\mca{C}$ or even $\mca{C}_2$ may be employed for simpler interpretation and estimation.

\section{Results}\label{sec:results}
\subsection{Time-cost-error trade-off relation}\label{sec:tradeoff}
With the quantifications of cost and error above, we are now ready to explain our central result. 
We rigorously prove that there exists a fundamental trade-off between time, cost, and error, which holds true for arbitrary times and protocols.
The trade-off is explicitly given by
\begin{equation}\label{eq:main.res}
	\tau\mca{C}\varepsilon_\tau\ge 1-\eta,
\end{equation}
where $\eta\coloneqq{\varepsilon_\tau}/{\varepsilon_0}$ is the relative error, which should vanish for a perfect process.
Alternatively, it can also be expressed in a more intuitive form of $\varepsilon_0\tau\mca{C}\ge\eta^{-1}-1$. Note that the lower bound can become negative when $\eta>1$, resulting in a trivial inequality. Therefore, it is more relevant to focus solely on meaningful processes where the error is eventually reduced (i.e., $\eta\le 1$).
The proof of the relation \eqref{eq:main.res} is presented in Appendix \ref{app:proof.main.res}.

Below, we discuss the crucial physics derived from this three-way trade-off relation.
When the initial error is zero, meaning $p_{\mfr{u}}(0)=0$, the task of thermodynamic operation is already complete.
Therefore, we focus on the scenario where the initial error $\varepsilon_0$ is nonzero, which can be finite or infinite.
The inequality implies that achieving a perfect thermodynamic operation, where $\varepsilon_\tau =0$, is impossible without either infinite operation time or infinite cost. 
This does not only qualitatively yield a no-go theorem for the general task of reducing the probability to zero but also provides a quantification onto how much the resource must be invested to achieve a predetermined error.
Moreover, the derived inequality is tight and can be saturated for generic thermodynamic processes (see Sec.~\ref{sec:demon} for numerical demonstration).
In addition, the quantum extensions of this result are obtained and discussed in Sec.~\ref{sec:quantum.gen}.

The three-way trade-off relation \eqref{eq:main.res} can be interpreted as a generalization of the third law of thermodynamics within Markovian dynamics. 
In the literature, the third law comprises two key statements. The first concerns the system's entropy at absolute zero temperature, known as Nernst's heat theorem \cite{Nernst.1906} and Planck's formulation \cite{Planck.1911}. The second statement is the {\it unattainability principle}, which asserts that bringing any system to its ground state requires an infinite number of thermodynamic operations or infinite time \cite{Allahverdyan.2011.PRE,Levy.2012.PRE,Silva.2016.PRE,Masanes.2017.NC,Wilming.2017.PRX,Scharlau.2018.Q,Clivaz.2019.PRL,Buffoni.2022.PRL,Taranto.2023.PRXQ}. Our focus is on the latter statement. 
The three-way trade-off relation allows us to generalize the unattainability principle to: ``{\it bringing any system to a separated state (determined by $\mfr{u}$ and $\mfr{d}$ spaces) necessitates either an infinite cost $\mca{C}$ or infinite time $\tau$}.'' This relation quantifies and reinforces the unattainability principle in a broader context.

The derivation of the relation \eqref{eq:main.res} assumes microscopic reversibility, meaning all transitions are bidirectional. However, when the dynamics involves unidirectional transitions (i.e., certain transitions have no reverse counterparts), the quantity $\overline{\sigma}$ diverges because the entropy production rate defined in Eq.~\eqref{eq:thermo.cont} becomes infinite. As a result, the cost is governed solely by the kinetic factor, with no thermodynamic contribution (i.e., $\mca{C}=\omega$). To obtain a more refined bound in the presence of unidirectional transitions, we generalize the cost by introducing $\mca{C}^{\rm uni}$, which separates bidirectional and unidirectional contributions:
\begin{equation}\label{eq:main.res.uni}
	\mca{C}^{\rm uni}\coloneqq\omega^{b}\Phi\qty(\overline{\sigma}^{b})+\omega^{u}.
\end{equation}
Here, superscripts $b$ and $u$ indicate that the corresponding quantity is calculated using the bidirectional and unidirectional parts, respectively.
As shown, $\mca{C}^{\rm uni}$ comprises two terms: the first term captures the cost associated with the bidirectional transitions, identical in form to Eq.~\eqref{eq:cost.def}, while the second term accounts for the purely kinetic cost arising from the unidirectional transitions.
With this generalized cost, the relation \eqref{eq:main.res} can be analogously derived; the proof is presented in Appendix \ref{app:proof.main.res.uni}.
Note that this generalized cost reduces exactly to the conventional one \eqref{eq:cost.def} in the absence of unidirectional transitions (i.e., when $\omega^u=0$).

We have several additional remarks on the three-way trade-off relation. First, the trade-off relation can also be derived using the nonadiabatic entropy production rate as part of the thermodynamic cost, instead of the total entropy production rate.
It has been shown that the total entropy production, which defines the cost, can be decomposed into adiabatic and nonadiabatic contributions \cite{Hatano.2001.PRL,Esposito.2010.PRL}. Roughly speaking, adiabatic entropy production relates to the dissipation required to maintain the stationary state. Conversely, the nonadiabatic component, which vanishes in the stationary state, characterizes the additional dissipation needed for the transition between stationary states. Notably, only the nonadiabatic entropy production is responsible for state transformations \cite{Shiraishi.2018.PRL,Vo.2020.PRE}. Therefore, it is reasonable to anticipate a similar trade-off relation involving the nonadiabatic cost, and we demonstrate that this is indeed the case.
The derivation is detailed in Appendix \ref{app:nonad.tradeoff}.

Second, for thermodynamic processes driven by time-independent protocols, a more compact trade-off solely in terms of entropy production can be established (see Appendix \ref{app:relax.tradeoff}). Furthermore, we can demonstrate that excess entropy production limits how much the error can be reduced. Such a relationship can have significant implications for understanding the connection between thermodynamic cost and functional performance. For instance, in a proofreading mechanism, certain kinetic parameters may be abruptly quenched to adapt to changes in the surrounding environment, leading the system to subsequently relax toward another stationary state. The achievable error in this new stationary state is then constrained by the excess cost incurred during the relaxation process.

Third, we clarify the qualitative differences between our result and existing trade-off relations. One of the most prominent results is the thermodynamic uncertainty relation (TUR) \cite{Barato.2015.PRL,Gingrich.2016.PRL,Horowitz.2017.PRE}, which takes the form
\begin{equation}\label{eq:TUR}
	\tau\overline{\Sigma}\varepsilon_\jmath\ge 2
\end{equation}
for stationary Markov jump processes. Here, $\overline{\Sigma}$ represents the entropy production rate and $\varepsilon_\jmath\coloneqq\mvar[\jmath]/\ev{\jmath}^2$ is the relative fluctuation of current $\jmath$.
Interpreting $\varepsilon_\jmath$ as the current error, the inequality \eqref{eq:TUR} can be seen as a trade-off between time, cost, and error.
However, this error pertains to the fluctuation of trajectory currents, which is entirely unrelated to the problem of reducing probability considered in this study.
Another relevant result is the thermodynamic speed limit (TSL), formulated as \cite{Dechant.2022.JPA,Lee.2022.PRL,Vu.2023.PRX}:
\begin{equation}\label{eq:tsl}
	\tau f(\overline{\Sigma},\overline{\mca{A}})\ge\mca{W}(p_0,p_\tau),
\end{equation}
where $\overline{\mca{A}}$ is the time-averaged dynamical activity, $\mca{W}$ denotes the Wasserstein distance between probability distributions, and $f$ is a function whose explicit form is omitted.
The TSL can be interpreted as a trade-off between time and cost in state transformations, indicating that these resources cannot simultaneously be small when changing states. 
Unlike our focus on the third law, both the TUR and TSL serve as quantitative refinements of the second law of thermodynamics.
Crucially, these relations cannot establish the third law of thermodynamics, as the inequalities \eqref{eq:TUR} and \eqref{eq:tsl} do not rule out the possibility of reaching absolute zero within finite time and cost \cite{fnt2}. Therefore, our result is conceptually and qualitatively distinct from the existing relations.

In what follows, we demonstrate several essential applications and quantum extensions of the central relation \eqref{eq:main.res}.

\subsection{More on the unattainability principle in the third law of thermodynamics}\label{sec:uptl}
We now scrutinize the connection between our result and the third law of thermodynamics. 
As discussed in the previous subsection, the trade-off relation encompasses the concept of the {\it unattainability principle} in the third law of thermodynamics. This principle is intrinsically linked to the fundamental tasks such as creating pure states and erasing information, both of which are critical in computation.
Although the impossibility of these processes has been proved in contexts involving finite-size reservoirs \cite{Franco.2013.SR,Wu.2013.SR}, it remains crucial to determine the extent to which one can approach the desired states with limited resources. In this direction, finite bounds on the achievable temperature have been identified, defined in terms of resources such as the reservoir volume and the worst-case work injected \cite{Masanes.2017.NC,Clivaz.2019.PRL}. In the present study, we advance this direction by elucidating the relationship between the achievable error and the resources of {\it time} and {\it thermokinetic cost}.
While we investigate the classical case here, the same result can be analogously obtained in quantum scenarios, which will be discussed in the subsequent subsection.

First, we consider the general problem of creating separated distributions that have only one desired state (i.e., $\mfr{d}=\{k\}$ for some state $k$).
Such distributions are also known as {\it degenerate} (or one-point) distributions, with support only at a single state.
This class of distributions trivially includes ground states.
Within a duration of $\tau$, the objective is to drive the system to the target degenerate distribution $\ket{g}$ as close as possible, irrespective of the initial distribution.
Here, the degenerate distribution $\ket{g}$ is specified as $g_n=\delta_{nk}$.
The error can be appropriately evaluated by setting $\mfr{u}=\iset{d}\setminus\{k\}$ in Eq.~\eqref{eq:error.def}.
Evidently, the error $\varepsilon_\tau$ should vanish as the final distribution approaches the degenerate distribution [i.e., when $p_k(\tau)\to 1$].
The trade-off relation \eqref{eq:main.res} then provides a lower bound on the error in terms of time and thermokinetic cost as
\begin{equation}\label{eq:third.law.res}
	\varepsilon_\tau\ge\frac{1}{\tau\mca{C}+\varepsilon_0^{-1}}.
\end{equation}
Essentially, this inequality implies that infinite resources of time and cost are required to perfectly generate degenerate distributions.
Notably, the relation \eqref{eq:third.law.res} extends beyond merely affirming the no-go theorem regarding the creation of degenerate distributions; it quantitatively determines the extent of resources required to attain a specific level of precision.

\begin{figure}[t!]
\centering
\includegraphics[width=1.0\linewidth]{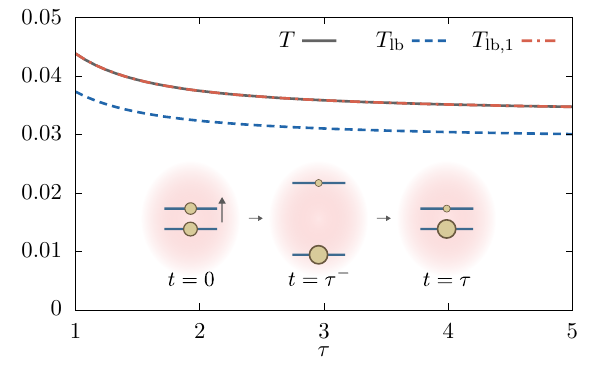}
\protect\caption{Numerical illustration of the bound \eqref{eq:temp.lb} in cooling a classical bit. The energy gap is modulated as $\Delta_g(t)=0.1+2.9t/\tau$, with transition rates satisfying $w_{12}(t)+w_{21}(t)=5$. The duration $\tau$ is varied, while the bath inverse temperature is fixed at $\beta=1$. The effective temperature $T$, shown as the solid line, is consistently bounded from below by both $T_{\rm lb}$ and $T_{\rm{lb},1}$, represented by the dashed and dash-dotted lines, respectively.}\label{fig:TempRes}
\end{figure}

Next, we address the problem of cooling to the ground state and show that the relation \eqref{eq:main.res} can be reformulated as a lower bound on the attainable temperature in terms of time and cost.
Note that while temperature is well defined for systems in thermal equilibrium, challenges present when moving onto nonequilibrium systems due to the breakdown of the equivalence of statistical ensembles.
In the literature, several notions of effective temperature have been proposed \cite{Vazquez.2003.RPP,Bartosik.2023.PRL}.
We herein consider the most common approach, which assigns an effective temperature to a nonequilibrium state using temperature of a Gibbs thermal state with the same energy.
That is, the effective temperature $T$ associated with the nonequilibrium distribution $\ket{p_\tau}$ is determined such that the following equality is fulfilled:
\begin{equation}\label{eq:eff.temp.def}
	\braket{E}{p_\tau}=\braket{E}{\pi(T,E)}.
\end{equation}
Here, $\ket{E}=[E_1,\dots,E_d]^\top$ denotes the vector of energy levels and $\ket{\pi(T,E)}$ is the equilibrium distribution with temperature $T$ and energy levels $\ket{E}$.
Without loss of generality, we can assume that $0\le E_1=\dots=E_\kappa<E_{\kappa+1}\le\dots\le E_d$. Here, $\kappa$ is the degree of degeneracy in the ground state.
By utilizing Eq.~\eqref{eq:main.res}, we can derive that the achievable temperature is lower bounded by the resources of time and cost as (see Appendix \ref{app:proof.temp.lb} for the proof)
\begin{equation}\label{eq:temp.lb}
	T\ge\frac{\Delta_g}{\ln\qty[\frac{(d-\kappa)\Delta_b}{\kappa\Delta_g}e^{\tau\mca{C}+\varepsilon_0^{-1}}-\frac{1}{\kappa}]},
\end{equation}
where $\Delta_g\coloneqq E_{\kappa+1}-E_1$ and $\Delta_b\coloneqq E_d-E_1$ are the energy gap and energy bandwidth, respectively.
More precisely, $\Delta_g$ denotes the energy gap between the ground state and the first excited state, whereas $\Delta_b$ represents the energy difference between the lowest and highest energy levels.
This inequality not only manifests the third law of thermodynamics but also establishes a fundamental limit on the minimum temperature achievable with a finite amount of resources.
The quantum version of this result is presented in Appendix \ref{app:qtemp.lb}.

We demonstrate the tightness of the bound \eqref{eq:temp.lb} using a classical bit coupled to a thermal bath. Initially, the bit is in an equilibrium state, and its energy gap is gradually increased, inducing cooling toward the ground state. At the final time, the energy gap is instantaneously reset to its initial value. In this scenario, since $\kappa=1$, $d=2$, and $\Delta_g=\Delta_b$, the bound \eqref{eq:temp.lb} simplifies to
\begin{equation}
	T\ge\frac{\Delta_g}{\ln\qty(e^{\tau\mca{C}+\varepsilon_0^{-1}}-1)}\eqqcolon T_{\rm lb}.
\end{equation}
For comparison, we also define the lower bound $T_{\rm{lb},1}$, where the cost $\mca{C}$ is replaced by $\mca{C}_1$. In general, the hierarchy $T\ge T_{\rm{lb}, 1}\ge T_{\rm lb}$ holds. We vary the operational time and compute the effective temperature of the bit at the final time using Eq.~\eqref{eq:eff.temp.def}. As depicted in Fig.~\ref{fig:TempRes}, the effective temperature $T$ is tightly bounded from below by $T_{\rm lb}$, accurately capturing the temperature's decreasing trend as a function of the time duration. Notably, the lower bound $T_{\rm{lb},1}$ is saturated for all times.

\subsection{Limitation in preparation of separated states}\label{sec:lpss}
As demonstrated above, the trade-off relation \eqref{eq:main.res} provides an essential implication for the preparation of degenerate distributions or ground states.
Here we show that it can address a more general problem of preparing a separated state from a non-separated state (Fig.~\ref{fig:SepStateMarkov}).
In this context, the relation \eqref{eq:main.res} immediately implies the impossibility of bringing a Markovian system from a non-separated state to a separated state.
In addition to this qualitative implication, it provides a quantitative bound on the final state's separation.
To measure the degree of separation of the final distribution, we examine the probability $p_{\mfr{u}}(\tau)$, with the goal of minimizing it.
Then, the following lower bound on $p_{\mfr{u}}(\tau)$ in terms of time and cost can be derived from our central result:
\begin{equation}
	p_{\mfr{u}}(\tau)\ge p_{\mfr{u}}(0)e^{-\tau\mca{C}}.
\end{equation}

\begin{figure}[t!]
\centering
\includegraphics[width=1.0\linewidth]{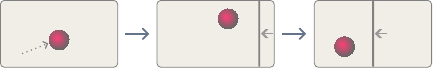}
\protect\caption{An example of preparing a separated state using a particle in a box. This scheme reduces the accessible state space during the operation, and therefore, it is not included in our setup.}\label{fig:SepStateBox}
\end{figure}

It is worth noting that a separated state can be prepared from a non-separated state without difficulty if reduction of the state space is permitted.
For example, consider a particle inside a box surrounded by a thermal reservoir at a fixed temperature (see Fig.~\ref{fig:SepStateBox} for illustration).
Initially, the particle is in a fully supported distribution, meaning it can occupy any position within the box.
Consider a situation where one aims to create a separated state wherein the particle is confined only to the left half of the box.
This target state can then be easily prepared by inserting a partition at the right end of the box and moving it to the center.
Note that during operation, the protocol reduces the state space that the particle can occupy.
That is, the space created by the partition on the right side of the box is inaccessible to the particle.
The crucial difference between this kind of protocol and our setup lies in whether the reduction of state space is allowed or not.
For Markov jump processes, the reduction of state space cannot be achieved by simply setting the transition rates to zero, as the system can be trapped at some state with a nonzero probability.

\subsection{No-go theorem of exact classical copying}\label{sec:copying}
We consider the problem of classical copying from the standpoint of stochastic thermodynamics.
The composite system has three subsystems: a source whose distribution is copied, a machine that performs copying, and a target into which distribution is copied.
Each subsystem has a finite space of discrete states, denoted by $\Omega^{x}$ for $x\in\{\msf{s},\msf{m},\msf{t}\}$.
Hereafter, the superscripts $\msf{s}$, $\msf{m}$, and $\msf{t}$ indicate the source, machine, and target subsystems, respectively.
Note that the state spaces of the source and target must be identical to guarantee accurate copying (i.e., $\Omega^{\msf{s}}\equiv\Omega^{\msf{t}}$).
The probability distribution of the composite system at time $t$ is denoted by $p_t(n,m,o)$ and evolves according to the master equation \eqref{eq:mas.eq}.
The initial distribution is of a product form
\begin{equation}
	p_0(n,m,o)=p_0^{\msf{s}}(n)p_0^{\msf{m}}(m)p_0^{\msf{t}}(o),
\end{equation}
where the initial distribution $\ket{p_0^{\msf{s}}}$ of the source can be arbitrary and those of the machine and target are some fixed distributions irrespective of $\ket{p_0^{\msf{s}}}$.
The final distribution may be non-factorizable, and we are primarily interested in the marginal distribution of the target. 
The copying process is deemed perfect only if $\ket{p_\tau^{\msf{t}}}=\ket{p_0^{\msf{s}}}$, which is determined by the following calculation of the marginal distribution,
\begin{equation}
	p_\tau^{\msf{t}}(o)=\sum_{n\in\Omega^{\msf{s}},m\in\Omega^{\msf{m}}}p_\tau(n,m,o).
\end{equation}
We emphasize that this setup is less restrictive than the problem of classical {\it cloning} \cite{Daffertshofer.2002.PRL}, where there is an additional requirement that the marginal distribution of the source remains unchanged (i.e., $\ket{p_\tau^{\msf{s}}}=\ket{p_0^{\msf{s}}}$).
For this general setup, we derive the no-go theorem for classical copying from the relation \eqref{eq:main.res} (see Fig.~\ref{fig:Copy} for illustration).

\begin{figure}[t!]
\centering
\includegraphics[width=1.0\linewidth]{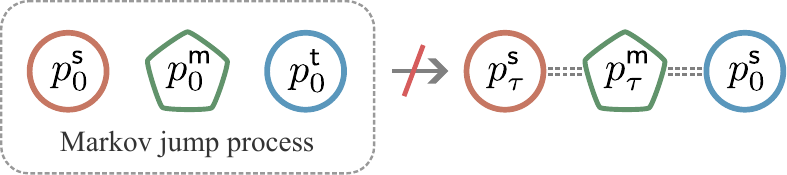}
\protect\caption{Impossibility of exact classical copying. The composite system, which includes a source, a machine, and a target, is initially in a product state. The final state may be a correlated state, wherein information from the source is expected to be replicated in the target's marginal distribution. The no-go theorem asserts that it is impossible to perfectly copy information from the source to the target within a finite amount of time and thermokinetic cost through a thermodynamic Markov jump process.}\label{fig:Copy}
\end{figure}

To this end, it suffices to show that copying cannot be perfect for an instance of the source distribution.
Let $o_*=\arg\max_{o\in\Omega^{\msf{t}}}p_0^{\msf{t}}(o)$ and consider an initial distribution $\ket{p_0^{\msf{s}}}$ with $p_0^{\msf{s}}(o_*)=0$, which we want to copy.
Then, $p_\tau^{\msf{t}}(o_*)=0$ should be fulfilled to achieve a perfect copying.
By choosing the set $\mfr{u}=\{(n,m,o_*)|n\in\Omega^{\msf{s}},m\in\Omega^{\msf{m}}\}$, the error reads
\begin{equation}
	\varepsilon_t=-[\ln p_t^{\msf{t}}(o_*)]^{-1}.
\end{equation}
Note that $\varepsilon_0>0$, and a perfect copying implies $\varepsilon_\tau=0$.
Applying the trade-off relation \eqref{eq:main.res} readily yields the no-go theorem for exact classical copying, since $\varepsilon_\tau=0$ necessitates infinite resources of time and cost.

\subsubsection{Hardness of copying}
In addition to the no-go theorem for exact copying, it is important to discuss the inherent difficulty of achieving exact copying in practice. For simplicity, we consider a copying process involving two classical bits \cite{Mulder.2025.PRR}, $\msf{X}$ and $\msf{Y}$, where the goal is to copy the value of bit $\msf{X}$ onto bit $\msf{Y}$. These two bits form a bipartite system \cite{Horowitz.2014.PRX}, with the key assumption that bit $\msf{X}$ remains unaffected by the state of bit $\msf{Y}$. The copying error can be quantified by the probability $p_e(t)$, which represents the likelihood that bits $\msf{X}$ and $\msf{Y}$ hold different values:
\begin{equation}
p_e(t) \coloneqq p_t(0,1) + p_t(1,0),
\end{equation}
where $p_t(x,y)$ denotes the joint probability of finding bits $\msf{X}$ and $\msf{Y}$ in states $x$ and $y$, respectively, at time $t$.
Perfect copying corresponds to the condition $p_e(\tau) = 0$. However, we can quantitatively demonstrate the difficulty of achieving this. To this end, we define the undesired and desired spaces as $\mfr{u}\coloneqq\{(0,1),(1,0)\}$ and $\mfr{d}\coloneqq\{(0,0),(1,1)\}$. Applying the trade-off relation \eqref{eq:main.res}, we obtain a quantitative bound that highlights the difficulty of achieving perfect copying within finite resources of time and cost:
\begin{equation}\label{eq:copy.err.lb}
p_e(\tau) \ge p_e(0) e^{-\tau\mca{C}}.
\end{equation}
This bound shows that while the error can decrease exponentially with increasing time $\tau$ and cost $\mca{C}$, it remains strictly positive as long as these resources are finite.

The difficulty of copying can also be examined from the perspective of information theory. To quantify how much information is transferred from $\msf{X}$ to $\msf{Y}$, it is convenient to use the mutual information, defined as
\begin{equation}
I_t\coloneqq\sum_{x,y}p_t(x,y)\ln\frac{p_t(x,y)}{p_t(x)p_t(y)}.
\end{equation}
Physically, $I_t$ measures how much information about $\msf{X}$ is contained in $\msf{Y}$. Since the bits $\msf{X}$ and $\msf{Y}$ are initially uncorrelated (i.e., in a product state), the initial mutual information is zero, $I_0 = 0$.
Because $I_t \le S_t(\msf{X}) = S_0(\msf{X})$, the maximum mutual information achievable is $I_{\rm max} \coloneqq S_0(\msf{X})$, where $S_t(\msf{X}) \coloneqq -\sum_x p_t(x) \ln p_t(x)$ is the Shannon entropy of bit $\msf{X}$ at time $t$. For simplicity, we assume $p_e(\tau) \le 1/2$, as we are primarily interested in processes that reduce the error probability toward zero \cite{fnt3}. Perfect copying within a finite time $\tau$ would require achieving $I_\tau = I_{\rm max}$. However, by combining Eq.~\eqref{eq:copy.err.lb} with standard information-theoretic inequalities, we reveal the inherent difficulty of this task through the following bound:
\begin{equation}\label{eq:copy.mi.lb}
I_\tau\le I_{\rm max}-(4\ln 2)h(I_{\rm max})p_e(0)e^{-\tau\mca{C}}.
\end{equation}
Here, $h\in[0,1/2]$ denotes the inverse function of $-p\ln p-(1-p)\ln(1-p)$, defined for $p\in[0,1]$. This inequality demonstrates that achieving maximum mutual information is fundamentally limited when only finite time and cost resources are available.
The detailed derivation of Eq.~\eqref{eq:copy.mi.lb} is provided in Appendix \ref{app:copy.mi.lb}.

\subsection{Generalization to quantum cases}\label{sec:quantum.gen}
Here, we show that the trade-off relation \eqref{eq:main.res} not only applies to classical cases but can also be generalized to quantum scenarios. Specifically, we consider for both Markovian dynamics with infinite-size reservoirs and non-Markovian dynamics with finite-size reservoirs. In quantum systems, preparing pure states and ground states are crucial tasks, particularly in quantum computation. To address these tasks, we derive quantum generalizations of the trade-off relation \eqref{eq:main.res}, revealing the fundamental limitations inherent in these processes.

\subsubsection{Markovian dynamics}
We consider a $d$-dimensional quantum system that is weakly coupled to infinite-size reservoirs.
The time evolution of the density matrix can be described by the the Gorini-Kossakowski-Sudarshan-Lindblad (GKSL) equation \cite{Gorini.1976.JMP,Lindblad.1976.CMP}:
\begin{equation}\label{eq:GKSL.eq}
	\dot\varrho_t=-i[H_t,\varrho_t]+\sum_kD[L_k(t)]\varrho_t,
\end{equation}
where $H_t$ is the controllable Hamiltonian, $\{L_k(t)\}$ are jump operators, and $D[L]\varrho\coloneqq L\varrho L^\dagger-\{L^\dagger L,\varrho\}/2$ is the dissipator.
To guarantee the thermodynamical consistency of the dynamics, we assume the local detailed balance \cite{Horowitz.2013.NJP,Manzano.2022.QS}.
That is, the jump operators come in pairs $(k,k')$ such that $L_k=e^{\Delta s_k/2}L_{k'}^\dagger$, where $\Delta s_k$ denotes the environmental entropy change associated with $k$th jump.
Note that it is possible that $k'=k$, which implies $\Delta s_k=0$ and the Hermiticity of jump operator $L_k$.

In analogy with the classical case, we introduce the kinetic and thermodynamic contributions to define the cost.
The thermodynamic contribution can be defined similarly to the classical case. 
According to the framework of quantum thermodynamics, entropy production is defined as the sum of entropy changes in both the system and the environment; its rate is given by
\begin{equation}
	\sigma_t=-\tr{\dot\varrho_t\ln\varrho_t}+\sum_k\Delta s_k(t)\tr{L_k(t)\varrho_tL_k(t)^\dagger}.
\end{equation}
Meanwhile, dynamical activity is quantified by the average number of quantum jumps that have occurred over time, and its rate is equal to
\begin{equation}\label{eq:qu.dynact}
	a_t=\sum_k\tr{L_k(t)\varrho_tL_k(t)^\dagger}.
\end{equation}
Let $\varrho_t=\sum_np_n(t)\dyad{n_t}$ be the spectral decomposition of the density matrix, and define the transition rates $w_{mn}^k(t)\coloneqq|\mel{m_t}{L_k(t)}{n_t}|^2$ between the eigenbasis.
Then, the rates of entropy production and dynamical activity can be decomposed as $\sigma_t=(1/2)\sum_n\sigma_n(t)$ and $a_t=(1/2)\sum_na_n(t)$ \cite{Vu.2024.PRA}, where the partial terms are given by
\begin{align}
	\sigma_n&\coloneqq\sum_{k,m}(w_{mn}^kp_n-w_{nm}^{k'}p_m)\ln\frac{w_{mn}^kp_n}{w_{nm}^{k'}p_m},\\
	a_n&\coloneqq\sum_{k,m}(w_{mn}^kp_n+w_{nm}^{k'}p_m).
\end{align}
The thermodynamic contribution, characterizing the average dissipation per single jump, can be defined as
\begin{equation}
	\overline{\sigma}_t\coloneqq\max_{n\in\iset{d}}\frac{\sigma_n(t)}{a_n(t)}.
\end{equation}

The kinetic contribution reflects how rapidly the system's state can evolve due to interactions with the environment. This aspect can be captured by the dynamical activity, defined in Eq.~\eqref{eq:qu.dynact}, which is a state-dependent quantity. To characterize the system's potential for activity independently of its state, we introduce the following state-independent kinetic contribution in terms of the jump operators:
\begin{equation}\label{eq:qu.kinetic}
	\omega_t\coloneqq\|\sum_kL_k(t)^\dagger L_k(t)\|,
\end{equation}
where $\|.\|$ denotes the operator norm.
This quantity provides an upper bound on the dynamical activity (i.e., $\omega_t \ge a_t$) \cite{fnt4}, and thus represents the maximum jump frequency of the system across all possible quantum states. Physically, $\omega_t$ is related to the coupling strength of the dissipators and can serve as a measure of the thermalization strength in a thermodynamic process. Notably, in the classical limit, $\omega_t$ reduces exactly to the maximum escape rate among states.
Using the thermodynamic and kinetic contributions provided above, the cost $\mca{C}$ is defined in accordance with Eq.~\eqref{eq:cost.def}, mirroring the approach used in the classical case.

Let $\{p_n^{\le}(t)\}$ be increasing eigenvalues of $\varrho_t$, i.e., $p_n^{\le}(t)=p_{i_n}(t)$, where $\{i_1,\dots,i_d\}$ is a permutation of $\iset{d}$ such that $p_{i_1}(t)\le\dots\le p_{i_d}(t)$.
Then, the error with respect to a set $\mfr{u}\subset\iset{d}$ can be defined as
\begin{equation}
	\varepsilon_t\coloneqq-[\ln\sum_{n\in \mfr{u}}p_n^{\le}(t)]^{-1}.
\end{equation}
For example, $\mfr{u}=\iset{d-1}$ indicates the preparation of pure states, while $\mfr{u}=\iset{k}$ corresponds to the general task of reducing the rank of the density matrix to $d-k$ or less.
Given these definitions of cost and error, we can prove that the trade-off relation \eqref{eq:main.res} remains valid for arbitrary protocols, indicating the universality of the result in both classical and quantum domains.
The proof is presented in Appendix \ref{app:proof.quantum.markov}.
It is noteworthy that, for the cooling problem, a lower bound on the achievable temperature in terms of time and cost, similar to Eq.~\eqref{eq:temp.lb}, can be analogously derived (see Appendix \ref{app:qtemp.lb} for the derivation).

\subsubsection{Non-Markovian dynamics}
Next, we consider a case where the system is coupled to a finite-dimensional reservoir.
The Hamiltonian of the composite system is expressed as $H_t=H_S(t)+H_{SR}(t)+H_R$, representing the sum of the Hamiltonians for the system, the interaction, and the reservoir, respectively.
Note that both $H_S$ and $H_{SR}$ can be time-dependent, whereas $H_R$ remains fixed.
The initial state of the composite system is of a product form $\varrho_0=\varrho_S\otimes\pi_R$, where $\pi_R=e^{-\beta H_R}/\tr e^{-\beta H_R}$ is the Gibbs thermal state at inverse temperature $\beta$.
The composite system undergoes a unitary transformation and its final state reads $\varrho_\tau=U(\varrho_S\otimes\pi_R)U^\dagger$, where $U\coloneqq\mca{T}\exp\qty(-i\int_0^\tau\dd{t}H_t)$ is the unitary operator and $\mca{T}$ denotes the time-ordering operator. 
It was shown that irreversible entropy production can be expressed as a correlation between the system and the reservoir as \cite{Esposito.2010.NJP}
\begin{equation}\label{eq:ent.prod.nonM}
	\Sigma_\tau=\mca{D}(\varrho_\tau\|\varrho_S(\tau)\otimes\pi_R),
\end{equation}
where $\varrho_S(t)=\tr_R\varrho_t$ and $\mca{D}(\varrho\|\sigma)\coloneqq\tr{\varrho(\ln\varrho-\ln\sigma)}$ is the quantum relative entropy of state $\varrho$ with respect to state $\sigma$.

We introduce the cost for general non-Markovian dynamics. It is important to highlight that non-Markovian dynamics pose substantial technical challenges. In particular, two key distinctions from the Markovian case arise: (i) the entropy production rate can become negative due to information back-flow from the reservoir to the system, and (ii) the structure of quantum jumps in non-Markovian dynamics is far more elusive compared to that in Markovian dynamics, which obstructs a clear-cut definition of the kinetic contribution. These facts indicate that the approach developed in the Markovian context is not applicable in this case. As a consequence, it is naturally expected that the conventional form of the cost should be modified appropriately. To address this issue, we alternatively consider the cost defined in terms of entropy production as follows:
\begin{equation}\label{eq:nonMarkov.cost}
	\mca{C}\coloneqq\tau^{-1}\Psi\qty(\lambda^{-1}\Sigma_\tau),
\end{equation}
where $\lambda$ represents the smallest eigenvalues of the initial state $\varrho_0$ of the composite system, and $\Psi$ is the inverse function of $\psi(x)$ defined over $[0,+\infty]$ as
\begin{equation}\label{eq:psi.func.def}
	\psi(x)\coloneqq\min_{p\in[0,1]}\qty(x+\frac{1-p}{p}\ln\frac{1-p}{1-pe^{-x}})\ge 0.
\end{equation}
Note that both $\Psi$ and $\psi$ are monotonically increasing functions, and a simple lower bound for $\psi(x)$ can be derived (see Proposition \ref{prop3} in Appendix \ref{app:useful.ines} for the proof):
\begin{equation}
	\psi(x)\ge\frac{2x}{1+e^{3/x}}.
\end{equation}
Since $\Psi(x)\approx x$ for $x\gg 1$, the cost $\mca{C}$ is proportional to $\lambda^{-1}\overline{\Sigma}$, where $\overline{\Sigma}\coloneqq\tau^{-1}\Sigma_\tau$ is the time-averaged entropy production.
It is worth noting that, although the definition of cost differs from that in the Markovian case, it essentially captures both thermodynamic and kinetic contributions. Specifically, the kinetic aspect, such as the coupling strength between the system and the reservoir, is implicitly incorporated into the cost defined in Eq.~\eqref{eq:nonMarkov.cost} through $\Sigma_\tau$. This is justified by the fact that the magnitude of $\Sigma_\tau$, which characterizes the correlation between the system and the reservoir as in Eq.~\eqref{eq:ent.prod.nonM}, is strongly dependent on the coupling strength.

We consider the task of reducing the smallest eigenvalue of an initial mixed state $\varrho_S$ of the system (or equivalently decreasing its rank to below $d$).
Previous studies showed that it is impossible to accomplish the exact reduction of the rank \cite{Wu.2013.SR,Ticozzi.2014.SR,Reeb.2014.NJP}.
It is thus appropriate to investigate the achievable error $\varepsilon_t\coloneqq-[\ln\lambda_S(t)]^{-1}$, where $\lambda_S(t)$ denotes the smallest eigenvalue of the system's density matrix $\varrho_S(t)$.
With the definition of this error and the cost \eqref{eq:nonMarkov.cost}, we can derive the trade-off relation \eqref{eq:main.res}, which can be expressed in the following form:
\begin{equation}\label{eq:main.res.nonMarkov}
	\tau\mca{C}\varepsilon_\tau=\Psi(\lambda^{-1}\Sigma_\tau)\varepsilon_\tau\ge 1-\eta,
\end{equation}
whose proof is presented in Appendix \ref{app:proof.quantum.gen}.
Crucially, this relation indicates that dissipation is unavoidable in order to decrease the smallest eigenvalue.

We discuss the relevance of this result in light of previous findings.
First, the result can be applied to establish a lower bound on heat dissipation $\Delta Q$ as follows:
\begin{equation}
	\beta\Delta Q\ge -\Delta S+\max\qty{0,\lambda\psi(\varepsilon_\tau^{-1}-\varepsilon_0^{-1})}.
\end{equation}
This bound provides a refinement for the second law of thermodynamics, suggesting that heat dissipation approaches infinity as the quantum system is cooled down to the ground state (i.e., when $\varepsilon_\tau\to 0$).
Second, it was demonstrated that the initial state $\varrho_S$ can be transformed into an arbitrary final state $\varrho_S'$ with non-decreasing rank using arbitrarily small entropy production \cite{Reeb.2014.NJP}.
However, the dimension of the reservoir must be large for such protocols, resulting in a tendency for $\lambda$ to vanish.
Consequently, this makes the relation \eqref{eq:main.res.nonMarkov} consistent with the established statement.
Last, we briefly mention a relevant result reported in Ref.~\cite{Reeb.2014.NJP}, given by
\begin{equation}
	(2\beta\|H_R\|)\varepsilon_\tau\ge 1-\eta.
\end{equation}
This inequality characterizes a trade-off between the error and the maximum energy level of the reservoir. 
Conversely, our result illustrates a trade-off between the error, time, and dissipation, indicating that these findings can be viewed as complementary relations.

\begin{figure*}[t!]
\centering
\includegraphics[width=1.0\linewidth]{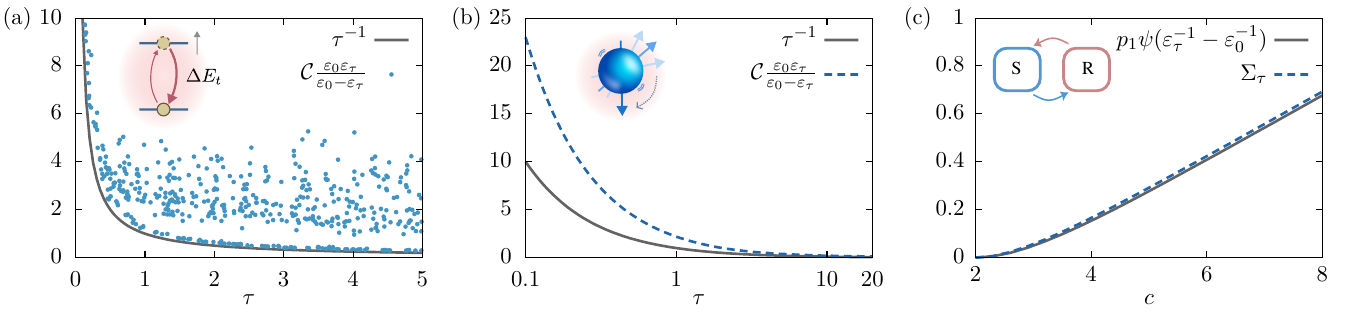}
\protect\caption{(a) Cooling the classical bit. The energy gap $\Delta E_t$ is a smooth function generated by interpolating a sequence of random numbers sampled in the range $[0.1,3]$. Each blue circle represents the product of cost and error [$\mca{C}\varepsilon_0\varepsilon_\tau/(\varepsilon_0-\varepsilon_\tau)$] as a function of time for a random configuration of the energy gap, while the solid line depicts the reciprocal of time ($\tau^{-1}$). Time $\tau$ is changed, whereas other parameters are fixed as $\ket{p_0}=[0.5,0.5]^\top$, $w_{\rm tot}=5$, and $\beta=1$. (b) Erasing information of the qubit. The product of cost and error and the reciprocal of time are depicted by the dashed and solid lines, respectively. Time $\tau$ is varied, while other parameters are fixed as $\varrho_0=\mds{1}/2$, $\alpha=0.1$, $\beta=1$, $E_0=0.1$, and $E_\tau=5$. (c) Reducing the smallest eigenvalue of quantum states using a swap protocol. The entropy production $\Sigma_\tau$ and the lower bound $p_1\psi(\varepsilon_\tau^{-1}-\varepsilon_0^{-1})$ are represented by the dashed and solid lines, respectively. The parameter $c$ is varied, whereas others are set to $p_1=e^{-2}$ and $p_d=0.5$.}\label{fig:NumRes}
\end{figure*}

\section{Demonstrations}\label{sec:demon}
In this section, we numerically illustrate our results in both classical and quantum systems.
\subsection{Classical bit}
We exemplify the trade-off relation \eqref{eq:main.res} in a cooling problem of a classical bit.
The bit is modeled as a two-level system, which is coupled to a thermal reservoir at inverse temperature $\beta$.
The time-dependent transition rates between the two levels satisfy the detailed balance condition, that is,
\begin{equation}
	\frac{w_{12}(t)}{w_{21}(t)}=e^{\beta\Delta E_t},
\end{equation}
where $\Delta E_t$ denotes the energy gap.
Cooling the bit to the ground state can be achieved by appropriately changing the energy gap, thereby enhancing the probability of finding the bit in the lowest-energy level. 
This kind of protocol can also function as an information eraser, which consistently drives the system to the ground state, irrespective of the initial state.
The error of cooling can be evaluated by examining the probability of the excited state, or equivalently by setting $\mfr{u}=\{2\}$.

We fix the total transition rate $w_{21}(t)+w_{12}(t)=w_{\rm tot}$ for all times and designate the energy gap $\Delta E_t$ as the sole control parameter.
The transition rates can be explicitly written as
\begin{equation}
	w_{12}(t)=w_{\rm tot}\frac{e^{\beta\Delta E_t}}{1+e^{\beta\Delta E_t}},~w_{21}(t)=w_{\rm tot}\frac{1}{1+e^{\beta\Delta E_t}}.
\end{equation}
The trade-off relation \eqref{eq:main.res} can be alternatively expressed as a lower bound for the product of cost and error in terms of time,
\begin{equation}\label{eq:main.res.ex1}
	\mca{C}\frac{\varepsilon_0\varepsilon_\tau}{\varepsilon_0-\varepsilon_\tau}\ge\frac{1}{\tau}.
\end{equation}
For each operational time $\tau$, we randomly generate $\Delta E_t$ as a function of time over the interval $[0,\tau]$, numerically calculate the relevant quantities, and verify the bound \eqref{eq:main.res.ex1} as shown in Fig.~\ref{fig:NumRes}(a).
It can be observed that the bound is valid for all configurations, and notably, it can be saturated when the energy gap $\Delta E_t$ is a monotonically increasing function of time.

\subsection{Qubit}
Next, we consider the problem of erasing information on a qubit.
The qubit is modeled as a spin-$1/2$ particle, which is weakly coupled to a thermal reservoir consisting of bosonic harmonic oscillators \cite{Leggett.1987.RMP}.
The time evolution of the density matrix of the qubit is described by the GKSL equation \eqref{eq:GKSL.eq} with the Hamiltonian and jump operators given by
\begin{align}
H_t&=\frac{E_t}{2}\qty[\cos(\theta_t)\sigma_z+\sin(\theta_t)\sigma_x],\\
L_1(t)&=\sqrt{\alpha E_t(N_t+1)}\dyad{0_t}{1_t},\\
L_2(t)&=\sqrt{\alpha E_tN_t}\dyad{1_t}{0_t}.
\end{align}
Here, $E_t$ and $\theta_t$ are controllable parameters, $\sigma_{x,y,z}$ denote the Pauli matrices, $\alpha$ is the coupling strength, $N_t\coloneqq 1/(e^{\beta E_t}-1)$ is the Planck distribution, and $\{\ket{0_t},\ket{1_t}\}$ denote the instantaneous energy eigenstates of the Hamiltonian.
Physically, $E_t$ is the energy gap between the energy eigenstates, whereas $\theta_t$ characterizes the strength of coherent tunneling.

While many feasible protocols exist for information erasure, we employ the following simple protocol \cite{Miller.2020.PRL.QLP} for clarity and ease of explanation: 
\begin{equation}
	E_t=E_0+(E_\tau-E_0)\sin\qty(\frac{\pi t}{2\tau})^2,~\theta_t=\pi\qty(\frac{t}{\tau}-1).
\end{equation}
For this protocol, the energy gap monotonically increases from $E_0$ to $E_\tau$, and the coherent parameter $\theta_t$ vanishes at the final time.
Therefore, the qubit is expected to be reset to the ground state of $H_\tau\propto\sigma_z$ upon completing the erasure process.
We set the initial density matrix to the maximally mixed state and vary the operational time $\tau$.
For each instance of time, we calculate the relevant quantities and verify the bound \eqref{eq:main.res.ex1} in Fig.~\ref{fig:NumRes}(b). 
As can be seen, the product of cost and error is consistently bounded from below by the reciprocal of time, thus validating the derived relation.

\subsection{Finite-size reservoir}
Last, we exemplify the relation \eqref{eq:main.res.nonMarkov} in a quantum process that reduces the smallest eigenvalue of quantum states.
We consider a $d$-dimensional system prepared in an initial mixed state $\varrho_S$ with full rank.
Let $\varrho_S=\sum_{n=1}^dp_n\dyad{n}$ be its spectral decomposition (i.e., $0<p_1\le\dots\le p_d$).
The smallest eigenvalue of the initial state $\varrho_S$ can be reduced by performing a swap between the states of the system and the reservoir.
The reservoir is another $d$-dimensional system initialized in a state 
\begin{equation}
	\pi_R=e^{-c}\dyad{1}+\sum_{n=2}^{d-1}p_n\dyad{n}+(p_d+p_1-e^{-c})\dyad{d},
\end{equation}
where $c\ge-\ln p_1$ is a positive constant.
The reservoir Hamiltonian is given by $H_R=-\ln\pi_R$, meaning that the reservoir is in a Gibbs thermal state at inverse temperature $\beta=1$.
In order to perform a swap, the unitary operator is designed such that $U(\varrho_S\otimes\pi_R)U^\dagger=\pi_R\otimes\varrho_S$, resulting in the final state of the system being $\varrho_S'=\pi_R$.
Upon completion, the smallest eigenvalue of the system state is reduced from $p_1$ to $e^{-c}\,(\le p_1)$.
Entropy production of this thermodynamic process can be calculated as
\begin{equation}
	\Sigma_\tau=\mca{D}(\varrho_S\|\pi_R)=p_1\ln\frac{p_1}{e^{-c}}+p_d\ln\frac{p_d}{p_d+p_1-e^{-c}}.
\end{equation}
In this case, the relation \eqref{eq:main.res.nonMarkov} can be quantitatively improved, resulting in the following bound of dissipation in terms of the errors:
\begin{equation}\label{eq:qua.res.alt}
	\Sigma_\tau\ge p_1\psi(\varepsilon_\tau^{-1}-\varepsilon_0^{-1}).
\end{equation}

We fix $p_1$ and $p_d$, and vary $c$, which determines the smallest eigenvalue of the final state.
All quantities appearing in the bound \eqref{eq:qua.res.alt} are calculated and plotted in Fig.~\ref{fig:NumRes}(c).
As can be observed, dissipation increases as the error decreases, and the bound is numerically validated.

\section{Summary and outlook}\label{sec:sum.out}
In this paper, we resolved a fundamental problem in thermodynamics by establishing a rigorous quantitative relationship among the foundational quantities of time, cost, and error. Specifically, we uncovered a three-way trade-off relation between these quantities for thermodynamic processes aimed at minimizing the observation probability to zero. Such processes are crucial in various scientific fields, including physics and biology. For example, in bit erasure or cooling systems to the ground state, the achievable error limits are inevitably tied to the resource expenditures of time and cost. The obtained relation indicates that reducing the error to zero necessitates infinite resources of time and cost. This finding deepens our understanding of thermodynamic limits and addresses several essential problems in nonequilibrium thermodynamics. These include the third law of thermodynamics in the form of the unattainability principle, the limitation on preparing separated states, the no-go theorem of exact copying, and the thermodynamic cost associated with accurate information erasure. Our study not only advances the theoretical framework of thermodynamics but also sets the stage for further investigations into the practical limitations and capabilities of both natural and engineered systems operating near these fundamental limits.

We discuss several future directions for further investigation. In this work, we identified the thermokinetic cost relevant to the task of generating separated states, thereby providing a unified framework that encompasses the third law of thermodynamics. Notably, this cost inherently includes thermodynamic contributions, which, although theoretically estimable, may be difficult to measure reliably in practical experimental settings. A promising direction for future research is to explore whether a more operationally accessible formulation of the third law can be established---one that preserves theoretical integrity while aligning more closely with measurable physical quantities.

Since this study primarily focused on discrete-variable systems, it is natural to consider the generalization of our results to continuous-variable systems, such as those governed by overdamped or underdamped Langevin dynamics. However, a straightforward extension---particularly through coarse-graining the continuous space into a finite number of discrete states---is generally not viable. This is because the dynamical activity in continuous systems generally diverges, even after coarse-graining, which renders the trade-off relation \eqref{eq:main.res} trivial. In the context of cooling quantum systems to their ground states, efficient protocols may involve continuous measurement and feedback control to reduce the average energy. In such cases, the recently developed theoretical framework \cite{Andersson.2022.PRL}, where the joint dynamics of the system and the controller is governed by a quantum Fokker-Planck master equation, can serve as a good starting point.

While we focused on the observation probability in this study, extending our results to other types of observables would be interesting. In nonequilibrium systems, observables play a key role in defining fundamental quantities. For example, temperature cannot be unambiguously defined as in thermal equilibrium, and its definition generally involves an energetic observable. For noisy molecular oscillators, the degree of synchronization is generally quantified by the order parameter, which is significantly constrained by the energy dissipation \cite{Zhang.2019.NP}. Generalizing our findings to other observables may uncover novel limitations of nonequilibrium systems that are not revealed by existing speed limits. Additionally, this approach could facilitate the direct study of continuous-variable systems without the need for coarse-graining.

Due to the analogy between the thermodynamic frameworks of Markov jump processes and chemical reaction networks \cite{Rao.2016.PRX}, it is possible to extend our results to the latter case. Note that the time evolution of chemical reactants is given by deterministic rate equations, and the concepts of adiabatic and nonadiabatic entropy production have been well established for complex-balanced networks. While the first and second law of thermodynamics have been formulated for both closed and open reaction systems, the third law of thermodynamics remains elusive. Our study can provide a resolution for this problem, suggesting that the third law can be discussed in a similar manner.

We also anticipate that the same trade-off relation applies to other information-processing tasks, such as measurement. In this type of process, the controller aims to acquire as much information about the target system as possible within a finite duration. However, achieving perfect measurement outcomes is prohibited. A recent study on quantum measurement has demonstrated that there exists a trade-off relation between time and error of measurement outcome \cite{Nakajima.2024.arxiv}. Therefore, it would be natural to elucidate a three-way trade-off relation among time, cost, and error for these processes \cite{Guryanova.2020.Q,Mohammady.2023.PRA}.

Finally, the existence of the trade-off relation for non-Markovian dynamics with {\it infinite}-dimensional reservoirs remains an open question. Our trade-off relation \eqref{eq:main.res.nonMarkov}, which was derived for finite-size reservoirs, includes a dimensional factor of the reservoirs. This factor impedes the achievement of a meaningful bound in the infinite-dimensional limit. Resolving this issue would provide a comprehensive understanding of the third law in the form of the unattainability principle for quantum dynamics \cite{Cleuren.2012.PRL,Kolar.2012.PRL,Kosloff.2013.E}.

\begin{acknowledgments}
This work was partially supported by the Future Development Funding Program of Kyoto University Research Coordination Alliance. {T.V.V.} was supported by JSPS KAKENHI Grant No.~JP23K13032. {K.S.} was supported by JSPS KAKENHI Grant No.~JP23K25796.
\end{acknowledgments}

\appendix

\section{Proof of the tradeoff relation \eqref{eq:main.res}}\label{app:proof.main.res}
Here we provide a detailed derivation of Eq.~\eqref{eq:main.res}.
Define the coarse-grained probability $p_{\mfr{d}}(t)\coloneqq\sum_{n\in \mfr{d}}p_n(t)$, which satisfies $p_{\mfr{u}}(t)+p_{\mfr{d}}(t)=1$.
The time evolution of the probability $p_{\mfr{u}}(t)$ can be described as
\begin{align}
	\dot p_{\mfr{u}}(t)&=\sum_{n\in \mfr{u}}\sum_{m(\neq n)}[w_{nm}(t)p_m(t)-w_{mn}(t)p_n(t)]\notag\\
	&=\sum_{n\in \mfr{u},m\in \mfr{d}}[w_{nm}(t)p_m(t)-w_{mn}(t)p_n(t)]\notag\\
	&=w_{\mfr{u}\mfr{d}}(t)p_{\mfr{d}}(t)-w_{\mfr{d}\mfr{u}}(t)p_{\mfr{u}}(t),
\end{align}
where $w_{\mfr{u}\mfr{d}}(t)$ and $w_{\mfr{d}\mfr{u}}(t)$ are coarse-grained transition rates, defined as
\begin{align}
	w_{\mfr{u}\mfr{d}}(t)&\coloneqq\frac{\sum_{n\in \mfr{u},m\in \mfr{d}}w_{nm}(t)p_m(t)}{\sum_{m\in \mfr{d}}p_m(t)},\\
	w_{\mfr{d}\mfr{u}}(t)&\coloneqq\frac{\sum_{n\in \mfr{u},m\in \mfr{d}}w_{mn}(t)p_n(t)}{\sum_{n\in \mfr{u}}p_n(t)}.
\end{align}
Note that $a_t=w_{\mfr{u}\mfr{d}}(t)p_{\mfr{d}}(t)+w_{\mfr{d}\mfr{u}}(t)p_{\mfr{u}}(t)$.
Furthermore, we can prove that the transition rate $w_{\mfr{d}\mfr{u}}(t)$ is always upper bounded by the maximum escape rate as follows:
\begin{align}
	w_{\mfr{d}\mfr{u}}(t)&=\frac{\sum_{n\in \mfr{u},m\in \mfr{d}}w_{mn}(t)p_n(t)}{\sum_{n\in \mfr{u}}p_n(t)}\notag\\
	&=\frac{\sum_{n\in \mfr{u}}p_n(t)\sum_{m\in \mfr{d}}w_{mn}(t)}{\sum_{n\in \mfr{u}}p_n(t)}\notag\\
	&\le\max_{n\in \mfr{u}}\sum_{m\in \mfr{d}}w_{mn}(t)\notag\\
	&=\omega_t.\label{eq:wJI.up}
\end{align}
Here, we apply the inequality $\sum_nx_n/\sum_ny_n\le\max_nx_n/y_n$ for any nonnegative numbers $\{x_n\}$ and $\{y_n\}$ to obtain the third line, whereas the last line is readily obtained from the definition of $\omega_t$.
The notation of time $t$ is omitted hereafter for the sake of notational simplicity.
For convenience, we define the following nonnegative function over $(0,+\infty)\times(0,+\infty)$:
\begin{equation}
	F(x,y)\coloneqq (x-y)\ln\frac{x}{y}.
\end{equation}
It can be verified that $F(x,y)$ is a symmetric convex function.
We can prove 
\begin{equation}
	\sigma_t\ge F(w_{\mfr{u}\mfr{d}}p_{\mfr{d}},w_{\mfr{d}\mfr{u}}p_{\mfr{u}})
\end{equation}
by utilizing the convexity of function $F$ as follows:
\begin{align}
	\sigma_t&=\sum_{m>n}F\qty(w_{nm}p_m,w_{mn}p_n)\notag\\
	&\ge\sum_{n\in \mfr{u},m\in \mfr{d}}F\qty(w_{nm}p_m,w_{mn}p_n)\notag\\
	&\ge|\mfr{u}||\mfr{d}|F\qty(\frac{\sum_{n\in \mfr{u},m\in \mfr{d}}w_{nm}p_m}{|\mfr{u}||\mfr{d}|},\frac{\sum_{n\in \mfr{u},m\in \mfr{d}}w_{mn}p_n}{|\mfr{u}||\mfr{d}|})\notag\\
	&=F(w_{\mfr{u}\mfr{d}}p_{\mfr{d}},w_{\mfr{d}\mfr{u}}p_{\mfr{u}}).
\end{align}
For any nonnegative numbers $x$ and $y$, it is evident that
\begin{equation}\label{eq:rel.Phi.func}
	\frac{F(x,y)}{x+y}=\phi\qty(\frac{x}{y}),
\end{equation}
where $\phi(x)=(x+1)^{-1}(x-1)\ln x \ge 0$. 
Since the equation $z=\phi(x)$ has solutions $x_1\le 1\le x_2$ for any $z>0$, we can define the inverse function $\phi^{-1}$ that takes values in $(0,1]$.
By this definition, we obtain from Eq.~\eqref{eq:rel.Phi.func} that
\begin{equation}\label{eq:inv.f.ine}
	x\ge\phi^{-1}\qty[\phi(x)].
\end{equation}
Noting that $\phi^{-1}(x)$ is a decreasing convex function and using Eq.~\eqref{eq:inv.f.ine}, we can calculate as follows:
\begin{align}
	-\frac{\dot p_{\mfr{u}}}{p_{\mfr{u}}}&=\frac{1}{p_{\mfr{u}}}\qty(w_{\mfr{d}\mfr{u}}p_{\mfr{u}}-w_{\mfr{u}\mfr{d}}p_{\mfr{d}})\notag\\
	&=w_{\mfr{d}\mfr{u}}\qty(1-\frac{w_{\mfr{u}\mfr{d}}p_{\mfr{d}}}{w_{\mfr{d}\mfr{u}}p_{\mfr{u}}})\notag\\
	&\le w_{\mfr{d}\mfr{u}}\qty{1-\phi^{-1}\qty[\phi\qty(\frac{w_{\mfr{u}\mfr{d}}p_{\mfr{d}}}{w_{\mfr{d}\mfr{u}}p_{\mfr{u}}})]}\notag\\
	&\le\omega_t\qty{1-\phi^{-1}\qty[\frac{F(w_{\mfr{u}\mfr{d}}p_{\mfr{d}},w_{\mfr{d}\mfr{u}}p_{\mfr{u}})}{w_{\mfr{u}\mfr{d}}p_{\mfr{d}}+w_{\mfr{d}\mfr{u}}p_{\mfr{u}}}]}\notag\\
	&\le\omega_t\qty[1-\phi^{-1}\qty(\frac{\sigma_t}{a_t})]=\omega_t\Phi\qty(\overline{\sigma}_t).\label{eq:proof.tmp1}
\end{align}
Taking the time integration of Eq.~\eqref{eq:proof.tmp1} from $t=0$ to $t=\tau$ leads to the following inequality:
\begin{align}
	\ln p_{\mfr{u}}(0)-\ln p_{\mfr{u}}(\tau)\le\int_0^\tau\dd{t}\omega_t\Phi\qty(\overline{\sigma}_t)=\tau\mca{C}_1,\label{eq:proof.tmp2}
\end{align}
which can be alternatively expressed as $\tau\mca{C}_1\ge\varepsilon_\tau^{-1}-\varepsilon_0^{-1}$.
Multiplying both sides of this inequality by $\varepsilon_\tau$ immediately gives
\begin{equation}\label{eq:proof.tmp5}
	\tau\mca{C}_1\varepsilon_\tau\ge 1-\eta.
\end{equation}
Since $\mca{C}_1\le\mca{C}$ due to $\omega_t\le\omega$ and the concavity of the function $\Phi(x)$, combining this with Eq.~\eqref{eq:proof.tmp5} directly leads to the trade-off relation \eqref{eq:main.res}.

\section{Generalization of the trade-off relation in the presence of unidirectional transitions}\label{app:proof.main.res.uni}
We consider a general dynamics that involves unidirectional transitions.
The transition rates can be generally expressed as
\begin{equation}
	w_{mn}(t)=\sum_{\gamma}w_{mn}^{(\gamma)}(t)+\sum_{\xi}w_{mn}^{(\xi)}(t),
\end{equation}
where $w_{mn}^{(\gamma)}(t)$ and $w_{mn}^{(\xi)}(t)$ denote the bidirectional and unidirectional transition rates, respectively.
That is for any $m\neq n$, $w_{mn}^{(\gamma)}>0$ whenever $w_{nm}^{(\gamma)}>0$, and at least one of the two transition rates, $w_{mn}^{(\xi)}$ and $w_{nm}^{(\xi)}$, is zero.
We define the entropy production rate, the dynamical activity rate, and the maximum escape rate as follows:
\begin{align}
	\sigma_t^b&\coloneqq\sum_{\gamma,m>n}[w_{mn}^{(\gamma)}p_n-w_{nm}^{(\gamma)}p_m]\ln\frac{w_{mn}^{(\gamma)}p_n}{w_{nm}^{(\gamma)}p_m},\\
	a_t^{b}&\coloneqq\sum_{\gamma,n\in \mfr{u},m\in \mfr{d}}[w_{mn}^{(\gamma)}p_n+w_{nm}^{(\gamma)}p_m],\\
	\omega_t^{b}&\coloneqq\max_{n\in\mfr{u}}\sum_{\gamma,m\in\mfr{d}}w_{mn}^{(\gamma)},\\
	\omega_t^{u}&\coloneqq\max_{n\in\mfr{u}}\sum_{\xi,m\in\mfr{d}}w_{mn}^{(\xi)}.
\end{align}
Hereafter, the superscripts $b$ and $u$ indicate that the corresponding quantity is calculated using only the bidirectional and unidirectional parts, respectively.
Similarly, we define the coarse-grained transition rates in terms of bidirectional contributions as
\begin{align}
	w_{\mfr{u}\mfr{d}}&\coloneqq\frac{\sum_{\gamma,n\in \mfr{u},m\in \mfr{d}}w_{nm}^{(\gamma)}p_m}{\sum_{m\in \mfr{d}}p_m},\\
	w_{\mfr{d}\mfr{u}}&\coloneqq\frac{\sum_{\gamma,n\in \mfr{u},m\in \mfr{d}}w_{mn}^{(\gamma)}p_n}{\sum_{n\in \mfr{u}}p_n}.
\end{align}
Note that $w_{\mfr{d}\mfr{u}}\le\omega_t^{b}$, $a_t^b=w_{\mfr{u}\mfr{d}}p_{\mfr{d}}+w_{\mfr{d}\mfr{u}}p_{\mfr{u}}$, and $\sigma_t^b\ge F(w_{\mfr{u}\mfr{d}}p_{\mfr{d}},w_{\mfr{d}\mfr{u}}p_{\mfr{u}})$.
Then, similarly to Eq.~\eqref{eq:proof.tmp1}, the time evolution of the coarse-grained probability $p_{\mfr{u}}(t)$ can be bounded as follows:
\begin{align}
	-\frac{\dot p_{\mfr{u}}}{p_{\mfr{u}}}&=\frac{1}{p_{\mfr{u}}}\qty(w_{\mfr{d}\mfr{u}}p_{\mfr{u}}-w_{\mfr{u}\mfr{d}}p_{\mfr{d}}+\sum_{\xi,n\in \mfr{u},m\in \mfr{d}}[w_{mn}^{(\xi)}p_n-w_{nm}^{(\xi)}p_m])\notag\\
	&\le \omega_t^b\Phi\qty(\overline{\sigma}_t^b)+\omega_t^{u}.
\end{align}
Consequently, by following the same procedure as in Appendix \ref{app:proof.main.res}, we obtain the desired relation with the generalized cost defined in Eq.~\eqref{eq:main.res.uni}.

\section{The trade-off relation in terms of a nonadiabatic cost}\label{app:nonad.tradeoff}
Here we derive the trade-off relation using the nonadiabatic entropy production.
To this end, we introduce the kinetic and thermodynamic contributions.
First, we define the maximum escape rate $\widetilde{\omega}_t\coloneqq\max_{n\in\mfr{u}}\sum_{m(\neq n)}w_{mn}$, which corresponds to the kinetic contribution.
Applying the inequality $\sum_nx_n/\sum_ny_n\le\max_nx_n/y_n$ for any nonnegative numbers $\{x_n\}$ and $\{y_n\}$, we can prove that $\widetilde{\omega}_t$ is lower bounded as
\begin{equation}
	\widetilde{\omega}_t\ge\frac{\sum_{n\in\mfr{u}}\sum_{m(\neq n)}w_{mn}p_n}{\sum_{n\in\mfr{u}}p_n}.
\end{equation}
Next, we introduce the thermodynamic contribution.
Let $\ket{p_t^{\rm ss}}$ be the instantaneous steady-state distribution of the dynamics at time $t$, that is, $\msf{W}_t\ket{p_t^{\rm ss}}=0$.
We consider the dual dynamics with the transition rates given by
\begin{equation}
	\widetilde{w}_{mn}\coloneqq\frac{w_{nm}p_m^{\rm ss}}{p_n^{\rm ss}}.
\end{equation}
It is evident that
\begin{equation}
	\sum_{m(\neq n)}\widetilde{w}_{mn}=-\widetilde{w}_{nn}=-w_{nn}=\sum_{m(\neq n)}w_{mn}.
\end{equation}
For convenience, we define the frequency of incoming and outgoing jumps associated with state $n$ as 
\begin{align}
	a_{n\leftarrow}&\coloneqq \sum_{m(\neq n)}w_{nm}p_m,\\
	\widetilde{a}_{n\rightarrow}&\coloneqq \sum_{m(\neq n)}\widetilde{w}_{mn}p_n.
\end{align}
Then, the dynamical activity rates in both the original and dual dynamics are the same,
\begin{equation}
	\sum_na_{n\leftarrow}=\sum_{n}\widetilde{a}_{n\rightarrow}.
\end{equation}
The rate of nonadiabatic entropy production, also known as Hatano-Sasa entropy production \cite{Hatano.2001.PRL}, can be defined as \cite{Esposito.2010.PRL}
\begin{equation}
	\sigma_t^{\rm na}\coloneqq\sum_{m\neq n}w_{nm}p_m\ln\frac{w_{nm}p_m}{\widetilde{w}_{mn}p_n}\ge 0.
\end{equation}
Note that $\sigma_t^{\rm na}$ vanishes when the system is in the stationary state.
In general, the nonadiabatic entropy production rate can be decomposed into nonnegative components as
\begin{equation}\label{eq:nonad.entprod.decom}
	\sigma_t^{\rm na}=\sum_n\sigma_n^{\rm na},
\end{equation}
where $\sigma_n^{\rm na}\ge 0$ is the rate associated with state $n$ and given by
\begin{equation}
	\sigma_n^{\rm na}\coloneqq\sum_{m(\neq n)}\qty(w_{nm}p_m\ln\frac{w_{nm}p_m}{\widetilde{w}_{mn}p_n}-w_{nm}p_m+\widetilde{w}_{mn}p_n).
\end{equation}
The dynamical activity associated with states belonging to the set $\mfr{u}$ is defined as
\begin{equation}
	a_t^{\rm na}\coloneqq\sum_{n\in\mfr{u}}\sum_{m(\neq n)}(w_{nm}p_m+w_{mn}p_n).
\end{equation}
Note that $a_t^{\rm na}=\sum_{n\in\mfr{u}}(a_{n\leftarrow}+\widetilde{a}_{n\rightarrow})$.
The definition introduced in Eq.~\eqref{eq:ther.cost.def} motivates us to define the nonadiabatic thermodynamic contribution as
\begin{equation}
	\overline{\sigma}_t^{\rm na}\coloneqq\frac{\sigma_t^{\rm na}}{a_t^{\rm na}}.
\end{equation}
Evidently, $\overline{\sigma}_t^{\rm na}=0$ if and only if $\sigma_t^{\rm na}=0$.

Define the function $\widetilde{\phi}(x)\coloneqq(x+1)^{-1}(x\ln x-x+1)$ over $(0,+\infty)$.
Since the equation $z=\widetilde{\phi}(x)$ has a solution $0\le x\le 1$ for any $0\le z\le 1$, we can define an inverse function $\widetilde{\phi}^{-1}\in[0,1]$ as $\widetilde{\phi}^{-1}(z)\coloneqq x$ for $z\in[0,1]$ and $\widetilde{\phi}^{-1}(z)=0$ for $z>1$.
Then, we have the following inequality from its definition:
\begin{equation}\label{eq:nonad.tmp1}
	x\ge\widetilde{\phi}^{-1}\qty[\widetilde{\phi}(x)].
\end{equation}
In addition, by exploiting the convexity of function $x\ln(x/y)-x+y$ over $(0,+\infty)\times(0,+\infty)$, we obtain the following inequality:
\begin{align}
	\sigma_t^{\rm na}&\ge\sum_{n\in\mfr{u}}\sigma_n^{\rm na}\notag\\
	&\ge\sum_{n\in\mfr{u}}\qty(a_{n\leftarrow}\ln\frac{a_{n\leftarrow}}{\widetilde{a}_{n\rightarrow}}-a_{n\leftarrow}+\widetilde{a}_{n\rightarrow})\notag\\
	&\ge \qty(\sum_{n\in\mfr{u}}a_{n\leftarrow})\ln\frac{\sum_{n\in\mfr{u}}a_{n\leftarrow}}{\sum_{n\in\mfr{u}}\widetilde{a}_{n\rightarrow}}-\sum_{n\in\mfr{u}}a_{n\leftarrow}+\sum_{n\in\mfr{u}}\widetilde{a}_{n\rightarrow}\notag\\
	&=a_t^{\rm na}\widetilde{\phi}\qty(\frac{\sum_{n\in\mfr{u}}a_{n\leftarrow}}{\sum_{n\in\mfr{u}}\widetilde{a}_{n\rightarrow}}),
\end{align}
which can be alternatively written as
\begin{equation}\label{eq:nonad.tmp2}
	\frac{\sigma_t^{\rm na}}{a_t^{\rm na}}\ge\widetilde{\phi}\qty(\frac{\sum_{n\in\mfr{u}}a_{n\leftarrow}}{\sum_{n\in\mfr{u}}\widetilde{a}_{n\rightarrow}}).
\end{equation}
By combining Eqs.~\eqref{eq:nonad.tmp1} and \eqref{eq:nonad.tmp2} and noticing that $\widetilde{\phi}^{-1}$ is a decreasing function, we get
\begin{equation}
	\widetilde{\phi}^{-1}\qty(\frac{\sigma_t^{\rm na}}{a_t^{\rm na}})\le\frac{\sum_{n\in\mfr{u}}a_{n\leftarrow}}{\sum_{n\in\mfr{u}}\widetilde{a}_{n\rightarrow}}.
\end{equation}
Applying this inequality, we can calculate as follows:
\begin{align}
	-\frac{\dot p_{\mfr{u}}}{p_{\mfr{u}}}&=\frac{1}{p_{\mfr{u}}}\qty(\sum_{n\in\mfr{u}}\sum_{m(\neq n)}w_{mn}p_n-\sum_{n\in\mfr{u}}\sum_{m(\neq n)}w_{nm}p_m)\notag\\
	&=\frac{1}{p_{\mfr{u}}}\qty(\sum_{n\in\mfr{u}}\sum_{m(\neq n)}\widetilde{w}_{mn}p_n-\sum_{n\in\mfr{u}}\sum_{m(\neq n)}w_{nm}p_m)\notag\\
	&=\qty(\frac{\sum_{n\in\mfr{u}}\sum_{m(\neq n)}w_{mn}p_n}{\sum_{n\in\mfr{u}}p_n})\qty(1-\frac{\sum_{n\in\mfr{u}}a_{n\leftarrow}}{\sum_{n\in\mfr{u}}\widetilde{a}_{n\rightarrow}})\notag\\
	&\le\qty(\frac{\sum_{n\in\mfr{u}}\sum_{m(\neq n)}w_{mn}p_n}{\sum_{n\in\mfr{u}}p_n})\qty[1-\widetilde{\phi}^{-1}\qty(\frac{\sigma_t^{\rm na}}{a_t^{\rm na}})]\notag\\
	&\le\widetilde{\omega}_t\widetilde{\Phi}(\overline{\sigma}_t^{\rm na}),\label{eq:proof.na.tmp0}
\end{align}
where we define $\widetilde{\Phi}\coloneqq 1-\widetilde{\phi}^{-1}$.
Taking the time integration of Eq.~\eqref{eq:proof.na.tmp0} from $t=0$ to $t=\tau$ leads to the following inequality:
\begin{align}
	\ln p_{\mfr{u}}(0)-\ln p_{\mfr{u}}(\tau)\le\int_0^\tau\dd{t}\widetilde{\omega}_t\widetilde{\Phi}\qty(\overline\sigma_t^{\rm na})=\tau\mca{C}_1^{\rm na},\label{eq:proof.na.tmp2}
\end{align}
where the nonadiabatic cost is defined as
\begin{equation}
	\mca{C}_1^{\rm na}\coloneqq\ev{\widetilde{\omega}_t\widetilde{\Phi}(\overline{\sigma}_t^{\rm na})}_\tau.
\end{equation}
Noting that $\widetilde{\Phi}$ is a concave function, the cost $\mca{C}_1^{\rm na}$ can be upper bounded as
\begin{equation}
	\mca{C}_1^{\rm na}\le\widetilde{\omega}\widetilde{\Phi}(\overline{\sigma}^{\rm na})\eqqcolon\mca{C}^{\rm na},
\end{equation} 
where $\widetilde{\omega}\coloneqq\max_{0\le t\le\tau}\widetilde{\omega}_t$ and $\overline{\sigma}^{\rm na}\coloneqq\ev{\overline{\sigma}_t^{\rm na}}_\tau$.
Consequently, the inequality \eqref{eq:proof.na.tmp2} yields the following hierarchy of trade-off relations:
\begin{equation}
	\tau\mca{C}^{\rm na}\varepsilon_\tau\ge\tau\mca{C}_1^{\rm na}\varepsilon_\tau\ge 1-\eta.
\end{equation}

\section{Trade-off relation for time-independent driving}\label{app:relax.tradeoff}
We consider thermodynamic processes that are driven by time-independent protocols.
That is, the system relaxes toward a stationary state, which can be equilibrium or nonequilibrium.
For such the processes, it was proved that the irreversible entropy production is lower bounded by the relative entropy between the initial and final distributions \cite{Shiraishi.2019.PRL,Vu.2021.PRL2},
\begin{equation}
	\Sigma_\tau\ge\mca{D}(p_0\|p_\tau),
\end{equation}
where $\mca{D}(p\|q)\coloneqq\sum_np_n\ln(p_n/q_n)$.
Due to the monotonicity of the relative entropy under information processing, we obtain a further lower bound for $\Sigma_\tau$ in terms of the relative entropy between coarse-grained probability distributions as
\begin{equation}
	\Sigma_\tau\ge\mca{D}(P_0\|P_\tau),
\end{equation}
where $\ket{P_t}\coloneqq[p_{\mfr{u}}(t),p_{\mfr{d}}(t)]^\top$.
By applying Proposition \ref{prop2}, we can arrive at the following inequality:
\begin{equation}
	\Psi\qty(e^{1/\varepsilon_0}\Sigma_\tau)\ge\varepsilon_\tau^{-1}-\varepsilon_0^{-1},
\end{equation}
which can be translated into the form of trade-off relation \eqref{eq:main.res} by defining the cost as $\mca{C}\coloneqq\tau^{-1}\Psi\qty(e^{1/\varepsilon_0}\Sigma_\tau)$.

Using an information-geometric approach \cite{Kolchinsky.2022.arxiv}, the irreversible entropy production can be decomposed into nonnegative housekeeping and excess parts as
\begin{equation}
	\Sigma_\tau=\Sigma_\tau^{\rm hk}+\Sigma_\tau^{\rm ex}.
\end{equation}
Roughly speaking, the excess entropy production is the minimum cost of a process driven by conservative forces with the same trajectory of probability distributions.
The excess entropy production was shown to be lower bounded by the relative entropy between the initial and final distributions as \cite{Kolchinsky.2022.arxiv}
\begin{equation}
	\Sigma_\tau^{\rm ex}\ge \mca{D}(p_0\|p_\tau).
\end{equation}
Following the same procedure as above, we readily obtain a similar trade-off relation in terms of the excess entropy production.

\section{Proof of Eq.~\eqref{eq:temp.lb}}\label{app:proof.temp.lb}
The definition \eqref{eq:eff.temp.def} of the effective temperature can be explicitly written as
\begin{equation}
	\sum_{n\ge 1}E_n\pi_n=\sum_{n\ge 1}E_np_n(\tau),\label{eq:ener.con}
\end{equation}
where $\pi_n=e^{-\beta E_n}/\sum_{n\ge 1}e^{-\beta E_n}$ is the Gibbs thermal state with inverse temperature $\beta\coloneqq T^{-1}$ and $0\le E_1=\dots=E_\kappa<E_{\kappa+1}\le\dots\le E_d$ are the energy levels.
First, we derive a bound for the effective inverse temperature $\beta$ in terms of the error $\varepsilon_\tau=-[\ln\sum_{n>\kappa}p_n(\tau)]^{-1}$ (i.e., $\mfr{u}=\iset{d}\setminus\iset{\kappa}$).
Equation \eqref{eq:ener.con} can be transformed as follows:
\begin{equation}\label{eq:ener.con.tmp1}
	\frac{\kappa E_1+\sum_{n>\kappa}E_ne^{-\beta(E_n-E_1)}}{\kappa+\sum_{n>\kappa}e^{-\beta(E_n-E_1)}}=\sum_{n\ge 1}E_np_n(\tau).
\end{equation}
Subtracting $E_1=E_1\sum_{n\ge 1}p_n(\tau)$ from both sides of Eq.~\eqref{eq:ener.con.tmp1}, it can be reduced to
\begin{equation}
	\frac{\sum_{n>\kappa}(E_n-E_1)e^{-\beta(E_n-E_1)}}{\kappa+\sum_{n>\kappa}e^{-\beta(E_n-E_1)}}=\sum_{n>\kappa}(E_n-E_1)p_n(\tau).\label{eq:ener.con.tmp2}
\end{equation}
Defining the energy gap $\Delta_g\coloneqq E_{\kappa+1}-E_1>0$ and the energy bandwidth $\Delta_b\coloneqq E_d-E_1>0$, we can evaluate both sides of Eq.~\eqref{eq:ener.con.tmp2} as
\begin{align}
	&\sum_{n>\kappa}(E_n-E_1)p_n(\tau)\ge \Delta_g\sum_{n>\kappa}p_n(\tau)=\Delta_gp_{\mfr{u}}(\tau),\\
	&\frac{\sum_{n>\kappa}(E_n-E_1)e^{-\beta(E_n-E_1)}}{\kappa+\sum_{n>\kappa}e^{-\beta(E_n-E_1)}}\le\frac{(d-\kappa)\Delta_be^{-\beta\Delta_g}}{\kappa+e^{-\beta\Delta_g}}.
\end{align}
Consequently, an upper bound on the inverse temperature can be obtained as
\begin{equation}
	\beta\le \frac{1}{\Delta_g}\ln\qty[\frac{(d-\kappa)\Delta_b}{\kappa\Delta_g}\frac{1}{p_{\mfr{u}}(\tau)}-\frac{1}{\kappa}].\label{eq:inv.temp.ub}
\end{equation}
On the other hand, using the trade-off relation \eqref{eq:main.res}, the following upper bound on $\ln[1/p_{\mfr{u}}(\tau)]$ can be derived:
\begin{equation}
	\ln\frac{1}{p_{\mfr{u}}(\tau)}=\varepsilon_\tau^{-1}\le\tau\mca{C}+\varepsilon_0^{-1}.\label{eq:inv.temp.ub.tmp1}
\end{equation}
Combining Eqs.~\eqref{eq:inv.temp.ub} and \eqref{eq:inv.temp.ub.tmp1} leads to a lower bound of the achievable temperature as
\begin{equation}
	T\ge\frac{\Delta_g}{\ln\qty[\frac{(d-\kappa)\Delta_b}{\kappa\Delta_g}e^{\tau\mca{C}+\varepsilon_0^{-1}}-\frac{1}{\kappa}]},\label{eq:temp.lb.res}
\end{equation}
which is exactly the relation \eqref{eq:temp.lb}.

\section{Proof of Eq.~\eqref{eq:copy.mi.lb}}\label{app:copy.mi.lb}
Since $I_{\rm max} = S_0(\msf{X}) = S_\tau(\msf{X})$, the difference between $I_{\rm max}$ and the mutual information $I_\tau$ can be expressed in terms of the conditional entropy as
\begin{equation}\label{eq:mi.tmp1}
	I_{\rm max} - I_\tau = \sum_{y} p_\tau(y) S_\tau(\msf{X}|y).
\end{equation}
Note that the Shannon entropy of any binary distribution can be bounded from below by its elements as
\begin{align}
S(p) &\coloneqq -p\ln p - (1-p)\ln(1-p) \notag\\
&\ge (2\ln 2) \min(p, 1-p).
\end{align}
Applying this inequality to Eq.~\eqref{eq:mi.tmp1}, we obtain a lower bound for $I_{\rm max} - I_\tau$ as follows:
\begin{align}
I_{\rm max} - I_\tau &\ge (2\ln 2) \min\{p_\tau(0,0), p_\tau(1,0)\} \notag\\
&+ (2\ln 2) \min\{p_\tau(0,1), p_\tau(1,1)\} \notag\\
&= (2\ln 2) \min[\min\{p_\tau(x=0), 1-p_\tau(x=0)\},\notag\\
&\hspace{2.15cm}\min\{p_e(\tau), 1-p_e(\tau)\}]. \label{eq:mi.tmp2}
\end{align}
Here, we use the identity $\min(x,y) + \min(z,w) = \min(x+z, y+w, x+w, y+z)$ to derive the last line. Next, note that $p_e(\tau) \le 1/2$, $\min(p, 1-p) = h[S(p)]$, where $h \in [0, 1/2]$ is the inverse function of $S(p)$, and $\min(x,y) \ge 2xy$ for any $0 \le x,y \le 1/2$. Using these relations, Eq.~\eqref{eq:mi.tmp2} further yields:
\begin{align}
I_{\rm max} - I_\tau &\ge (2\ln 2) \min\{h(I_{\rm max}), p_e(\tau)\} \notag\\
&\ge (4\ln 2) h(I_{\rm max}) p_e(\tau).
\end{align}
Finally, combining this inequality with Eq.~\eqref{eq:copy.err.lb} immediately leads to Eq.~\eqref{eq:copy.mi.lb},
\begin{equation}
I_{\rm max} - I_\tau \ge (4\ln 2) h(I_{\rm max}) p_e(0) e^{-\tau\mca{C}}.
\end{equation}

\section{Derivation of quantum generalizations}
\subsection{Markovian dynamics with infinite-size reservoirs}\label{app:proof.quantum.markov}
Here we derive the trade-off relation \eqref{eq:main.res} for quantum Markovian dynamics. Recall the spectral decomposition of the density matrix: $\varrho_t=\sum_np_n(t)\dyad{n_t}$.
For simplicity, we define the incoming and outgoing transition frequencies as
\begin{align}
	a_{n\leftarrow}&\coloneqq\sum_{k,m}w_{nm}^kp_{m},\\
	a_{n\rightarrow}&\coloneqq\sum_{k,m}w_{mn}^kp_n.
\end{align}
Note that $a_{n\leftarrow}+a_{n\rightarrow}=a_n$.
Following the same technique in Appendix \ref{app:proof.main.res}, we can prove that
\begin{equation}
	\frac{a_{n\leftarrow}}{a_{n\rightarrow}}\ge{\phi}^{-1}\qty[\phi\qty(\frac{a_{n\leftarrow}}{a_{n\rightarrow}})]={\phi}^{-1}\qty[\frac{F(a_{n\leftarrow},a_{n\rightarrow})}{a_n}]\ge {\phi}^{-1}\qty(\frac{\sigma_n}{a_n}).
\end{equation}
Taking the time derivative of $p_{n}(t)=\mel{n_t}{\varrho_t}{n_t}$, we can proceed as follows:
\begin{align}
	-\frac{\dot p_n}{p_n}&=\frac{1}{p_n}\sum_{k,m}\qty(w_{mn}^kp_n - w_{nm}^{k'}p_m)\notag\\
	&=\qty(\sum_{k,m}w_{mn}^k)\qty(1-\frac{a_{n\leftarrow}}{a_{n\rightarrow}})\notag\\
	&\le\qty(\sum_{k,m}w_{mn}^k)\qty[1-{\phi}^{-1}\qty(\frac{\sigma_n}{a_n})]\notag\\
	&\le\omega_t\qty[1-{\phi}^{-1}\qty(\frac{\sigma_n}{a_n})].
\end{align}
Here, we apply the following inequality to obtain the last line:
\begin{align}
	\sum_{k,m}w_{mn}^k&=\sum_{k,m}|\mel{m}{L_k}{n}|^2\notag\\
	&=\sum_{k,m}\mel{n}{L_k^\dagger}{m}\mel{m}{L_k}{n}\notag\\
	&=\sum_k\mel{n}{L_k^\dagger L_k}{n}\notag\\
	&\le\|\sum_kL_k^\dagger L_k\|=\omega_t.
\end{align}
Let $\ket{r}$ be a reference distribution, which will be determined later.
Taking the time derivative of the relative entropy $\mca{D}(r\|p_t)$, we can evaluate as follows:
\begin{align}
	\frac{d}{dt}\mca{D}(r\|p_t)&=-\sum_{n=1}^dr_n\frac{\dot p_n(t)}{p_n(t)}\notag\\
	&\le\sum_{n=1}^dr_n\omega_t\qty{1-{\phi}^{-1}\qty[\frac{\sigma_n(t)}{a_n(t)}]}\notag\\
	&\le\omega_t\qty[1-{\phi}^{-1}(\overline\sigma_t)]=\omega_t\Phi(\overline\sigma_t).\label{eq:qua.proof.tmp2}
\end{align}
Taking the time integration of Eq.~\eqref{eq:qua.proof.tmp2} from $t=0$ to $t=\tau$ leads to the following inequality:
\begin{align}
	\mca{D}(r\|p_\tau)-\mca{D}(r\|p_0)&\le\int_0^\tau\dd{t}\omega_t\Phi\qty(\overline\sigma_t)\eqqcolon\tau\mca{C}_1\notag\\
	&\le\tau\omega\Phi\qty(\overline\sigma)\eqqcolon\tau\mca{C},\label{eq:qua.proof.tmp3}
\end{align}
where $\omega\coloneqq\max_{0\le t\le\tau}\omega_t$ and $\overline\sigma\coloneqq\ev{\overline\sigma_t}_\tau$.
Now, let us specify the reference distribution as $r_n=\delta_{nk}$, where $k=\arg\max_np_n(0)/p_n(\tau)$.
By applying Proposition \ref{prop1}, we can further bound the quantity in the left-hand side of Eq.~\eqref{eq:qua.proof.tmp3} from below as
\begin{align}
	\mca{D}(r\|p_\tau)-\mca{D}(r\|p_0)&=\max_n\ln\frac{p_n(0)}{p_n(\tau)}\notag\\
	&\ge\max_n\ln\frac{p_n^\le(0)}{p_n^\le(\tau)}\notag\\
	&\ge\ln\frac{\sum_{n\in \mfr{u}}p_n^\le(0)}{\sum_{n\in \mfr{u}}p_n^\le(\tau)}\notag\\
	&=\varepsilon_\tau^{-1}-\varepsilon_0^{-1}.\label{eq:qua.proof.tmp4}
\end{align}
Combining Eqs.~\eqref{eq:qua.proof.tmp3} and \eqref{eq:qua.proof.tmp4} leads to the desired relation \eqref{eq:main.res}.

\subsection{Lower bound on the achievable temperature}\label{app:qtemp.lb}
Here we derive a lower bound on the achievable temperature in terms of time and cost.
To this end, we define the effective temperature $T$ associated with the final quantum state $\varrho_\tau$ as
\begin{equation}\label{eq:q.temp.def}
	\tr{H\varrho_\tau}=\tr{H\pi(H,T)},
\end{equation}
where $\pi(H,T)\coloneqq e^{-\beta H}/\tr e^{-\beta H}$ is the Gibbs thermal state and $\beta\coloneqq 1/T$.
Hereafter, we omit the time notation $\tau$ for simplicity.
Consider the spectral decompositions of the quantum state and the Hamiltonian, $\varrho=\sum_np_n^\le\dyad{n^\le}$ and $H=\sum_nE_n\dyad{\epsilon_n}$, where $0\le E_1=\dots=E_\kappa<E_{\kappa+1}\le\dots\le E_d$.
The equality \eqref{eq:q.temp.def} can be expressed as
\begin{equation}\label{eq:qtemp.tmp1}
	\sum_{n\ge 1}\mel{n^\le}{H}{n^\le}p_n^\le=\frac{\sum_{n\ge 1}E_ne^{-\beta E_n}}{\sum_{n\ge 1}e^{-\beta E_n}}.
\end{equation}
Subtracting $E_1$ from both sides of Eq.~\eqref{eq:qtemp.tmp1}, we obtain
\begin{equation}
	\sum_{n\ge 1}\qty(\mel{n^\le}{H}{n^\le}-E_1)p_n^\le=\frac{\sum_{n>\kappa}(E_n-E_1)e^{-\beta(E_n-E_1)}}{\kappa+\sum_{n>\kappa}e^{-\beta(E_n-E_1)}}.
\end{equation}
We consider the set $\mfr{u}=\{1\}$.
Since $p_n^\le\ge p_{\mfr{u}}$ for any $n\ge 1$, we can show that
\begin{align}
	&\sum_{n\ge 1}\qty(\mel{n^\le}{H}{n^\le}-E_1)p_n^\le\ge\sum_{n\ge 1}(E_n-E_1)p_{\mfr{u}}\ge(d-\kappa)\Delta_gp_{\mfr{u}},\\
	&\frac{\sum_{n>\kappa}(E_n-E_1)e^{-\beta(E_n-E_1)}}{\kappa+\sum_{n>\kappa}e^{-\beta(E_n-E_1)}}\le\frac{(d-\kappa)\Delta_be^{-\beta\Delta_g}}{\kappa+e^{-\beta\Delta_g}}.
\end{align}
Here, $\Delta_g\coloneqq E_{\kappa+1}-E_1$ and $\Delta_b\coloneqq E_d-E_1$ are the energy gap and the energy bandwidth, respectively.
Then, following the same procedure as in Appendix \ref{app:proof.temp.lb}, we obtain a similar lower bound on the achievable temperature:
\begin{equation}
	T\ge\frac{\Delta_g}{\ln\qty[\frac{\Delta_b}{\kappa\Delta_g}e^{\tau\mca{C}+\varepsilon_0^{-1}}-\frac{1}{\kappa}]}.
\end{equation}

\subsection{Non-Markovian dynamics with finite-size reservoirs}\label{app:proof.quantum.gen}
Here we show the proof of Eq.~\eqref{eq:main.res.nonMarkov}.
According to the monotonicity of the relative entropy under information processing, for any set of positive operator valued measures $\{\Pi_k\}_k$ satisfying $\sum_k\Pi_k=\mds{1}$, the following inequality holds true for any quantum states $\varrho$ and $\sigma$:
\begin{equation}
	\mca{D}(\varrho\|\sigma)\ge\sum_k\tr{\Pi_k\varrho}\ln\frac{\tr{\Pi_k\varrho}}{\tr{\Pi_k\sigma}}.\label{eq:rel.ent.proj}
\end{equation}
Let $\ket{n}$ and $\ket{\mu}$ be the eigenvectors corresponding to the smallest eigenvalues of $\varrho_S(\tau)$ and $\pi_R$, respectively.
Define two projective operators $\Pi_1\coloneqq\dyad{n,\mu}$ and $\Pi_2\coloneqq\mds{1}-\Pi_1$.
Applying the inequality \eqref{eq:rel.ent.proj} for these operators yields the following lower bound on the entropy production:
\begin{equation}\label{eq:ent.prod.lb}
	\Sigma_\tau\ge\mca{D}(u\|v),
\end{equation}
where $\ket{u}$ and $\ket{v}$ are two-dimensional probability distributions, given by
\begin{align}
	\ket{u}&\coloneqq[\tr{\Pi_1\varrho_\tau},\tr{\Pi_2\varrho_\tau}]^\top,\\
	\ket{v}&\coloneqq[\tr{\Pi_1\varrho_S(\tau)\otimes\pi_R},\tr{\Pi_2\varrho_S(\tau)\otimes\pi_R}]^\top.
\end{align}
By applying Proposition \ref{prop2}, the relative entropy between distributions $\ket{u}$ and $\ket{v}$ can be bounded further from below as
\begin{equation}\label{eq:KL.lb.ine}
	\mca{D}(u\|v)\ge u_1\psi(\max\{0,x\}),
\end{equation}
where $x\coloneqq\ln(u_1/v_1)$.
Note that $\lambda=\lambda_S(0)\lambda_R$ and $v_1=\lambda_S(\tau)\lambda_R$, where $\lambda_R$ denotes the smallest eigenvalue of $\pi_R$.
Furthermore, since $\lambda$ is the smallest eigenvalue of $\varrho_0$, we readily obtain that
\begin{equation}\label{eq:eigen.lb}
	u_1=\mel{n,\mu}{\varrho_\tau}{n,\mu}=\mel{n,\mu}{U\varrho_0U^\dagger}{n,\mu}\ge\lambda.
\end{equation}
Combining Eqs.~\eqref{eq:ent.prod.lb}, \eqref{eq:KL.lb.ine}, and \eqref{eq:eigen.lb} yields the following inequality:
\begin{equation}\label{eq:ent.prod.lb.tmp1}
	\lambda^{-1}\Sigma_\tau\ge \psi(\max\{0,x\}).
\end{equation}
Noticing that $\Psi$ is an increasing function, Eq.~\eqref{eq:ent.prod.lb.tmp1} immediately implies
\begin{equation}
	\Psi\qty(\lambda^{-1}\Sigma_\tau)\ge\max\{0,x\}\ge x.\label{eq:ent.prod.lb.tmp2}
\end{equation}
Finally, note that
\begin{equation}
	x=\ln(u_1/v_1)\ge\ln[\lambda_S(0)\lambda_R/\lambda_S(\tau)\lambda_R]=\varepsilon_\tau^{-1}-\varepsilon_0^{-1}.\label{eq:prob.rat.lb}
\end{equation}
Combining Eqs.~\eqref{eq:ent.prod.lb.tmp2} and \eqref{eq:prob.rat.lb} leads to the desired relation \eqref{eq:main.res.nonMarkov}.

\section{Proof of inequalities used in the derivation of main results}\label{app:useful.ines}

\begin{proposition}\label{prop1}
Given two probability distributions $\ket{p}=[p_1,\dots,p_d]^\top$ and $\ket{q}=[q_1,\dots,q_d]^\top$, let $\ket{p^{\le}}$ and $\ket{q^{\le}}$ be their corresponding sorted distributions (i.e., $p_1^\le\le\dots\le p_d^\le$ and $q_1^\le\le\dots\le q_d^\le$).
Then, the following inequality holds true:
\begin{equation}
	\max_n\frac{p_n}{q_n}\ge\max_n\frac{p_n^\le}{q_n^\le}.\label{eq:app.prob.rat.ine}
\end{equation}
\end{proposition}
\begin{proof}
Consider the R{\'e}nyi divergence of order $\alpha$ of a distribution $\ket{p}$ from a distribution $\ket{q}$, defined as
\begin{equation}
	\mca{D}_\alpha(p\|q)\coloneqq\frac{1}{\alpha-1}\ln\sum_n\frac{p_n^\alpha}{q_n^{\alpha-1}}.
\end{equation}
Then, the log of the maximum ratio of the probabilities is given by
\begin{equation}
	\max_n\frac{p_n}{q_n}=e^{\mca{D}_\infty(p\|q)}.\label{eq:app.tmp6}
\end{equation}
First, we show that for any $\alpha>1$, the following inequality always holds:
\begin{equation}
	\mca{D}_\alpha(p\|q)\ge \mca{D}_\alpha(p^\le\|q^\le).\label{eq:app.tmp5}
\end{equation}
Indeed, according to the rearrangement inequality, we have
\begin{equation}
	\sum_n\frac{p_n^\alpha}{q_n^{\alpha-1}}\ge \sum_n\frac{(p_n^\le)^\alpha}{(q_n^\le)^{\alpha-1}},
\end{equation}
which immediately derives Eq.~\eqref{eq:app.tmp5}.
Taking the $\alpha\to\infty$ limit of Eq.~\eqref{eq:app.tmp5} leads to $\mca{D}_\infty(p\|q)\ge\mca{D}_\infty(p^\le\|q^\le)$.
Combining this with the equality \eqref{eq:app.tmp6} verifies Eq.~\eqref{eq:app.prob.rat.ine}.
\end{proof}

\begin{proposition}\label{prop2}
Given two distributions $\ket{p}=[p_1,p_2]^\top$ and $\ket{q}=[q_1,q_2]^\top$, the following inequality always holds true:
\begin{equation}\label{eq:app.lb.KL}
	\mca{D}(p\|q)\ge p_1\psi(\max\{0,x\}),
\end{equation}
where $x\coloneqq\ln(p_1/q_1)$.	
\end{proposition}
\begin{proof}
The relative entropy $\mca{D}(p\|q)$ can be expressed in terms of $p_1$ and $x$ as follows:
\begin{equation}
	\mca{D}(p\|q)=p_1\qty(x+\frac{1-p_1}{p_1}\ln\frac{1-p_1}{1-e^{-x}p_1}).
\end{equation}
In the case of $x\le 0$, the inequality \eqref{eq:app.lb.KL} trivially holds since $\mca{D}(p\|q)\ge 0=\psi(0)=\psi(\max\{0,x\})$.
For $x>0$, the inequality \eqref{eq:app.lb.KL} can be proved using the definition of $\psi(x)$ as follows:
\begin{align}
	\mca{D}(p\|q)&\ge p_1\min_{p\in[0,1]}\qty(x+\frac{1-p}{p}\ln\frac{1-p}{1-e^{-x}p})\notag\\
	&=p_1\psi(x)=p_1\psi(\max\{0,x\}).
\end{align}
\end{proof}

\begin{proposition}\label{prop3}
The function $\psi(x)$ defined in Eq.~\eqref{eq:psi.func.def} is a monotonically increasing function and is bounded from below as
\begin{equation}\label{eq:psi.lb.KL}
	\psi(x)\ge\frac{2x}{1+e^{3/x}}
\end{equation}
for any $x\ge 0$.
\end{proposition}
\begin{proof}
First, we prove the monotonicity of function $\psi$, i.e., $\psi(x)\ge \psi(y)$ for any $x\ge y\ge 0$.
Define the following function:
\begin{equation}
	\psi_p(x)\coloneqq x+\frac{1-p}{p}\ln\frac{1-p}{1-e^{-x}p}.
\end{equation}
Then, the equality $\psi(x)=\min_{p\in[0,1]}\psi_p(x)$ holds trivially.
We will prove that $\psi_p(x)\ge\psi_p(y)$ for any $x\ge y\ge 0$.
To this end, we take the derivative of $\psi_p(x)$ with respect to $x$ and obtain
\begin{equation}
	\psi_p'(x)=\frac{1-e^{-x}}{1-e^{-x}p}\ge 0.
\end{equation}
Therefore, $\psi(x)=\min_{p\in[0,1]}\psi_p(x)\ge\min_{p\in[0,1]}\psi_p(y)=\psi(y)$, which verifies the monotonicity of $\psi$.

Next, we derive the lower bound \eqref{eq:psi.lb.KL} for $\psi(x)$.
The inequality \eqref{eq:psi.lb.KL} can be explicitly written as
\begin{equation}
	\min_{p\in[0,1]}\qty(x+\frac{1-p}{p}\ln\frac{1-p}{1-e^{-x}p})\ge\frac{2x}{1+e^{3/x}}.
\end{equation}
Define the following function:
\begin{equation}
	g_p(x)\coloneqq x+\frac{1-p}{p}\ln\frac{1-p}{1-e^{-x}p}-\frac{2x}{1+e^{3/x}}.
\end{equation}
Proving Eq.~\eqref{eq:psi.lb.KL} is equivalent to showing that $g_p(x)\ge 0$ for any $x\in[0,+\infty)$ and $p\in[0,1]$.
To this end, it is sufficient to prove that $g_p(x)$ is an increasing function over $[0,+\infty)$ [because $g_p(0)=0$], or equivalently $g_p'(x)\ge 0$.
Taking the derivative of $g_p(x)$ with respect to $x$, we obtain
\begin{equation}
	g_p'(x)=1-\frac{2}{1+e^{3/x}}-\frac{1-p}{e^x-p}-\frac{6e^{3/x}}{x(1+e^{3/x})^2}.
\end{equation}
Since $(1-p)/(e^x-p)\le e^{-x}$, the following inequality holds true:
\begin{equation}
	g_p'(x)\ge 1-\frac{2}{1+e^{3/x}}-e^{-x}-\frac{6e^{3/x}}{x(1+e^{3/x})^2}.
\end{equation}
Therefore, we need only prove that
\begin{equation}
	1-\frac{2}{1+e^{3/x}}-e^{-x}-\frac{6e^{3/x}}{x(1+e^{3/x})^2}\ge 0,
\end{equation}
which is equivalent to the following inequality by changing $x\to 1/x$:
\begin{equation}
	e^{1/x}\qty(e^{6x}-6xe^{3x}-1)\ge (e^{3x}+1)^2.\label{eq:app.tmp1}
\end{equation}
Since
\begin{equation}
	e^{1/x}\ge\sum_{n=0}^{4}\frac{1}{n!x^n}=1+\frac{1}{x}+\frac{1}{2x^2}+\frac{1}{6x^3}+\frac{1}{24x^4}\eqqcolon h(x),
\end{equation}
the inequality \eqref{eq:app.tmp1} will be immediately validated if we could prove that
\begin{equation}
	h(x)\ge\frac{(e^{3x}+1)^2}{e^{6x}-6xe^{3x}-1}.\label{eq:app.tmp2}
\end{equation}
Through simple algebraic calculations, we can show that Eq.~\eqref{eq:app.tmp2} is equivalent to
\begin{widetext}
\begin{align}
	e^{3x}&\ge\frac{3xh(x)+1+\sqrt{[3xh(x)+1]^2+h(x)^2-1}}{h(x)-1}\notag\\
	&=3x+\frac{24x^4+72x^5+\sqrt{5184 x^{10}+13824 x^9+14400 x^8+9792 x^7+5184 x^6+2208 x^5+744 x^4+216 x^3+49 x^2+8 x+1}}{24 x^3+12 x^2+4 x+1}.\label{eq:app.tmp3}
\end{align}
To prove Eq.~\eqref{eq:app.tmp3}, it is sufficient to use the following lower bound of $e^{3x}$:
\begin{equation}
	e^{3x}\ge 1 + 3 x + \frac{9 x^2}{2} + \frac{9 x^3}{2} + \frac{27 x^4}{8}.
\end{equation}
Using this inequality and performing some transformations, Eq.~\eqref{eq:app.tmp3} is reduced to proving the following inequality:
\begin{align}
	139968 x^{11}+513216 x^{10}+828144 x^9+833328 x^8&+604152 x^7+283608 x^6+72819 x^5+64 x+192 \notag\\
	&\ge 5304 x^4+4296 x^3+1184 x^2.\label{eq:app.tmp4}
\end{align}
\end{widetext}
By applying the Cauchy-Schwarz inequality, it can be easily shown that
\begin{align}
	833328 x^8+100&\ge 2\sqrt{833328\times 100}x^4\ge 5304 x^4 + 12953 x^4,\\
	283608 x^6 + 20&\ge 2\sqrt{283608\times 20}x^3\ge 4296 x^3,\\
	12953 x^4 + 70&\ge 2\sqrt{12953\times 70}x^2\ge 1184 x^2.
\end{align}
Combining these inequalities verifies Eq.~\eqref{eq:app.tmp4}, thereby completing the proof.
\end{proof}

\begin{proposition}\label{prop4}
The inverse function $\phi^{-1}$ of $\phi(x)\coloneqq(x+1)^{-1}(x-1)\ln x$ is lower bounded as
\begin{equation}
	\phi^{-1}(x)\ge e^{-c\max\qty(\sqrt{x},x)}~\forall x\ge 0,
\end{equation}
where $c\approx 1.543$ is the solution of equation $z(1-e^{-z})=1+e^{-z}$.
\end{proposition}
\begin{proof}
We divide into two cases: $x\ge 1$ and $0\le x<1$.
For $x\ge 1$, it is sufficient to prove $\phi^{-1}(x)\ge e^{-cx}$.
Since $\phi(x)$ is a decreasing function over $[0,1]$, we need only show that
\begin{equation}
	\phi[\phi^{-1}(x)]=x\le\phi(e^{-cx})=\frac{cx(1-e^{-cx})}{1+e^{-cx}},
\end{equation}
which is equivalently expressed as
\begin{equation}\label{eq:prop4.tmp1}
	\frac{1+e^{-cx}}{1-e^{-cx}}\le c.
\end{equation}
Since the function in the left-hand side is a decreasing function with respect to $x$, Eq.~\eqref{eq:prop4.tmp1} is immediately obtained as
\begin{equation}
	\frac{1+e^{-cx}}{1-e^{-cx}}\le\frac{1+e^{-c}}{1-e^{-c}}=c.
\end{equation}
For the $0\le x<1$ case, we can apply the same strategy to prove $\phi^{-1}(x)\ge e^{-c\sqrt{x}}$, which is equivalent to the following inequality:
\begin{equation}\label{eq:prop4.tmp2}
	x\le\frac{c\sqrt{x}(1-e^{-c\sqrt{x}})}{1+e^{-c\sqrt{x}}}.
\end{equation}
Defining the function
\begin{equation}
	g(x)\coloneqq\frac{1-e^{-cx}}{x(1+e^{-cx})},
\end{equation}
then its derivative with respect to $x$ can be calculated as
\begin{equation}
	g'(x)=\frac{1-e^{2cx}+2cxe^{cx}}{x^2(1+e^{cx})^2}.
\end{equation}
It is evident that $g'(x)\le 0$ because
\begin{equation}
	\frac{d}{dx}\qty(1-e^{2cx}+2cxe^{cx})=2ce^{cx}(1+cx-e^{cx})\le 0.
\end{equation}
Since $g(x)$ is a decreasing function, we obtain $g(\sqrt{x})\ge g(1)=1/c$, which is exactly Eq.~\eqref{eq:prop4.tmp2}.
This completes the proof.
\end{proof}


\begin{thebibliography}{120}%
\makeatletter
\providecommand \@ifxundefined [1]{%
 \@ifx{#1\undefined}
}%
\providecommand \@ifnum [1]{%
 \ifnum #1\expandafter \@firstoftwo
 \else \expandafter \@secondoftwo
 \fi
}%
\providecommand \@ifx [1]{%
 \ifx #1\expandafter \@firstoftwo
 \else \expandafter \@secondoftwo
 \fi
}%
\providecommand \natexlab [1]{#1}%
\providecommand \enquote  [1]{``#1''}%
\providecommand \bibnamefont  [1]{#1}%
\providecommand \bibfnamefont [1]{#1}%
\providecommand \citenamefont [1]{#1}%
\providecommand \href@noop [0]{\@secondoftwo}%
\providecommand \href [0]{\begingroup \@sanitize@url \@href}%
\providecommand \@href[1]{\@@startlink{#1}\@@href}%
\providecommand \@@href[1]{\endgroup#1\@@endlink}%
\providecommand \@sanitize@url [0]{\catcode `\\12\catcode `\$12\catcode `\&12\catcode `\#12\catcode `\^12\catcode `\_12\catcode `\%12\relax}%
\providecommand \@@startlink[1]{}%
\providecommand \@@endlink[0]{}%
\providecommand \url  [0]{\begingroup\@sanitize@url \@url }%
\providecommand \@url [1]{\endgroup\@href {#1}{\urlprefix }}%
\providecommand \urlprefix  [0]{URL }%
\providecommand \Eprint [0]{\href }%
\providecommand \doibase [0]{https://doi.org/}%
\providecommand \selectlanguage [0]{\@gobble}%
\providecommand \bibinfo  [0]{\@secondoftwo}%
\providecommand \bibfield  [0]{\@secondoftwo}%
\providecommand \translation [1]{[#1]}%
\providecommand \BibitemOpen [0]{}%
\providecommand \bibitemStop [0]{}%
\providecommand \bibitemNoStop [0]{.\EOS\space}%
\providecommand \EOS [0]{\spacefactor3000\relax}%
\providecommand \BibitemShut  [1]{\csname bibitem#1\endcsname}%
\let\auto@bib@innerbib\@empty
\bibitem [{\citenamefont {Lan}\ \emph {et~al.}(2012)\citenamefont {Lan}, \citenamefont {Sartori}, \citenamefont {Neumann}, \citenamefont {Sourjik},\ and\ \citenamefont {Tu}}]{Lan.2012.NP}%
  \BibitemOpen
  \bibfield  {author} {\bibinfo {author} {\bibfnamefont {G.}~\bibnamefont {Lan}}, \bibinfo {author} {\bibfnamefont {P.}~\bibnamefont {Sartori}}, \bibinfo {author} {\bibfnamefont {S.}~\bibnamefont {Neumann}}, \bibinfo {author} {\bibfnamefont {V.}~\bibnamefont {Sourjik}},\ and\ \bibinfo {author} {\bibfnamefont {Y.}~\bibnamefont {Tu}},\ }\bibfield  {title} {\bibinfo {title} {{The energy–speed–accuracy trade-off in sensory adaptation}},\ }\href {https://doi.org/10.1038/nphys2276} {\bibfield  {journal} {\bibinfo  {journal} {Nat. Phys.}\ }\textbf {\bibinfo {volume} {8}},\ \bibinfo {pages} {422} (\bibinfo {year} {2012})}\BibitemShut {NoStop}%
\bibitem [{\citenamefont {B{\'{e}}rut}\ \emph {et~al.}(2012)\citenamefont {B{\'{e}}rut}, \citenamefont {Arakelyan}, \citenamefont {Petrosyan}, \citenamefont {Ciliberto}, \citenamefont {Dillenschneider},\ and\ \citenamefont {Lutz}}]{Brut.2012.N}%
  \BibitemOpen
  \bibfield  {author} {\bibinfo {author} {\bibfnamefont {A.}~\bibnamefont {B{\'{e}}rut}}, \bibinfo {author} {\bibfnamefont {A.}~\bibnamefont {Arakelyan}}, \bibinfo {author} {\bibfnamefont {A.}~\bibnamefont {Petrosyan}}, \bibinfo {author} {\bibfnamefont {S.}~\bibnamefont {Ciliberto}}, \bibinfo {author} {\bibfnamefont {R.}~\bibnamefont {Dillenschneider}},\ and\ \bibinfo {author} {\bibfnamefont {E.}~\bibnamefont {Lutz}},\ }\bibfield  {title} {\bibinfo {title} {{Experimental verification of Landauer's principle linking information and thermodynamics}},\ }\href {https://doi.org/10.1038/nature10872} {\bibfield  {journal} {\bibinfo  {journal} {Nature}\ }\textbf {\bibinfo {volume} {483}},\ \bibinfo {pages} {187} (\bibinfo {year} {2012})}\BibitemShut {NoStop}%
\bibitem [{\citenamefont {Jun}\ \emph {et~al.}(2014)\citenamefont {Jun}, \citenamefont {Gavrilov},\ and\ \citenamefont {Bechhoefer}}]{Jun.2014.PRL}%
  \BibitemOpen
  \bibfield  {author} {\bibinfo {author} {\bibfnamefont {Y.}~\bibnamefont {Jun}}, \bibinfo {author} {\bibfnamefont {M.~c.~v.}\ \bibnamefont {Gavrilov}},\ and\ \bibinfo {author} {\bibfnamefont {J.}~\bibnamefont {Bechhoefer}},\ }\bibfield  {title} {\bibinfo {title} {{High-precision test of Landauer's principle in a feedback trap}},\ }\href {https://doi.org/10.1103/PhysRevLett.113.190601} {\bibfield  {journal} {\bibinfo  {journal} {Phys. Rev. Lett.}\ }\textbf {\bibinfo {volume} {113}},\ \bibinfo {pages} {190601} (\bibinfo {year} {2014})}\BibitemShut {NoStop}%
\bibitem [{\citenamefont {Hong}\ \emph {et~al.}(2016)\citenamefont {Hong}, \citenamefont {Lambson}, \citenamefont {Dhuey},\ and\ \citenamefont {Bokor}}]{Hong.2016.SA}%
  \BibitemOpen
  \bibfield  {author} {\bibinfo {author} {\bibfnamefont {J.}~\bibnamefont {Hong}}, \bibinfo {author} {\bibfnamefont {B.}~\bibnamefont {Lambson}}, \bibinfo {author} {\bibfnamefont {S.}~\bibnamefont {Dhuey}},\ and\ \bibinfo {author} {\bibfnamefont {J.}~\bibnamefont {Bokor}},\ }\bibfield  {title} {\bibinfo {title} {{Experimental test of Landauer's principle in single-bit operations on nanomagnetic memory bits}},\ }\href {https://doi.org/10.1126/sciadv.1501492} {\bibfield  {journal} {\bibinfo  {journal} {Sci. Adv.}\ }\textbf {\bibinfo {volume} {2}},\ \bibinfo {pages} {e1501492} (\bibinfo {year} {2016})}\BibitemShut {NoStop}%
\bibitem [{\citenamefont {Yan}\ \emph {et~al.}(2018)\citenamefont {Yan}, \citenamefont {Xiong}, \citenamefont {Rehan}, \citenamefont {Zhou}, \citenamefont {Liang}, \citenamefont {Chen}, \citenamefont {Zhang}, \citenamefont {Yang}, \citenamefont {Ma},\ and\ \citenamefont {Feng}}]{Yan.2018.PRL}%
  \BibitemOpen
  \bibfield  {author} {\bibinfo {author} {\bibfnamefont {L.~L.}\ \bibnamefont {Yan}}, \bibinfo {author} {\bibfnamefont {T.~P.}\ \bibnamefont {Xiong}}, \bibinfo {author} {\bibfnamefont {K.}~\bibnamefont {Rehan}}, \bibinfo {author} {\bibfnamefont {F.}~\bibnamefont {Zhou}}, \bibinfo {author} {\bibfnamefont {D.~F.}\ \bibnamefont {Liang}}, \bibinfo {author} {\bibfnamefont {L.}~\bibnamefont {Chen}}, \bibinfo {author} {\bibfnamefont {J.~Q.}\ \bibnamefont {Zhang}}, \bibinfo {author} {\bibfnamefont {W.~L.}\ \bibnamefont {Yang}}, \bibinfo {author} {\bibfnamefont {Z.~H.}\ \bibnamefont {Ma}},\ and\ \bibinfo {author} {\bibfnamefont {M.}~\bibnamefont {Feng}},\ }\bibfield  {title} {\bibinfo {title} {{Single-atom demonstration of the quantum Landauer principle}},\ }\href {https://doi.org/10.1103/PhysRevLett.120.210601} {\bibfield  {journal} {\bibinfo  {journal} {Phys. Rev. Lett.}\ }\textbf {\bibinfo {volume} {120}},\ \bibinfo {pages} {210601} (\bibinfo {year} {2018})}\BibitemShut {NoStop}%
\bibitem [{\citenamefont {Dago}\ \emph {et~al.}(2021)\citenamefont {Dago}, \citenamefont {Pereda}, \citenamefont {Barros}, \citenamefont {Ciliberto},\ and\ \citenamefont {Bellon}}]{Dago.2021.PRL}%
  \BibitemOpen
  \bibfield  {author} {\bibinfo {author} {\bibfnamefont {S.}~\bibnamefont {Dago}}, \bibinfo {author} {\bibfnamefont {J.}~\bibnamefont {Pereda}}, \bibinfo {author} {\bibfnamefont {N.}~\bibnamefont {Barros}}, \bibinfo {author} {\bibfnamefont {S.}~\bibnamefont {Ciliberto}},\ and\ \bibinfo {author} {\bibfnamefont {L.}~\bibnamefont {Bellon}},\ }\bibfield  {title} {\bibinfo {title} {{Information and thermodynamics: Fast and precise approach to Landauer's bound in an underdamped micromechanical oscillator}},\ }\href {https://doi.org/10.1103/PhysRevLett.126.170601} {\bibfield  {journal} {\bibinfo  {journal} {Phys. Rev. Lett.}\ }\textbf {\bibinfo {volume} {126}},\ \bibinfo {pages} {170601} (\bibinfo {year} {2021})}\BibitemShut {NoStop}%
\bibitem [{\citenamefont {Scandi}\ \emph {et~al.}(2022)\citenamefont {Scandi}, \citenamefont {Barker}, \citenamefont {Lehmann}, \citenamefont {Dick}, \citenamefont {Maisi},\ and\ \citenamefont {Perarnau-Llobet}}]{Scandi.2022.PRL}%
  \BibitemOpen
  \bibfield  {author} {\bibinfo {author} {\bibfnamefont {M.}~\bibnamefont {Scandi}}, \bibinfo {author} {\bibfnamefont {D.}~\bibnamefont {Barker}}, \bibinfo {author} {\bibfnamefont {S.}~\bibnamefont {Lehmann}}, \bibinfo {author} {\bibfnamefont {K.~A.}\ \bibnamefont {Dick}}, \bibinfo {author} {\bibfnamefont {V.~F.}\ \bibnamefont {Maisi}},\ and\ \bibinfo {author} {\bibfnamefont {M.}~\bibnamefont {Perarnau-Llobet}},\ }\bibfield  {title} {\bibinfo {title} {{Minimally dissipative information erasure in a quantum dot via thermodynamic length}},\ }\href {https://doi.org/10.1103/PhysRevLett.129.270601} {\bibfield  {journal} {\bibinfo  {journal} {Phys. Rev. Lett.}\ }\textbf {\bibinfo {volume} {129}},\ \bibinfo {pages} {270601} (\bibinfo {year} {2022})}\BibitemShut {NoStop}%
\bibitem [{\citenamefont {Sekimoto}(2010)}]{Sekimoto.2010}%
  \BibitemOpen
  \bibfield  {author} {\bibinfo {author} {\bibfnamefont {K.}~\bibnamefont {Sekimoto}},\ }\href@noop {} {\emph {\bibinfo {title} {{Stochastic Energetics}}}},\ Vol.\ \bibinfo {volume} {799}\ (\bibinfo  {publisher} {Springer},\ \bibinfo {address} {Berlin},\ \bibinfo {year} {2010})\BibitemShut {NoStop}%
\bibitem [{\citenamefont {Seifert}(2012)}]{Seifert.2012.RPP}%
  \BibitemOpen
  \bibfield  {author} {\bibinfo {author} {\bibfnamefont {U.}~\bibnamefont {Seifert}},\ }\bibfield  {title} {\bibinfo {title} {{Stochastic thermodynamics, fluctuation theorems and molecular machines}},\ }\href {https://doi.org/10.1088/0034-4885/75/12/126001} {\bibfield  {journal} {\bibinfo  {journal} {Rep. Prog. Phys.}\ }\textbf {\bibinfo {volume} {75}},\ \bibinfo {pages} {126001} (\bibinfo {year} {2012})}\BibitemShut {NoStop}%
\bibitem [{\citenamefont {Vinjanampathy}\ and\ \citenamefont {Anders}(2016)}]{Vinjanampathy.2016.CP}%
  \BibitemOpen
  \bibfield  {author} {\bibinfo {author} {\bibfnamefont {S.}~\bibnamefont {Vinjanampathy}}\ and\ \bibinfo {author} {\bibfnamefont {J.}~\bibnamefont {Anders}},\ }\bibfield  {title} {\bibinfo {title} {{Quantum thermodynamics}},\ }\href {https://doi.org/10.1080/00107514.2016.1201896} {\bibfield  {journal} {\bibinfo  {journal} {Contemp. Phys.}\ }\textbf {\bibinfo {volume} {57}},\ \bibinfo {pages} {545} (\bibinfo {year} {2016})}\BibitemShut {NoStop}%
\bibitem [{\citenamefont {Goold}\ \emph {et~al.}(2016)\citenamefont {Goold}, \citenamefont {Huber}, \citenamefont {Riera}, \citenamefont {del Rio},\ and\ \citenamefont {Skrzypczyk}}]{Goold.2016.JPA}%
  \BibitemOpen
  \bibfield  {author} {\bibinfo {author} {\bibfnamefont {J.}~\bibnamefont {Goold}}, \bibinfo {author} {\bibfnamefont {M.}~\bibnamefont {Huber}}, \bibinfo {author} {\bibfnamefont {A.}~\bibnamefont {Riera}}, \bibinfo {author} {\bibfnamefont {L.}~\bibnamefont {del Rio}},\ and\ \bibinfo {author} {\bibfnamefont {P.}~\bibnamefont {Skrzypczyk}},\ }\bibfield  {title} {\bibinfo {title} {{The role of quantum information in thermodynamics{\textemdash}a topical review}},\ }\href {https://doi.org/10.1088/1751-8113/49/14/143001} {\bibfield  {journal} {\bibinfo  {journal} {J. Phys. A}\ }\textbf {\bibinfo {volume} {49}},\ \bibinfo {pages} {143001} (\bibinfo {year} {2016})}\BibitemShut {NoStop}%
\bibitem [{\citenamefont {Deffner}\ and\ \citenamefont {Campbell}(2019)}]{Deffner.2019}%
  \BibitemOpen
  \bibfield  {author} {\bibinfo {author} {\bibfnamefont {S.}~\bibnamefont {Deffner}}\ and\ \bibinfo {author} {\bibfnamefont {S.}~\bibnamefont {Campbell}},\ }\href@noop {} {\emph {\bibinfo {title} {{Quantum Thermodynamics}}}}\ (\bibinfo  {publisher} {Morgan \& Claypool Publishers},\ \bibinfo {address} {San Rafael},\ \bibinfo {year} {2019})\BibitemShut {NoStop}%
\bibitem [{\citenamefont {Barato}\ and\ \citenamefont {Seifert}(2015)}]{Barato.2015.PRL}%
  \BibitemOpen
  \bibfield  {author} {\bibinfo {author} {\bibfnamefont {A.~C.}\ \bibnamefont {Barato}}\ and\ \bibinfo {author} {\bibfnamefont {U.}~\bibnamefont {Seifert}},\ }\bibfield  {title} {\bibinfo {title} {{Thermodynamic uncertainty relation for biomolecular processes}},\ }\href {https://doi.org/10.1103/PhysRevLett.114.158101} {\bibfield  {journal} {\bibinfo  {journal} {Phys. Rev. Lett.}\ }\textbf {\bibinfo {volume} {114}},\ \bibinfo {pages} {158101} (\bibinfo {year} {2015})}\BibitemShut {NoStop}%
\bibitem [{\citenamefont {Gingrich}\ \emph {et~al.}(2016)\citenamefont {Gingrich}, \citenamefont {Horowitz}, \citenamefont {Perunov},\ and\ \citenamefont {England}}]{Gingrich.2016.PRL}%
  \BibitemOpen
  \bibfield  {author} {\bibinfo {author} {\bibfnamefont {T.~R.}\ \bibnamefont {Gingrich}}, \bibinfo {author} {\bibfnamefont {J.~M.}\ \bibnamefont {Horowitz}}, \bibinfo {author} {\bibfnamefont {N.}~\bibnamefont {Perunov}},\ and\ \bibinfo {author} {\bibfnamefont {J.~L.}\ \bibnamefont {England}},\ }\bibfield  {title} {\bibinfo {title} {{Dissipation bounds all steady-state current fluctuations}},\ }\href {https://doi.org/10.1103/PhysRevLett.116.120601} {\bibfield  {journal} {\bibinfo  {journal} {Phys. Rev. Lett.}\ }\textbf {\bibinfo {volume} {116}},\ \bibinfo {pages} {120601} (\bibinfo {year} {2016})}\BibitemShut {NoStop}%
\bibitem [{\citenamefont {Horowitz}\ and\ \citenamefont {Gingrich}(2017)}]{Horowitz.2017.PRE}%
  \BibitemOpen
  \bibfield  {author} {\bibinfo {author} {\bibfnamefont {J.~M.}\ \bibnamefont {Horowitz}}\ and\ \bibinfo {author} {\bibfnamefont {T.~R.}\ \bibnamefont {Gingrich}},\ }\bibfield  {title} {\bibinfo {title} {{Proof of the finite-time thermodynamic uncertainty relation for steady-state currents}},\ }\href {https://doi.org/10.1103/PhysRevE.96.020103} {\bibfield  {journal} {\bibinfo  {journal} {Phys. Rev. E}\ }\textbf {\bibinfo {volume} {96}},\ \bibinfo {pages} {020103(R)} (\bibinfo {year} {2017})}\BibitemShut {NoStop}%
\bibitem [{\citenamefont {Dechant}\ and\ \citenamefont {Sasa}(2018)}]{Dechant.2018.JSM}%
  \BibitemOpen
  \bibfield  {author} {\bibinfo {author} {\bibfnamefont {A.}~\bibnamefont {Dechant}}\ and\ \bibinfo {author} {\bibfnamefont {S.-i.}\ \bibnamefont {Sasa}},\ }\bibfield  {title} {\bibinfo {title} {{Current fluctuations and transport efficiency for general Langevin systems}},\ }\href {https://doi.org/10.1088/1742-5468/aac91a} {\bibfield  {journal} {\bibinfo  {journal} {J. Stat. Mech.}\ }\textbf {\bibinfo {volume} {2018}},\ \bibinfo {pages} {063209} (\bibinfo {year} {2018})}\BibitemShut {NoStop}%
\bibitem [{\citenamefont {Hasegawa}\ and\ \citenamefont {Van~Vu}(2019)}]{Hasegawa.2019.PRL}%
  \BibitemOpen
  \bibfield  {author} {\bibinfo {author} {\bibfnamefont {Y.}~\bibnamefont {Hasegawa}}\ and\ \bibinfo {author} {\bibfnamefont {T.}~\bibnamefont {Van~Vu}},\ }\bibfield  {title} {\bibinfo {title} {{Fluctuation theorem uncertainty relation}},\ }\href {https://doi.org/10.1103/PhysRevLett.123.110602} {\bibfield  {journal} {\bibinfo  {journal} {Phys. Rev. Lett.}\ }\textbf {\bibinfo {volume} {123}},\ \bibinfo {pages} {110602} (\bibinfo {year} {2019})}\BibitemShut {NoStop}%
\bibitem [{\citenamefont {Timpanaro}\ \emph {et~al.}(2019)\citenamefont {Timpanaro}, \citenamefont {Guarnieri}, \citenamefont {Goold},\ and\ \citenamefont {Landi}}]{Timpanaro.2019.PRL}%
  \BibitemOpen
  \bibfield  {author} {\bibinfo {author} {\bibfnamefont {A.~M.}\ \bibnamefont {Timpanaro}}, \bibinfo {author} {\bibfnamefont {G.}~\bibnamefont {Guarnieri}}, \bibinfo {author} {\bibfnamefont {J.}~\bibnamefont {Goold}},\ and\ \bibinfo {author} {\bibfnamefont {G.~T.}\ \bibnamefont {Landi}},\ }\bibfield  {title} {\bibinfo {title} {{Thermodynamic uncertainty relations from exchange fluctuation theorems}},\ }\href {https://doi.org/10.1103/PhysRevLett.123.090604} {\bibfield  {journal} {\bibinfo  {journal} {Phys. Rev. Lett.}\ }\textbf {\bibinfo {volume} {123}},\ \bibinfo {pages} {090604} (\bibinfo {year} {2019})}\BibitemShut {NoStop}%
\bibitem [{\citenamefont {Koyuk}\ and\ \citenamefont {Seifert}(2020)}]{Koyuk.2020.PRL}%
  \BibitemOpen
  \bibfield  {author} {\bibinfo {author} {\bibfnamefont {T.}~\bibnamefont {Koyuk}}\ and\ \bibinfo {author} {\bibfnamefont {U.}~\bibnamefont {Seifert}},\ }\bibfield  {title} {\bibinfo {title} {{Thermodynamic uncertainty relation for time-dependent driving}},\ }\href {https://doi.org/10.1103/PhysRevLett.125.260604} {\bibfield  {journal} {\bibinfo  {journal} {Phys. Rev. Lett.}\ }\textbf {\bibinfo {volume} {125}},\ \bibinfo {pages} {260604} (\bibinfo {year} {2020})}\BibitemShut {NoStop}%
\bibitem [{\citenamefont {Miller}\ \emph {et~al.}(2021)\citenamefont {Miller}, \citenamefont {Mohammady}, \citenamefont {Perarnau-Llobet},\ and\ \citenamefont {Guarnieri}}]{Miller.2021.PRL.TUR}%
  \BibitemOpen
  \bibfield  {author} {\bibinfo {author} {\bibfnamefont {H.~J.~D.}\ \bibnamefont {Miller}}, \bibinfo {author} {\bibfnamefont {M.~H.}\ \bibnamefont {Mohammady}}, \bibinfo {author} {\bibfnamefont {M.}~\bibnamefont {Perarnau-Llobet}},\ and\ \bibinfo {author} {\bibfnamefont {G.}~\bibnamefont {Guarnieri}},\ }\bibfield  {title} {\bibinfo {title} {{Thermodynamic uncertainty relation in slowly driven quantum heat engines}},\ }\href {https://doi.org/10.1103/PhysRevLett.126.210603} {\bibfield  {journal} {\bibinfo  {journal} {Phys. Rev. Lett.}\ }\textbf {\bibinfo {volume} {126}},\ \bibinfo {pages} {210603} (\bibinfo {year} {2021})}\BibitemShut {NoStop}%
\bibitem [{\citenamefont {Van~Vu}\ and\ \citenamefont {Saito}(2022{\natexlab{a}})}]{Vu.2022.PRL.TUR}%
  \BibitemOpen
  \bibfield  {author} {\bibinfo {author} {\bibfnamefont {T.}~\bibnamefont {Van~Vu}}\ and\ \bibinfo {author} {\bibfnamefont {K.}~\bibnamefont {Saito}},\ }\bibfield  {title} {\bibinfo {title} {{Thermodynamics of precision in Markovian open quantum dynamics}},\ }\href {https://doi.org/10.1103/PhysRevLett.128.140602} {\bibfield  {journal} {\bibinfo  {journal} {Phys. Rev. Lett.}\ }\textbf {\bibinfo {volume} {128}},\ \bibinfo {pages} {140602} (\bibinfo {year} {2022}{\natexlab{a}})}\BibitemShut {NoStop}%
\bibitem [{\citenamefont {Horowitz}\ and\ \citenamefont {Gingrich}(2020)}]{Horowitz.2020.NP}%
  \BibitemOpen
  \bibfield  {author} {\bibinfo {author} {\bibfnamefont {J.~M.}\ \bibnamefont {Horowitz}}\ and\ \bibinfo {author} {\bibfnamefont {T.~R.}\ \bibnamefont {Gingrich}},\ }\bibfield  {title} {\bibinfo {title} {{Thermodynamic uncertainty relations constrain non-equilibrium fluctuations}},\ }\href {https://doi.org/10.1038/s41567-019-0702-6} {\bibfield  {journal} {\bibinfo  {journal} {Nat. Phys.}\ }\textbf {\bibinfo {volume} {16}},\ \bibinfo {pages} {15} (\bibinfo {year} {2020})}\BibitemShut {NoStop}%
\bibitem [{\citenamefont {Pietzonka}\ and\ \citenamefont {Seifert}(2018)}]{Pietzonka.2018.PRL}%
  \BibitemOpen
  \bibfield  {author} {\bibinfo {author} {\bibfnamefont {P.}~\bibnamefont {Pietzonka}}\ and\ \bibinfo {author} {\bibfnamefont {U.}~\bibnamefont {Seifert}},\ }\bibfield  {title} {\bibinfo {title} {{Universal trade-off between power, efficiency, and constancy in steady-state heat engines}},\ }\href {https://doi.org/10.1103/PhysRevLett.120.190602} {\bibfield  {journal} {\bibinfo  {journal} {Phys. Rev. Lett.}\ }\textbf {\bibinfo {volume} {120}},\ \bibinfo {pages} {190602} (\bibinfo {year} {2018})}\BibitemShut {NoStop}%
\bibitem [{\citenamefont {Hartich}\ and\ \citenamefont {Godec}(2021)}]{Hartich.2021.PRL}%
  \BibitemOpen
  \bibfield  {author} {\bibinfo {author} {\bibfnamefont {D.}~\bibnamefont {Hartich}}\ and\ \bibinfo {author} {\bibfnamefont {A.~c.~v.}\ \bibnamefont {Godec}},\ }\bibfield  {title} {\bibinfo {title} {{Thermodynamic uncertainty relation bounds the extent of anomalous diffusion}},\ }\href {https://doi.org/10.1103/PhysRevLett.127.080601} {\bibfield  {journal} {\bibinfo  {journal} {Phys. Rev. Lett.}\ }\textbf {\bibinfo {volume} {127}},\ \bibinfo {pages} {080601} (\bibinfo {year} {2021})}\BibitemShut {NoStop}%
\bibitem [{\citenamefont {Li}\ \emph {et~al.}(2019)\citenamefont {Li}, \citenamefont {Horowitz}, \citenamefont {Gingrich},\ and\ \citenamefont {Fakhri}}]{Li.2019.NC}%
  \BibitemOpen
  \bibfield  {author} {\bibinfo {author} {\bibfnamefont {J.}~\bibnamefont {Li}}, \bibinfo {author} {\bibfnamefont {J.~M.}\ \bibnamefont {Horowitz}}, \bibinfo {author} {\bibfnamefont {T.~R.}\ \bibnamefont {Gingrich}},\ and\ \bibinfo {author} {\bibfnamefont {N.}~\bibnamefont {Fakhri}},\ }\bibfield  {title} {\bibinfo {title} {{Quantifying dissipation using fluctuating currents}},\ }\href {https://doi.org/10.1038/s41467-019-09631-x} {\bibfield  {journal} {\bibinfo  {journal} {Nat. Commun.}\ }\textbf {\bibinfo {volume} {10}},\ \bibinfo {pages} {1666} (\bibinfo {year} {2019})}\BibitemShut {NoStop}%
\bibitem [{\citenamefont {Manikandan}\ \emph {et~al.}(2020)\citenamefont {Manikandan}, \citenamefont {Gupta},\ and\ \citenamefont {Krishnamurthy}}]{Manikandan.2020.PRL}%
  \BibitemOpen
  \bibfield  {author} {\bibinfo {author} {\bibfnamefont {S.~K.}\ \bibnamefont {Manikandan}}, \bibinfo {author} {\bibfnamefont {D.}~\bibnamefont {Gupta}},\ and\ \bibinfo {author} {\bibfnamefont {S.}~\bibnamefont {Krishnamurthy}},\ }\bibfield  {title} {\bibinfo {title} {{Inferring entropy production from short experiments}},\ }\href {https://doi.org/10.1103/PhysRevLett.124.120603} {\bibfield  {journal} {\bibinfo  {journal} {Phys. Rev. Lett.}\ }\textbf {\bibinfo {volume} {124}},\ \bibinfo {pages} {120603} (\bibinfo {year} {2020})}\BibitemShut {NoStop}%
\bibitem [{\citenamefont {Van~Vu}\ \emph {et~al.}(2020)\citenamefont {Van~Vu}, \citenamefont {Vo},\ and\ \citenamefont {Hasegawa}}]{Vu.2020.PRE}%
  \BibitemOpen
  \bibfield  {author} {\bibinfo {author} {\bibfnamefont {T.}~\bibnamefont {Van~Vu}}, \bibinfo {author} {\bibfnamefont {V.~T.}\ \bibnamefont {Vo}},\ and\ \bibinfo {author} {\bibfnamefont {Y.}~\bibnamefont {Hasegawa}},\ }\bibfield  {title} {\bibinfo {title} {{Entropy production estimation with optimal current}},\ }\href {https://doi.org/10.1103/PhysRevE.101.042138} {\bibfield  {journal} {\bibinfo  {journal} {Phys. Rev. E}\ }\textbf {\bibinfo {volume} {101}},\ \bibinfo {pages} {042138} (\bibinfo {year} {2020})}\BibitemShut {NoStop}%
\bibitem [{\citenamefont {Otsubo}\ \emph {et~al.}(2020)\citenamefont {Otsubo}, \citenamefont {Ito}, \citenamefont {Dechant},\ and\ \citenamefont {Sagawa}}]{Otsubo.2020.PRE}%
  \BibitemOpen
  \bibfield  {author} {\bibinfo {author} {\bibfnamefont {S.}~\bibnamefont {Otsubo}}, \bibinfo {author} {\bibfnamefont {S.}~\bibnamefont {Ito}}, \bibinfo {author} {\bibfnamefont {A.}~\bibnamefont {Dechant}},\ and\ \bibinfo {author} {\bibfnamefont {T.}~\bibnamefont {Sagawa}},\ }\bibfield  {title} {\bibinfo {title} {{Estimating entropy production by machine learning of short-time fluctuating currents}},\ }\href {https://doi.org/10.1103/PhysRevE.101.062106} {\bibfield  {journal} {\bibinfo  {journal} {Phys. Rev. E}\ }\textbf {\bibinfo {volume} {101}},\ \bibinfo {pages} {062106} (\bibinfo {year} {2020})}\BibitemShut {NoStop}%
\bibitem [{\citenamefont {Dechant}\ and\ \citenamefont {Sasa}(2021)}]{Dechant.2021.PRX}%
  \BibitemOpen
  \bibfield  {author} {\bibinfo {author} {\bibfnamefont {A.}~\bibnamefont {Dechant}}\ and\ \bibinfo {author} {\bibfnamefont {S.-i.}\ \bibnamefont {Sasa}},\ }\bibfield  {title} {\bibinfo {title} {{Improving thermodynamic bounds using correlations}},\ }\href {https://doi.org/10.1103/PhysRevX.11.041061} {\bibfield  {journal} {\bibinfo  {journal} {Phys. Rev. X}\ }\textbf {\bibinfo {volume} {11}},\ \bibinfo {pages} {041061} (\bibinfo {year} {2021})}\BibitemShut {NoStop}%
\bibitem [{\citenamefont {Aurell}\ \emph {et~al.}(2012)\citenamefont {Aurell}, \citenamefont {Gaw{\c{e}}dzki}, \citenamefont {Mej{\'\i}a-Monasterio}, \citenamefont {Mohayaee},\ and\ \citenamefont {Muratore-Ginanneschi}}]{Aurell.2012.JSP}%
  \BibitemOpen
  \bibfield  {author} {\bibinfo {author} {\bibfnamefont {E.}~\bibnamefont {Aurell}}, \bibinfo {author} {\bibfnamefont {K.}~\bibnamefont {Gaw{\c{e}}dzki}}, \bibinfo {author} {\bibfnamefont {C.}~\bibnamefont {Mej{\'\i}a-Monasterio}}, \bibinfo {author} {\bibfnamefont {R.}~\bibnamefont {Mohayaee}},\ and\ \bibinfo {author} {\bibfnamefont {P.}~\bibnamefont {Muratore-Ginanneschi}},\ }\bibfield  {title} {\bibinfo {title} {{Refined second law of thermodynamics for fast random processes}},\ }\href {https://doi.org/10.1007/s10955-012-0478-x} {\bibfield  {journal} {\bibinfo  {journal} {J. Stat. Phys.}\ }\textbf {\bibinfo {volume} {147}},\ \bibinfo {pages} {487} (\bibinfo {year} {2012})}\BibitemShut {NoStop}%
\bibitem [{\citenamefont {Shiraishi}\ \emph {et~al.}(2018)\citenamefont {Shiraishi}, \citenamefont {Funo},\ and\ \citenamefont {Saito}}]{Shiraishi.2018.PRL}%
  \BibitemOpen
  \bibfield  {author} {\bibinfo {author} {\bibfnamefont {N.}~\bibnamefont {Shiraishi}}, \bibinfo {author} {\bibfnamefont {K.}~\bibnamefont {Funo}},\ and\ \bibinfo {author} {\bibfnamefont {K.}~\bibnamefont {Saito}},\ }\bibfield  {title} {\bibinfo {title} {{Speed limit for classical stochastic processes}},\ }\href {https://doi.org/10.1103/PhysRevLett.121.070601} {\bibfield  {journal} {\bibinfo  {journal} {Phys. Rev. Lett.}\ }\textbf {\bibinfo {volume} {121}},\ \bibinfo {pages} {070601} (\bibinfo {year} {2018})}\BibitemShut {NoStop}%
\bibitem [{\citenamefont {Vo}\ \emph {et~al.}(2020)\citenamefont {Vo}, \citenamefont {Van~Vu},\ and\ \citenamefont {Hasegawa}}]{Vo.2020.PRE}%
  \BibitemOpen
  \bibfield  {author} {\bibinfo {author} {\bibfnamefont {V.~T.}\ \bibnamefont {Vo}}, \bibinfo {author} {\bibfnamefont {T.}~\bibnamefont {Van~Vu}},\ and\ \bibinfo {author} {\bibfnamefont {Y.}~\bibnamefont {Hasegawa}},\ }\bibfield  {title} {\bibinfo {title} {{Unified approach to classical speed limit and thermodynamic uncertainty relation}},\ }\href {https://doi.org/10.1103/PhysRevE.102.062132} {\bibfield  {journal} {\bibinfo  {journal} {Phys. Rev. E}\ }\textbf {\bibinfo {volume} {102}},\ \bibinfo {pages} {062132} (\bibinfo {year} {2020})}\BibitemShut {NoStop}%
\bibitem [{\citenamefont {Ito}\ and\ \citenamefont {Dechant}(2020)}]{Ito.2020.PRX}%
  \BibitemOpen
  \bibfield  {author} {\bibinfo {author} {\bibfnamefont {S.}~\bibnamefont {Ito}}\ and\ \bibinfo {author} {\bibfnamefont {A.}~\bibnamefont {Dechant}},\ }\bibfield  {title} {\bibinfo {title} {{Stochastic time evolution, information geometry, and the Cram\'er-Rao bound}},\ }\href {https://doi.org/10.1103/PhysRevX.10.021056} {\bibfield  {journal} {\bibinfo  {journal} {Phys. Rev. X}\ }\textbf {\bibinfo {volume} {10}},\ \bibinfo {pages} {021056} (\bibinfo {year} {2020})}\BibitemShut {NoStop}%
\bibitem [{\citenamefont {Van~Vu}\ and\ \citenamefont {Hasegawa}(2021{\natexlab{a}})}]{Vu.2021.PRL}%
  \BibitemOpen
  \bibfield  {author} {\bibinfo {author} {\bibfnamefont {T.}~\bibnamefont {Van~Vu}}\ and\ \bibinfo {author} {\bibfnamefont {Y.}~\bibnamefont {Hasegawa}},\ }\bibfield  {title} {\bibinfo {title} {{Geometrical bounds of the irreversibility in Markovian systems}},\ }\href {https://doi.org/10.1103/PhysRevLett.126.010601} {\bibfield  {journal} {\bibinfo  {journal} {Phys. Rev. Lett.}\ }\textbf {\bibinfo {volume} {126}},\ \bibinfo {pages} {010601} (\bibinfo {year} {2021}{\natexlab{a}})}\BibitemShut {NoStop}%
\bibitem [{\citenamefont {Yoshimura}\ and\ \citenamefont {Ito}(2021)}]{Yoshimura.2021.PRL}%
  \BibitemOpen
  \bibfield  {author} {\bibinfo {author} {\bibfnamefont {K.}~\bibnamefont {Yoshimura}}\ and\ \bibinfo {author} {\bibfnamefont {S.}~\bibnamefont {Ito}},\ }\bibfield  {title} {\bibinfo {title} {{Thermodynamic uncertainty relation and thermodynamic speed limit in deterministic chemical reaction networks}},\ }\href {https://doi.org/10.1103/PhysRevLett.127.160601} {\bibfield  {journal} {\bibinfo  {journal} {Phys. Rev. Lett.}\ }\textbf {\bibinfo {volume} {127}},\ \bibinfo {pages} {160601} (\bibinfo {year} {2021})}\BibitemShut {NoStop}%
\bibitem [{\citenamefont {Salazar}(2022)}]{Salazar.2022.PRE}%
  \BibitemOpen
  \bibfield  {author} {\bibinfo {author} {\bibfnamefont {D.~S.~P.}\ \bibnamefont {Salazar}},\ }\bibfield  {title} {\bibinfo {title} {{Lower bound for entropy production rate in stochastic systems far from equilibrium}},\ }\href {https://doi.org/10.1103/PhysRevE.106.L032101} {\bibfield  {journal} {\bibinfo  {journal} {Phys. Rev. E}\ }\textbf {\bibinfo {volume} {106}},\ \bibinfo {pages} {L032101} (\bibinfo {year} {2022})}\BibitemShut {NoStop}%
\bibitem [{\citenamefont {Van~Vu}\ and\ \citenamefont {Saito}(2023{\natexlab{a}})}]{Vu.2023.PRL.TSL}%
  \BibitemOpen
  \bibfield  {author} {\bibinfo {author} {\bibfnamefont {T.}~\bibnamefont {Van~Vu}}\ and\ \bibinfo {author} {\bibfnamefont {K.}~\bibnamefont {Saito}},\ }\bibfield  {title} {\bibinfo {title} {{Topological speed limit}},\ }\href {https://doi.org/10.1103/PhysRevLett.130.010402} {\bibfield  {journal} {\bibinfo  {journal} {Phys. Rev. Lett.}\ }\textbf {\bibinfo {volume} {130}},\ \bibinfo {pages} {010402} (\bibinfo {year} {2023}{\natexlab{a}})}\BibitemShut {NoStop}%
\bibitem [{\citenamefont {Kwon}\ \emph {et~al.}(2024)\citenamefont {Kwon}, \citenamefont {Park}, \citenamefont {Lee},\ and\ \citenamefont {Baek}}]{Kwon.2024.PRE}%
  \BibitemOpen
  \bibfield  {author} {\bibinfo {author} {\bibfnamefont {E.}~\bibnamefont {Kwon}}, \bibinfo {author} {\bibfnamefont {J.-M.}\ \bibnamefont {Park}}, \bibinfo {author} {\bibfnamefont {J.~S.}\ \bibnamefont {Lee}},\ and\ \bibinfo {author} {\bibfnamefont {Y.}~\bibnamefont {Baek}},\ }\bibfield  {title} {\bibinfo {title} {{Unified hierarchical relationship between thermodynamic tradeoff relations}},\ }\href {https://doi.org/10.1103/PhysRevE.110.044131} {\bibfield  {journal} {\bibinfo  {journal} {Phys. Rev. E}\ }\textbf {\bibinfo {volume} {110}},\ \bibinfo {pages} {044131} (\bibinfo {year} {2024})}\BibitemShut {NoStop}%
\bibitem [{\citenamefont {Delvenne}\ and\ \citenamefont {Falasco}(2024)}]{Delvenne.2024.PRE}%
  \BibitemOpen
  \bibfield  {author} {\bibinfo {author} {\bibfnamefont {J.-C.}\ \bibnamefont {Delvenne}}\ and\ \bibinfo {author} {\bibfnamefont {G.}~\bibnamefont {Falasco}},\ }\bibfield  {title} {\bibinfo {title} {{Thermokinetic relations}},\ }\href {https://doi.org/10.1103/PhysRevE.109.014109} {\bibfield  {journal} {\bibinfo  {journal} {Phys. Rev. E}\ }\textbf {\bibinfo {volume} {109}},\ \bibinfo {pages} {014109} (\bibinfo {year} {2024})}\BibitemShut {NoStop}%
\bibitem [{\citenamefont {Landauer}(1961)}]{Landauer.1961.JRD}%
  \BibitemOpen
  \bibfield  {author} {\bibinfo {author} {\bibfnamefont {R.}~\bibnamefont {Landauer}},\ }\bibfield  {title} {\bibinfo {title} {{Irreversibility and heat generation in the computing process}},\ }\href {https://doi.org/10.1147/rd.53.0183} {\bibfield  {journal} {\bibinfo  {journal} {IBM J. Res. Dev.}\ }\textbf {\bibinfo {volume} {5}},\ \bibinfo {pages} {183} (\bibinfo {year} {1961})}\BibitemShut {NoStop}%
\bibitem [{\citenamefont {Goold}\ \emph {et~al.}(2015)\citenamefont {Goold}, \citenamefont {Paternostro},\ and\ \citenamefont {Modi}}]{Goold.2015.PRL}%
  \BibitemOpen
  \bibfield  {author} {\bibinfo {author} {\bibfnamefont {J.}~\bibnamefont {Goold}}, \bibinfo {author} {\bibfnamefont {M.}~\bibnamefont {Paternostro}},\ and\ \bibinfo {author} {\bibfnamefont {K.}~\bibnamefont {Modi}},\ }\bibfield  {title} {\bibinfo {title} {{Nonequilibrium quantum Landauer principle}},\ }\href {https://doi.org/10.1103/PhysRevLett.114.060602} {\bibfield  {journal} {\bibinfo  {journal} {Phys. Rev. Lett.}\ }\textbf {\bibinfo {volume} {114}},\ \bibinfo {pages} {060602} (\bibinfo {year} {2015})}\BibitemShut {NoStop}%
\bibitem [{\citenamefont {Proesmans}\ \emph {et~al.}(2020)\citenamefont {Proesmans}, \citenamefont {Ehrich},\ and\ \citenamefont {Bechhoefer}}]{Proesmans.2020.PRL}%
  \BibitemOpen
  \bibfield  {author} {\bibinfo {author} {\bibfnamefont {K.}~\bibnamefont {Proesmans}}, \bibinfo {author} {\bibfnamefont {J.}~\bibnamefont {Ehrich}},\ and\ \bibinfo {author} {\bibfnamefont {J.}~\bibnamefont {Bechhoefer}},\ }\bibfield  {title} {\bibinfo {title} {{Finite-time Landauer principle}},\ }\href {https://doi.org/10.1103/PhysRevLett.125.100602} {\bibfield  {journal} {\bibinfo  {journal} {Phys. Rev. Lett.}\ }\textbf {\bibinfo {volume} {125}},\ \bibinfo {pages} {100602} (\bibinfo {year} {2020})}\BibitemShut {NoStop}%
\bibitem [{\citenamefont {Zhen}\ \emph {et~al.}(2021)\citenamefont {Zhen}, \citenamefont {Egloff}, \citenamefont {Modi},\ and\ \citenamefont {Dahlsten}}]{Zhen.2021.PRL}%
  \BibitemOpen
  \bibfield  {author} {\bibinfo {author} {\bibfnamefont {Y.-Z.}\ \bibnamefont {Zhen}}, \bibinfo {author} {\bibfnamefont {D.}~\bibnamefont {Egloff}}, \bibinfo {author} {\bibfnamefont {K.}~\bibnamefont {Modi}},\ and\ \bibinfo {author} {\bibfnamefont {O.}~\bibnamefont {Dahlsten}},\ }\bibfield  {title} {\bibinfo {title} {{Universal bound on energy cost of bit reset in finite time}},\ }\href {https://doi.org/10.1103/PhysRevLett.127.190602} {\bibfield  {journal} {\bibinfo  {journal} {Phys. Rev. Lett.}\ }\textbf {\bibinfo {volume} {127}},\ \bibinfo {pages} {190602} (\bibinfo {year} {2021})}\BibitemShut {NoStop}%
\bibitem [{\citenamefont {Van~Vu}\ and\ \citenamefont {Saito}(2022{\natexlab{b}})}]{Vu.2022.PRL}%
  \BibitemOpen
  \bibfield  {author} {\bibinfo {author} {\bibfnamefont {T.}~\bibnamefont {Van~Vu}}\ and\ \bibinfo {author} {\bibfnamefont {K.}~\bibnamefont {Saito}},\ }\bibfield  {title} {\bibinfo {title} {{Finite-time quantum Landauer principle and quantum coherence}},\ }\href {https://doi.org/10.1103/PhysRevLett.128.010602} {\bibfield  {journal} {\bibinfo  {journal} {Phys. Rev. Lett.}\ }\textbf {\bibinfo {volume} {128}},\ \bibinfo {pages} {010602} (\bibinfo {year} {2022}{\natexlab{b}})}\BibitemShut {NoStop}%
\bibitem [{\citenamefont {Lee}\ \emph {et~al.}(2022)\citenamefont {Lee}, \citenamefont {Lee}, \citenamefont {Kwon},\ and\ \citenamefont {Park}}]{Lee.2022.PRL}%
  \BibitemOpen
  \bibfield  {author} {\bibinfo {author} {\bibfnamefont {J.~S.}\ \bibnamefont {Lee}}, \bibinfo {author} {\bibfnamefont {S.}~\bibnamefont {Lee}}, \bibinfo {author} {\bibfnamefont {H.}~\bibnamefont {Kwon}},\ and\ \bibinfo {author} {\bibfnamefont {H.}~\bibnamefont {Park}},\ }\bibfield  {title} {\bibinfo {title} {{Speed limit for a highly irreversible process and tight finite-time Landauer's bound}},\ }\href {https://doi.org/10.1103/PhysRevLett.129.120603} {\bibfield  {journal} {\bibinfo  {journal} {Phys. Rev. Lett.}\ }\textbf {\bibinfo {volume} {129}},\ \bibinfo {pages} {120603} (\bibinfo {year} {2022})}\BibitemShut {NoStop}%
\bibitem [{\citenamefont {Van~Vu}\ and\ \citenamefont {Saito}(2023{\natexlab{b}})}]{Vu.2023.PRX}%
  \BibitemOpen
  \bibfield  {author} {\bibinfo {author} {\bibfnamefont {T.}~\bibnamefont {Van~Vu}}\ and\ \bibinfo {author} {\bibfnamefont {K.}~\bibnamefont {Saito}},\ }\bibfield  {title} {\bibinfo {title} {{Thermodynamic unification of optimal transport: Thermodynamic uncertainty relation, minimum dissipation, and thermodynamic speed limits}},\ }\href {https://doi.org/10.1103/PhysRevX.13.011013} {\bibfield  {journal} {\bibinfo  {journal} {Phys. Rev. X}\ }\textbf {\bibinfo {volume} {13}},\ \bibinfo {pages} {011013} (\bibinfo {year} {2023}{\natexlab{b}})}\BibitemShut {NoStop}%
\bibitem [{\citenamefont {Rolandi}\ \emph {et~al.}(2023)\citenamefont {Rolandi}, \citenamefont {Abiuso},\ and\ \citenamefont {Perarnau-Llobet}}]{Rolandi.2023.PRL}%
  \BibitemOpen
  \bibfield  {author} {\bibinfo {author} {\bibfnamefont {A.}~\bibnamefont {Rolandi}}, \bibinfo {author} {\bibfnamefont {P.}~\bibnamefont {Abiuso}},\ and\ \bibinfo {author} {\bibfnamefont {M.}~\bibnamefont {Perarnau-Llobet}},\ }\bibfield  {title} {\bibinfo {title} {{Collective advantages in finite-time thermodynamics}},\ }\href {https://doi.org/10.1103/PhysRevLett.131.210401} {\bibfield  {journal} {\bibinfo  {journal} {Phys. Rev. Lett.}\ }\textbf {\bibinfo {volume} {131}},\ \bibinfo {pages} {210401} (\bibinfo {year} {2023})}\BibitemShut {NoStop}%
\bibitem [{\citenamefont {Falasco}\ and\ \citenamefont {Esposito}(2020)}]{Falasco.2020.PRL}%
  \BibitemOpen
  \bibfield  {author} {\bibinfo {author} {\bibfnamefont {G.}~\bibnamefont {Falasco}}\ and\ \bibinfo {author} {\bibfnamefont {M.}~\bibnamefont {Esposito}},\ }\bibfield  {title} {\bibinfo {title} {{Dissipation-time uncertainty relation}},\ }\href {https://doi.org/10.1103/PhysRevLett.125.120604} {\bibfield  {journal} {\bibinfo  {journal} {Phys. Rev. Lett.}\ }\textbf {\bibinfo {volume} {125}},\ \bibinfo {pages} {120604} (\bibinfo {year} {2020})}\BibitemShut {NoStop}%
\bibitem [{\citenamefont {Neri}(2022)}]{Neri.2022.SP}%
  \BibitemOpen
  \bibfield  {author} {\bibinfo {author} {\bibfnamefont {I.}~\bibnamefont {Neri}},\ }\bibfield  {title} {\bibinfo {title} {{Universal tradeoff relation between speed, uncertainty, and dissipation in nonequilibrium stationary states}},\ }\href {https://doi.org/10.21468/scipostphys.12.4.139} {\bibfield  {journal} {\bibinfo  {journal} {SciPost Phys.}\ }\textbf {\bibinfo {volume} {12}} (\bibinfo {year} {2022})}\BibitemShut {NoStop}%
\bibitem [{\citenamefont {Nicholson}\ \emph {et~al.}(2020)\citenamefont {Nicholson}, \citenamefont {Garc{\'{\i}}a-Pintos}, \citenamefont {del Campo},\ and\ \citenamefont {Green}}]{Nicholson.2020.NP}%
  \BibitemOpen
  \bibfield  {author} {\bibinfo {author} {\bibfnamefont {S.~B.}\ \bibnamefont {Nicholson}}, \bibinfo {author} {\bibfnamefont {L.~P.}\ \bibnamefont {Garc{\'{\i}}a-Pintos}}, \bibinfo {author} {\bibfnamefont {A.}~\bibnamefont {del Campo}},\ and\ \bibinfo {author} {\bibfnamefont {J.~R.}\ \bibnamefont {Green}},\ }\bibfield  {title} {\bibinfo {title} {{Time-information uncertainty relations in thermodynamics}},\ }\href {https://doi.org/10.1038/s41567-020-0981-y} {\bibfield  {journal} {\bibinfo  {journal} {Nat. Phys.}\ }\textbf {\bibinfo {volume} {16}},\ \bibinfo {pages} {1211} (\bibinfo {year} {2020})}\BibitemShut {NoStop}%
\bibitem [{\citenamefont {Garc\'{\i}a-Pintos}\ \emph {et~al.}(2022)\citenamefont {Garc\'{\i}a-Pintos}, \citenamefont {Nicholson}, \citenamefont {Green}, \citenamefont {del Campo},\ and\ \citenamefont {Gorshkov}}]{Pintos.2022.PRX}%
  \BibitemOpen
  \bibfield  {author} {\bibinfo {author} {\bibfnamefont {L.~P.}\ \bibnamefont {Garc\'{\i}a-Pintos}}, \bibinfo {author} {\bibfnamefont {S.~B.}\ \bibnamefont {Nicholson}}, \bibinfo {author} {\bibfnamefont {J.~R.}\ \bibnamefont {Green}}, \bibinfo {author} {\bibfnamefont {A.}~\bibnamefont {del Campo}},\ and\ \bibinfo {author} {\bibfnamefont {A.~V.}\ \bibnamefont {Gorshkov}},\ }\bibfield  {title} {\bibinfo {title} {{Unifying quantum and classical speed limits on observables}},\ }\href {https://doi.org/10.1103/PhysRevX.12.011038} {\bibfield  {journal} {\bibinfo  {journal} {Phys. Rev. X}\ }\textbf {\bibinfo {volume} {12}},\ \bibinfo {pages} {011038} (\bibinfo {year} {2022})}\BibitemShut {NoStop}%
\bibitem [{\citenamefont {Bennett}(1982)}]{Bennett.1982.IJTP}%
  \BibitemOpen
  \bibfield  {author} {\bibinfo {author} {\bibfnamefont {C.~H.}\ \bibnamefont {Bennett}},\ }\bibfield  {title} {\bibinfo {title} {{The thermodynamics of computation{\textemdash}a review}},\ }\href {https://doi.org/10.1007/bf02084158} {\bibfield  {journal} {\bibinfo  {journal} {Int. J. Theor. Phys.}\ }\textbf {\bibinfo {volume} {21}},\ \bibinfo {pages} {905} (\bibinfo {year} {1982})}\BibitemShut {NoStop}%
\bibitem [{\citenamefont {Rio}\ \emph {et~al.}(2011)\citenamefont {Rio}, \citenamefont {Åberg}, \citenamefont {Renner}, \citenamefont {Dahlsten},\ and\ \citenamefont {Vedral}}]{Rio.2011.N}%
  \BibitemOpen
  \bibfield  {author} {\bibinfo {author} {\bibfnamefont {L.~d.}\ \bibnamefont {Rio}}, \bibinfo {author} {\bibfnamefont {J.}~\bibnamefont {Åberg}}, \bibinfo {author} {\bibfnamefont {R.}~\bibnamefont {Renner}}, \bibinfo {author} {\bibfnamefont {O.}~\bibnamefont {Dahlsten}},\ and\ \bibinfo {author} {\bibfnamefont {V.}~\bibnamefont {Vedral}},\ }\bibfield  {title} {\bibinfo {title} {{The thermodynamic meaning of negative entropy}},\ }\href {https://doi.org/10.1038/nature10123} {\bibfield  {journal} {\bibinfo  {journal} {Nature}\ }\textbf {\bibinfo {volume} {474}},\ \bibinfo {pages} {61} (\bibinfo {year} {2011})}\BibitemShut {NoStop}%
\bibitem [{\citenamefont {Dillenschneider}\ and\ \citenamefont {Lutz}(2009)}]{Dillenschneider.2009.PRL}%
  \BibitemOpen
  \bibfield  {author} {\bibinfo {author} {\bibfnamefont {R.}~\bibnamefont {Dillenschneider}}\ and\ \bibinfo {author} {\bibfnamefont {E.}~\bibnamefont {Lutz}},\ }\bibfield  {title} {\bibinfo {title} {{Memory erasure in small systems}},\ }\href {https://doi.org/10.1103/PhysRevLett.102.210601} {\bibfield  {journal} {\bibinfo  {journal} {Phys. Rev. Lett.}\ }\textbf {\bibinfo {volume} {102}},\ \bibinfo {pages} {210601} (\bibinfo {year} {2009})}\BibitemShut {NoStop}%
\bibitem [{\citenamefont {Munson}\ \emph {et~al.}(2025)\citenamefont {Munson}, \citenamefont {Kothakonda}, \citenamefont {Haferkamp}, \citenamefont {Yunger~Halpern}, \citenamefont {Eisert},\ and\ \citenamefont {Faist}}]{Munson.2025.PRXQ}%
  \BibitemOpen
  \bibfield  {author} {\bibinfo {author} {\bibfnamefont {A.}~\bibnamefont {Munson}}, \bibinfo {author} {\bibfnamefont {N.~B.~T.}\ \bibnamefont {Kothakonda}}, \bibinfo {author} {\bibfnamefont {J.}~\bibnamefont {Haferkamp}}, \bibinfo {author} {\bibfnamefont {N.}~\bibnamefont {Yunger~Halpern}}, \bibinfo {author} {\bibfnamefont {J.}~\bibnamefont {Eisert}},\ and\ \bibinfo {author} {\bibfnamefont {P.}~\bibnamefont {Faist}},\ }\bibfield  {title} {\bibinfo {title} {{Complexity-constrained quantum thermodynamics}},\ }\href {https://doi.org/10.1103/PRXQuantum.6.010346} {\bibfield  {journal} {\bibinfo  {journal} {PRX Quantum}\ }\textbf {\bibinfo {volume} {6}},\ \bibinfo {pages} {010346} (\bibinfo {year} {2025})}\BibitemShut {NoStop}%
\bibitem [{\citenamefont {Schulman}\ \emph {et~al.}(2005)\citenamefont {Schulman}, \citenamefont {Mor},\ and\ \citenamefont {Weinstein}}]{Schulman.2005.PRL}%
  \BibitemOpen
  \bibfield  {author} {\bibinfo {author} {\bibfnamefont {L.~J.}\ \bibnamefont {Schulman}}, \bibinfo {author} {\bibfnamefont {T.}~\bibnamefont {Mor}},\ and\ \bibinfo {author} {\bibfnamefont {Y.}~\bibnamefont {Weinstein}},\ }\bibfield  {title} {\bibinfo {title} {{Physical limits of heat-bath algorithmic cooling}},\ }\href {https://doi.org/10.1103/PhysRevLett.94.120501} {\bibfield  {journal} {\bibinfo  {journal} {Phys. Rev. Lett.}\ }\textbf {\bibinfo {volume} {94}},\ \bibinfo {pages} {120501} (\bibinfo {year} {2005})}\BibitemShut {NoStop}%
\bibitem [{\citenamefont {Chan}\ \emph {et~al.}(2011)\citenamefont {Chan}, \citenamefont {Alegre}, \citenamefont {Safavi-Naeini}, \citenamefont {Hill}, \citenamefont {Krause}, \citenamefont {Gr\"{o}blacher}, \citenamefont {Aspelmeyer},\ and\ \citenamefont {Painter}}]{Chan.2011.N}%
  \BibitemOpen
  \bibfield  {author} {\bibinfo {author} {\bibfnamefont {J.}~\bibnamefont {Chan}}, \bibinfo {author} {\bibfnamefont {T.~P.~M.}\ \bibnamefont {Alegre}}, \bibinfo {author} {\bibfnamefont {A.~H.}\ \bibnamefont {Safavi-Naeini}}, \bibinfo {author} {\bibfnamefont {J.~T.}\ \bibnamefont {Hill}}, \bibinfo {author} {\bibfnamefont {A.}~\bibnamefont {Krause}}, \bibinfo {author} {\bibfnamefont {S.}~\bibnamefont {Gr\"{o}blacher}}, \bibinfo {author} {\bibfnamefont {M.}~\bibnamefont {Aspelmeyer}},\ and\ \bibinfo {author} {\bibfnamefont {O.}~\bibnamefont {Painter}},\ }\bibfield  {title} {\bibinfo {title} {{Laser cooling of a nanomechanical oscillator into its quantum ground state}},\ }\href {https://doi.org/10.1038/nature10461} {\bibfield  {journal} {\bibinfo  {journal} {Nature}\ }\textbf {\bibinfo {volume} {478}},\ \bibinfo {pages} {89} (\bibinfo {year} {2011})}\BibitemShut {NoStop}%
\bibitem [{\citenamefont {Teufel}\ \emph {et~al.}(2011)\citenamefont {Teufel}, \citenamefont {Donner}, \citenamefont {Li}, \citenamefont {Harlow}, \citenamefont {Allman}, \citenamefont {Cicak}, \citenamefont {Sirois}, \citenamefont {Whittaker}, \citenamefont {Lehnert},\ and\ \citenamefont {Simmonds}}]{Teufel.2011.N}%
  \BibitemOpen
  \bibfield  {author} {\bibinfo {author} {\bibfnamefont {J.~D.}\ \bibnamefont {Teufel}}, \bibinfo {author} {\bibfnamefont {T.}~\bibnamefont {Donner}}, \bibinfo {author} {\bibfnamefont {D.}~\bibnamefont {Li}}, \bibinfo {author} {\bibfnamefont {J.~W.}\ \bibnamefont {Harlow}}, \bibinfo {author} {\bibfnamefont {M.~S.}\ \bibnamefont {Allman}}, \bibinfo {author} {\bibfnamefont {K.}~\bibnamefont {Cicak}}, \bibinfo {author} {\bibfnamefont {A.~J.}\ \bibnamefont {Sirois}}, \bibinfo {author} {\bibfnamefont {J.~D.}\ \bibnamefont {Whittaker}}, \bibinfo {author} {\bibfnamefont {K.~W.}\ \bibnamefont {Lehnert}},\ and\ \bibinfo {author} {\bibfnamefont {R.~W.}\ \bibnamefont {Simmonds}},\ }\bibfield  {title} {\bibinfo {title} {{Sideband cooling of micromechanical motion to the quantum ground state}},\ }\href {https://doi.org/10.1038/nature10261} {\bibfield  {journal} {\bibinfo  {journal} {Nature}\ }\textbf {\bibinfo {volume} {475}},\ \bibinfo {pages} {359} (\bibinfo {year} {2011})}\BibitemShut {NoStop}%
\bibitem [{\citenamefont {Peterson}\ \emph {et~al.}(2016)\citenamefont {Peterson}, \citenamefont {Purdy}, \citenamefont {Kampel}, \citenamefont {Andrews}, \citenamefont {Yu}, \citenamefont {Lehnert},\ and\ \citenamefont {Regal}}]{Peterson.2016.PRL}%
  \BibitemOpen
  \bibfield  {author} {\bibinfo {author} {\bibfnamefont {R.~W.}\ \bibnamefont {Peterson}}, \bibinfo {author} {\bibfnamefont {T.~P.}\ \bibnamefont {Purdy}}, \bibinfo {author} {\bibfnamefont {N.~S.}\ \bibnamefont {Kampel}}, \bibinfo {author} {\bibfnamefont {R.~W.}\ \bibnamefont {Andrews}}, \bibinfo {author} {\bibfnamefont {P.-L.}\ \bibnamefont {Yu}}, \bibinfo {author} {\bibfnamefont {K.~W.}\ \bibnamefont {Lehnert}},\ and\ \bibinfo {author} {\bibfnamefont {C.~A.}\ \bibnamefont {Regal}},\ }\bibfield  {title} {\bibinfo {title} {{Laser cooling of a micromechanical membrane to the quantum backaction limit}},\ }\href {https://doi.org/10.1103/PhysRevLett.116.063601} {\bibfield  {journal} {\bibinfo  {journal} {Phys. Rev. Lett.}\ }\textbf {\bibinfo {volume} {116}},\ \bibinfo {pages} {063601} (\bibinfo {year} {2016})}\BibitemShut {NoStop}%
\bibitem [{\citenamefont {Guo}\ \emph {et~al.}(2019)\citenamefont {Guo}, \citenamefont {Norte},\ and\ \citenamefont {Gr\"oblacher}}]{Guo.2019.PRL}%
  \BibitemOpen
  \bibfield  {author} {\bibinfo {author} {\bibfnamefont {J.}~\bibnamefont {Guo}}, \bibinfo {author} {\bibfnamefont {R.}~\bibnamefont {Norte}},\ and\ \bibinfo {author} {\bibfnamefont {S.}~\bibnamefont {Gr\"oblacher}},\ }\bibfield  {title} {\bibinfo {title} {{Feedback cooling of a room temperature mechanical oscillator close to its motional ground state}},\ }\href {https://doi.org/10.1103/PhysRevLett.123.223602} {\bibfield  {journal} {\bibinfo  {journal} {Phys. Rev. Lett.}\ }\textbf {\bibinfo {volume} {123}},\ \bibinfo {pages} {223602} (\bibinfo {year} {2019})}\BibitemShut {NoStop}%
\bibitem [{\citenamefont {Soldati}\ \emph {et~al.}(2022)\citenamefont {Soldati}, \citenamefont {Dasari}, \citenamefont {Wrachtrup},\ and\ \citenamefont {Lutz}}]{Soldati.2022.PRL}%
  \BibitemOpen
  \bibfield  {author} {\bibinfo {author} {\bibfnamefont {R.~R.}\ \bibnamefont {Soldati}}, \bibinfo {author} {\bibfnamefont {D.~B.~R.}\ \bibnamefont {Dasari}}, \bibinfo {author} {\bibfnamefont {J.}~\bibnamefont {Wrachtrup}},\ and\ \bibinfo {author} {\bibfnamefont {E.}~\bibnamefont {Lutz}},\ }\bibfield  {title} {\bibinfo {title} {{Thermodynamics of a minimal algorithmic cooling refrigerator}},\ }\href {https://doi.org/10.1103/PhysRevLett.129.030601} {\bibfield  {journal} {\bibinfo  {journal} {Phys. Rev. Lett.}\ }\textbf {\bibinfo {volume} {129}},\ \bibinfo {pages} {030601} (\bibinfo {year} {2022})}\BibitemShut {NoStop}%
\bibitem [{\citenamefont {Taranto}\ \emph {et~al.}(2025)\citenamefont {Taranto}, \citenamefont {Lipka-Bartosik}, \citenamefont {Rodr\'{\i}guez-Briones}, \citenamefont {Perarnau-Llobet}, \citenamefont {Friis}, \citenamefont {Huber},\ and\ \citenamefont {Bakhshinezhad}}]{Taranto.2025.PRL}%
  \BibitemOpen
  \bibfield  {author} {\bibinfo {author} {\bibfnamefont {P.}~\bibnamefont {Taranto}}, \bibinfo {author} {\bibfnamefont {P.}~\bibnamefont {Lipka-Bartosik}}, \bibinfo {author} {\bibfnamefont {N.~A.}\ \bibnamefont {Rodr\'{\i}guez-Briones}}, \bibinfo {author} {\bibfnamefont {M.}~\bibnamefont {Perarnau-Llobet}}, \bibinfo {author} {\bibfnamefont {N.}~\bibnamefont {Friis}}, \bibinfo {author} {\bibfnamefont {M.}~\bibnamefont {Huber}},\ and\ \bibinfo {author} {\bibfnamefont {P.}~\bibnamefont {Bakhshinezhad}},\ }\bibfield  {title} {\bibinfo {title} {{Efficiently cooling quantum systems with finite resources: Insights from thermodynamic geometry}},\ }\href {https://doi.org/10.1103/PhysRevLett.134.070401} {\bibfield  {journal} {\bibinfo  {journal} {Phys. Rev. Lett.}\ }\textbf {\bibinfo {volume} {134}},\ \bibinfo {pages} {070401} (\bibinfo {year} {2025})}\BibitemShut {NoStop}%
\bibitem [{\citenamefont {Hopfield}(1974)}]{Hopfield.1974.PNAS}%
  \BibitemOpen
  \bibfield  {author} {\bibinfo {author} {\bibfnamefont {J.~J.}\ \bibnamefont {Hopfield}},\ }\bibfield  {title} {\bibinfo {title} {{Kinetic proofreading: A new mechanism for reducing errors in biosynthetic processes requiring high specificity}},\ }\href {https://doi.org/10.1073/pnas.71.10.4135} {\bibfield  {journal} {\bibinfo  {journal} {Proc. Natl. Acad. Sci. U.S.A}\ }\textbf {\bibinfo {volume} {71}},\ \bibinfo {pages} {4135} (\bibinfo {year} {1974})}\BibitemShut {NoStop}%
\bibitem [{\citenamefont {Bennett}(1979)}]{Bennett.1979.B}%
  \BibitemOpen
  \bibfield  {author} {\bibinfo {author} {\bibfnamefont {C.~H.}\ \bibnamefont {Bennett}},\ }\bibfield  {title} {\bibinfo {title} {{Dissipation-error tradeoff in proofreading}},\ }\href {https://doi.org/10.1016/0303-2647(79)90003-0} {\bibfield  {journal} {\bibinfo  {journal} {Biosystems}\ }\textbf {\bibinfo {volume} {11}},\ \bibinfo {pages} {85} (\bibinfo {year} {1979})}\BibitemShut {NoStop}%
\bibitem [{\citenamefont {Sartori}\ and\ \citenamefont {Pigolotti}(2015)}]{Sartori.2015.PRX}%
  \BibitemOpen
  \bibfield  {author} {\bibinfo {author} {\bibfnamefont {P.}~\bibnamefont {Sartori}}\ and\ \bibinfo {author} {\bibfnamefont {S.}~\bibnamefont {Pigolotti}},\ }\bibfield  {title} {\bibinfo {title} {{Thermodynamics of error correction}},\ }\href {https://doi.org/10.1103/PhysRevX.5.041039} {\bibfield  {journal} {\bibinfo  {journal} {Phys. Rev. X}\ }\textbf {\bibinfo {volume} {5}},\ \bibinfo {pages} {041039} (\bibinfo {year} {2015})}\BibitemShut {NoStop}%
\bibitem [{\citenamefont {Di~Franco}\ and\ \citenamefont {Paternostro}(2013)}]{Franco.2013.SR}%
  \BibitemOpen
  \bibfield  {author} {\bibinfo {author} {\bibfnamefont {C.}~\bibnamefont {Di~Franco}}\ and\ \bibinfo {author} {\bibfnamefont {M.}~\bibnamefont {Paternostro}},\ }\bibfield  {title} {\bibinfo {title} {{A no-go result on the purification of quantum states}},\ }\href {https://doi.org/10.1038/srep01387} {\bibfield  {journal} {\bibinfo  {journal} {Sci. Rep.}\ }\textbf {\bibinfo {volume} {3}} (\bibinfo {year} {2013})}\BibitemShut {NoStop}%
\bibitem [{\citenamefont {Wu}\ \emph {et~al.}(2013)\citenamefont {Wu}, \citenamefont {Segal},\ and\ \citenamefont {Brumer}}]{Wu.2013.SR}%
  \BibitemOpen
  \bibfield  {author} {\bibinfo {author} {\bibfnamefont {L.-A.}\ \bibnamefont {Wu}}, \bibinfo {author} {\bibfnamefont {D.}~\bibnamefont {Segal}},\ and\ \bibinfo {author} {\bibfnamefont {P.}~\bibnamefont {Brumer}},\ }\bibfield  {title} {\bibinfo {title} {{No-go theorem for ground state cooling given initial system-thermal bath factorization}},\ }\href {https://doi.org/10.1038/srep01824} {\bibfield  {journal} {\bibinfo  {journal} {Sci. Rep.}\ }\textbf {\bibinfo {volume} {3}} (\bibinfo {year} {2013})}\BibitemShut {NoStop}%
\bibitem [{\citenamefont {Ticozzi}\ and\ \citenamefont {Viola}(2014)}]{Ticozzi.2014.SR}%
  \BibitemOpen
  \bibfield  {author} {\bibinfo {author} {\bibfnamefont {F.}~\bibnamefont {Ticozzi}}\ and\ \bibinfo {author} {\bibfnamefont {L.}~\bibnamefont {Viola}},\ }\bibfield  {title} {\bibinfo {title} {{Quantum resources for purification and cooling: fundamental limits and opportunities}},\ }\href {https://doi.org/10.1038/srep05192} {\bibfield  {journal} {\bibinfo  {journal} {Sci. Rep.}\ }\textbf {\bibinfo {volume} {4}} (\bibinfo {year} {2014})}\BibitemShut {NoStop}%
\bibitem [{\citenamefont {Masanes}\ and\ \citenamefont {Oppenheim}(2017)}]{Masanes.2017.NC}%
  \BibitemOpen
  \bibfield  {author} {\bibinfo {author} {\bibfnamefont {L.}~\bibnamefont {Masanes}}\ and\ \bibinfo {author} {\bibfnamefont {J.}~\bibnamefont {Oppenheim}},\ }\bibfield  {title} {\bibinfo {title} {{A general derivation and quantification of the third law of thermodynamics}},\ }\href {https://doi.org/10.1038/ncomms14538} {\bibfield  {journal} {\bibinfo  {journal} {Nat. Commun.}\ }\textbf {\bibinfo {volume} {8}},\ \bibinfo {pages} {14538} (\bibinfo {year} {2017})}\BibitemShut {NoStop}%
\bibitem [{\citenamefont {Wilming}\ and\ \citenamefont {Gallego}(2017)}]{Wilming.2017.PRX}%
  \BibitemOpen
  \bibfield  {author} {\bibinfo {author} {\bibfnamefont {H.}~\bibnamefont {Wilming}}\ and\ \bibinfo {author} {\bibfnamefont {R.}~\bibnamefont {Gallego}},\ }\bibfield  {title} {\bibinfo {title} {{Third law of thermodynamics as a single inequality}},\ }\href {https://doi.org/10.1103/PhysRevX.7.041033} {\bibfield  {journal} {\bibinfo  {journal} {Phys. Rev. X}\ }\textbf {\bibinfo {volume} {7}},\ \bibinfo {pages} {041033} (\bibinfo {year} {2017})}\BibitemShut {NoStop}%
\bibitem [{\citenamefont {Scharlau}\ and\ \citenamefont {Mueller}(2018)}]{Scharlau.2018.Q}%
  \BibitemOpen
  \bibfield  {author} {\bibinfo {author} {\bibfnamefont {J.}~\bibnamefont {Scharlau}}\ and\ \bibinfo {author} {\bibfnamefont {M.~P.}\ \bibnamefont {Mueller}},\ }\bibfield  {title} {\bibinfo {title} {{Quantum Horn's lemma, finite heat baths, and the third law of thermodynamics}},\ }\href {https://doi.org/10.22331/q-2018-02-22-54} {\bibfield  {journal} {\bibinfo  {journal} {Quantum}\ }\textbf {\bibinfo {volume} {2}},\ \bibinfo {pages} {54} (\bibinfo {year} {2018})}\BibitemShut {NoStop}%
\bibitem [{\citenamefont {Clivaz}\ \emph {et~al.}(2019)\citenamefont {Clivaz}, \citenamefont {Silva}, \citenamefont {Haack}, \citenamefont {Brask}, \citenamefont {Brunner},\ and\ \citenamefont {Huber}}]{Clivaz.2019.PRL}%
  \BibitemOpen
  \bibfield  {author} {\bibinfo {author} {\bibfnamefont {F.}~\bibnamefont {Clivaz}}, \bibinfo {author} {\bibfnamefont {R.}~\bibnamefont {Silva}}, \bibinfo {author} {\bibfnamefont {G.}~\bibnamefont {Haack}}, \bibinfo {author} {\bibfnamefont {J.~B.}\ \bibnamefont {Brask}}, \bibinfo {author} {\bibfnamefont {N.}~\bibnamefont {Brunner}},\ and\ \bibinfo {author} {\bibfnamefont {M.}~\bibnamefont {Huber}},\ }\bibfield  {title} {\bibinfo {title} {{Unifying paradigms of quantum refrigeration: A universal and attainable bound on cooling}},\ }\href {https://doi.org/10.1103/PhysRevLett.123.170605} {\bibfield  {journal} {\bibinfo  {journal} {Phys. Rev. Lett.}\ }\textbf {\bibinfo {volume} {123}},\ \bibinfo {pages} {170605} (\bibinfo {year} {2019})}\BibitemShut {NoStop}%
\bibitem [{\citenamefont {Levy}\ \emph {et~al.}(2012)\citenamefont {Levy}, \citenamefont {Alicki},\ and\ \citenamefont {Kosloff}}]{Levy.2012.PRE}%
  \BibitemOpen
  \bibfield  {author} {\bibinfo {author} {\bibfnamefont {A.}~\bibnamefont {Levy}}, \bibinfo {author} {\bibfnamefont {R.}~\bibnamefont {Alicki}},\ and\ \bibinfo {author} {\bibfnamefont {R.}~\bibnamefont {Kosloff}},\ }\bibfield  {title} {\bibinfo {title} {{Quantum refrigerators and the third law of thermodynamics}},\ }\href {https://doi.org/10.1103/PhysRevE.85.061126} {\bibfield  {journal} {\bibinfo  {journal} {Phys. Rev. E}\ }\textbf {\bibinfo {volume} {85}},\ \bibinfo {pages} {061126} (\bibinfo {year} {2012})}\BibitemShut {NoStop}%
\bibitem [{\citenamefont {Kol\'a\ifmmode~\check{r}\else \v{r}\fi{}}\ \emph {et~al.}(2012)\citenamefont {Kol\'a\ifmmode~\check{r}\else \v{r}\fi{}}, \citenamefont {Gelbwaser-Klimovsky}, \citenamefont {Alicki},\ and\ \citenamefont {Kurizki}}]{Kolar.2012.PRL}%
  \BibitemOpen
  \bibfield  {author} {\bibinfo {author} {\bibfnamefont {M.}~\bibnamefont {Kol\'a\ifmmode~\check{r}\else \v{r}\fi{}}}, \bibinfo {author} {\bibfnamefont {D.}~\bibnamefont {Gelbwaser-Klimovsky}}, \bibinfo {author} {\bibfnamefont {R.}~\bibnamefont {Alicki}},\ and\ \bibinfo {author} {\bibfnamefont {G.}~\bibnamefont {Kurizki}},\ }\bibfield  {title} {\bibinfo {title} {{Quantum bath refrigeration towards absolute zero: Challenging the unattainability principle}},\ }\href {https://doi.org/10.1103/PhysRevLett.109.090601} {\bibfield  {journal} {\bibinfo  {journal} {Phys. Rev. Lett.}\ }\textbf {\bibinfo {volume} {109}},\ \bibinfo {pages} {090601} (\bibinfo {year} {2012})}\BibitemShut {NoStop}%
\bibitem [{\citenamefont {Reeb}\ and\ \citenamefont {Wolf}(2014)}]{Reeb.2014.NJP}%
  \BibitemOpen
  \bibfield  {author} {\bibinfo {author} {\bibfnamefont {D.}~\bibnamefont {Reeb}}\ and\ \bibinfo {author} {\bibfnamefont {M.~M.}\ \bibnamefont {Wolf}},\ }\bibfield  {title} {\bibinfo {title} {{An improved Landauer principle with finite-size corrections}},\ }\href {https://doi.org/10.1088/1367-2630/16/10/103011} {\bibfield  {journal} {\bibinfo  {journal} {New J. Phys.}\ }\textbf {\bibinfo {volume} {16}},\ \bibinfo {pages} {103011} (\bibinfo {year} {2014})}\BibitemShut {NoStop}%
\bibitem [{\citenamefont {Lieb}\ and\ \citenamefont {Robinson}(1972)}]{Lieb.1972.CMP}%
  \BibitemOpen
  \bibfield  {author} {\bibinfo {author} {\bibfnamefont {E.~H.}\ \bibnamefont {Lieb}}\ and\ \bibinfo {author} {\bibfnamefont {D.~W.}\ \bibnamefont {Robinson}},\ }\bibfield  {title} {\bibinfo {title} {{The finite group velocity of quantum spin systems}},\ }\href {https://doi.org/10.1007/bf01645779} {\bibfield  {journal} {\bibinfo  {journal} {Commun. Math. Phys.}\ }\textbf {\bibinfo {volume} {28}},\ \bibinfo {pages} {251} (\bibinfo {year} {1972})}\BibitemShut {NoStop}%
\bibitem [{\citenamefont {Landi}\ and\ \citenamefont {Paternostro}(2021)}]{Landi.2021.RMP}%
  \BibitemOpen
  \bibfield  {author} {\bibinfo {author} {\bibfnamefont {G.~T.}\ \bibnamefont {Landi}}\ and\ \bibinfo {author} {\bibfnamefont {M.}~\bibnamefont {Paternostro}},\ }\bibfield  {title} {\bibinfo {title} {{Irreversible entropy production: From classical to quantum}},\ }\href {https://doi.org/10.1103/RevModPhys.93.035008} {\bibfield  {journal} {\bibinfo  {journal} {Rev. Mod. Phys.}\ }\textbf {\bibinfo {volume} {93}},\ \bibinfo {pages} {035008} (\bibinfo {year} {2021})}\BibitemShut {NoStop}%
\bibitem [{\citenamefont {Maes}(2020)}]{Maes.2020.PR}%
  \BibitemOpen
  \bibfield  {author} {\bibinfo {author} {\bibfnamefont {C.}~\bibnamefont {Maes}},\ }\bibfield  {title} {\bibinfo {title} {{Frenesy: Time-symmetric dynamical activity in nonequilibria}},\ }\href {https://doi.org/10.1016/j.physrep.2020.01.002} {\bibfield  {journal} {\bibinfo  {journal} {Phys. Rep.}\ }\textbf {\bibinfo {volume} {850}},\ \bibinfo {pages} {1} (\bibinfo {year} {2020})}\BibitemShut {NoStop}%
\bibitem [{\citenamefont {Garrahan}(2017)}]{Garrahan.2017.PRE}%
  \BibitemOpen
  \bibfield  {author} {\bibinfo {author} {\bibfnamefont {J.~P.}\ \bibnamefont {Garrahan}},\ }\bibfield  {title} {\bibinfo {title} {{Simple bounds on fluctuations and uncertainty relations for first-passage times of counting observables}},\ }\href {https://doi.org/10.1103/PhysRevE.95.032134} {\bibfield  {journal} {\bibinfo  {journal} {Phys. Rev. E}\ }\textbf {\bibinfo {volume} {95}},\ \bibinfo {pages} {032134} (\bibinfo {year} {2017})}\BibitemShut {NoStop}%
\bibitem [{\citenamefont {Di~Terlizzi}\ and\ \citenamefont {Baiesi}(2019)}]{Terlizzi.2019.JPA}%
  \BibitemOpen
  \bibfield  {author} {\bibinfo {author} {\bibfnamefont {I.}~\bibnamefont {Di~Terlizzi}}\ and\ \bibinfo {author} {\bibfnamefont {M.}~\bibnamefont {Baiesi}},\ }\bibfield  {title} {\bibinfo {title} {{Kinetic uncertainty relation}},\ }\href {https://doi.org/10.1088/1751-8121/aaee34} {\bibfield  {journal} {\bibinfo  {journal} {J. Phys. A}\ }\textbf {\bibinfo {volume} {52}},\ \bibinfo {pages} {02LT03} (\bibinfo {year} {2019})}\BibitemShut {NoStop}%
\bibitem [{\citenamefont {Hasegawa}(2020)}]{Hasegawa.2020.PRL}%
  \BibitemOpen
  \bibfield  {author} {\bibinfo {author} {\bibfnamefont {Y.}~\bibnamefont {Hasegawa}},\ }\bibfield  {title} {\bibinfo {title} {{Quantum thermodynamic uncertainty relation for continuous measurement}},\ }\href {https://doi.org/10.1103/PhysRevLett.125.050601} {\bibfield  {journal} {\bibinfo  {journal} {Phys. Rev. Lett.}\ }\textbf {\bibinfo {volume} {125}},\ \bibinfo {pages} {050601} (\bibinfo {year} {2020})}\BibitemShut {NoStop}%
\bibitem [{\citenamefont {Vo}\ \emph {et~al.}(2022)\citenamefont {Vo}, \citenamefont {Van~Vu},\ and\ \citenamefont {Hasegawa}}]{Vo.2022.JPA}%
  \BibitemOpen
  \bibfield  {author} {\bibinfo {author} {\bibfnamefont {V.~T.}\ \bibnamefont {Vo}}, \bibinfo {author} {\bibfnamefont {T.}~\bibnamefont {Van~Vu}},\ and\ \bibinfo {author} {\bibfnamefont {Y.}~\bibnamefont {Hasegawa}},\ }\bibfield  {title} {\bibinfo {title} {{Unified thermodynamic-kinetic uncertainty relation}},\ }\href {https://doi.org/10.1088/1751-8121/ac9099} {\bibfield  {journal} {\bibinfo  {journal} {J. Phys. A}\ }\textbf {\bibinfo {volume} {55}},\ \bibinfo {pages} {405004} (\bibinfo {year} {2022})}\BibitemShut {NoStop}%
\bibitem [{fnt({\natexlab{a}})}]{fnt1}%
  \BibitemOpen
  \href@noop {} {}$c\approx 1.543$ is the solution of equation $z(1-e^{-z})=1+e^{-z}$.\BibitemShut {Stop}%
\bibitem [{\citenamefont {Van~Vu}\ and\ \citenamefont {Hasegawa}(2021{\natexlab{b}})}]{Vu.2021.PRL2}%
  \BibitemOpen
  \bibfield  {author} {\bibinfo {author} {\bibfnamefont {T.}~\bibnamefont {Van~Vu}}\ and\ \bibinfo {author} {\bibfnamefont {Y.}~\bibnamefont {Hasegawa}},\ }\bibfield  {title} {\bibinfo {title} {{Lower bound on irreversibility in thermal relaxation of open quantum systems}},\ }\href {https://doi.org/10.1103/PhysRevLett.127.190601} {\bibfield  {journal} {\bibinfo  {journal} {Phys. Rev. Lett.}\ }\textbf {\bibinfo {volume} {127}},\ \bibinfo {pages} {190601} (\bibinfo {year} {2021}{\natexlab{b}})}\BibitemShut {NoStop}%
\bibitem [{\citenamefont {Remlein}\ and\ \citenamefont {Seifert}(2021)}]{Remlein.2021.PRE}%
  \BibitemOpen
  \bibfield  {author} {\bibinfo {author} {\bibfnamefont {B.}~\bibnamefont {Remlein}}\ and\ \bibinfo {author} {\bibfnamefont {U.}~\bibnamefont {Seifert}},\ }\bibfield  {title} {\bibinfo {title} {{Optimality of nonconservative driving for finite-time processes with discrete states}},\ }\href {https://doi.org/10.1103/PhysRevE.103.L050105} {\bibfield  {journal} {\bibinfo  {journal} {Phys. Rev. E}\ }\textbf {\bibinfo {volume} {103}},\ \bibinfo {pages} {L050105} (\bibinfo {year} {2021})}\BibitemShut {NoStop}%
\bibitem [{\citenamefont {Dechant}(2022)}]{Dechant.2022.JPA}%
  \BibitemOpen
  \bibfield  {author} {\bibinfo {author} {\bibfnamefont {A.}~\bibnamefont {Dechant}},\ }\bibfield  {title} {\bibinfo {title} {{Minimum entropy production, detailed balance and Wasserstein distance for continuous-time Markov processes}},\ }\href {https://doi.org/10.1088/1751-8121/ac4ac0} {\bibfield  {journal} {\bibinfo  {journal} {J. Phys. A}\ }\textbf {\bibinfo {volume} {55}},\ \bibinfo {pages} {094001} (\bibinfo {year} {2022})}\BibitemShut {NoStop}%
\bibitem [{\citenamefont {Nernst}(1906)}]{Nernst.1906}%
  \BibitemOpen
  \bibfield  {author} {\bibinfo {author} {\bibfnamefont {W.}~\bibnamefont {Nernst}},\ }\bibfield  {title} {\bibinfo {title} {{Ueber die Berechnung chemischer Gleichgewichte aus thermischen Messungen}},\ }\href {http://eudml.org/doc/58630} {\bibfield  {journal} {\bibinfo  {journal} {Nachr. Kgl. Ges. Wiss. Goett.}\ }\textbf {\bibinfo {volume} {1906}},\ \bibinfo {pages} {1} (\bibinfo {year} {1906})}\BibitemShut {NoStop}%
\bibitem [{\citenamefont {Planck}(1911)}]{Planck.1911}%
  \BibitemOpen
  \bibfield  {author} {\bibinfo {author} {\bibfnamefont {M.}~\bibnamefont {Planck}},\ }\href@noop {} {\emph {\bibinfo {title} {{Thermodynamik}}}},\ \bibinfo {edition} {3rd}\ ed.\ (\bibinfo  {publisher} {De Gruyter},\ \bibinfo {year} {1911})\BibitemShut {NoStop}%
\bibitem [{\citenamefont {Allahverdyan}\ \emph {et~al.}(2011)\citenamefont {Allahverdyan}, \citenamefont {Hovhannisyan}, \citenamefont {Janzing},\ and\ \citenamefont {Mahler}}]{Allahverdyan.2011.PRE}%
  \BibitemOpen
  \bibfield  {author} {\bibinfo {author} {\bibfnamefont {A.~E.}\ \bibnamefont {Allahverdyan}}, \bibinfo {author} {\bibfnamefont {K.~V.}\ \bibnamefont {Hovhannisyan}}, \bibinfo {author} {\bibfnamefont {D.}~\bibnamefont {Janzing}},\ and\ \bibinfo {author} {\bibfnamefont {G.}~\bibnamefont {Mahler}},\ }\bibfield  {title} {\bibinfo {title} {{Thermodynamic limits of dynamic cooling}},\ }\href {https://doi.org/10.1103/PhysRevE.84.041109} {\bibfield  {journal} {\bibinfo  {journal} {Phys. Rev. E}\ }\textbf {\bibinfo {volume} {84}},\ \bibinfo {pages} {041109} (\bibinfo {year} {2011})}\BibitemShut {NoStop}%
\bibitem [{\citenamefont {Silva}\ \emph {et~al.}(2016)\citenamefont {Silva}, \citenamefont {Manzano}, \citenamefont {Skrzypczyk},\ and\ \citenamefont {Brunner}}]{Silva.2016.PRE}%
  \BibitemOpen
  \bibfield  {author} {\bibinfo {author} {\bibfnamefont {R.}~\bibnamefont {Silva}}, \bibinfo {author} {\bibfnamefont {G.}~\bibnamefont {Manzano}}, \bibinfo {author} {\bibfnamefont {P.}~\bibnamefont {Skrzypczyk}},\ and\ \bibinfo {author} {\bibfnamefont {N.}~\bibnamefont {Brunner}},\ }\bibfield  {title} {\bibinfo {title} {{Performance of autonomous quantum thermal machines: Hilbert space dimension as a thermodynamical resource}},\ }\href {https://doi.org/10.1103/PhysRevE.94.032120} {\bibfield  {journal} {\bibinfo  {journal} {Phys. Rev. E}\ }\textbf {\bibinfo {volume} {94}},\ \bibinfo {pages} {032120} (\bibinfo {year} {2016})}\BibitemShut {NoStop}%
\bibitem [{\citenamefont {Buffoni}\ \emph {et~al.}(2022)\citenamefont {Buffoni}, \citenamefont {Gherardini}, \citenamefont {Zambrini~Cruzeiro},\ and\ \citenamefont {Omar}}]{Buffoni.2022.PRL}%
  \BibitemOpen
  \bibfield  {author} {\bibinfo {author} {\bibfnamefont {L.}~\bibnamefont {Buffoni}}, \bibinfo {author} {\bibfnamefont {S.}~\bibnamefont {Gherardini}}, \bibinfo {author} {\bibfnamefont {E.}~\bibnamefont {Zambrini~Cruzeiro}},\ and\ \bibinfo {author} {\bibfnamefont {Y.}~\bibnamefont {Omar}},\ }\bibfield  {title} {\bibinfo {title} {{Third law of thermodynamics and the scaling of quantum computers}},\ }\href {https://doi.org/10.1103/PhysRevLett.129.150602} {\bibfield  {journal} {\bibinfo  {journal} {Phys. Rev. Lett.}\ }\textbf {\bibinfo {volume} {129}},\ \bibinfo {pages} {150602} (\bibinfo {year} {2022})}\BibitemShut {NoStop}%
\bibitem [{\citenamefont {Taranto}\ \emph {et~al.}(2023)\citenamefont {Taranto}, \citenamefont {Bakhshinezhad}, \citenamefont {Bluhm}, \citenamefont {Silva}, \citenamefont {Friis}, \citenamefont {Lock}, \citenamefont {Vitagliano}, \citenamefont {Binder}, \citenamefont {Debarba}, \citenamefont {Schwarzhans}, \citenamefont {Clivaz},\ and\ \citenamefont {Huber}}]{Taranto.2023.PRXQ}%
  \BibitemOpen
  \bibfield  {author} {\bibinfo {author} {\bibfnamefont {P.}~\bibnamefont {Taranto}}, \bibinfo {author} {\bibfnamefont {F.}~\bibnamefont {Bakhshinezhad}}, \bibinfo {author} {\bibfnamefont {A.}~\bibnamefont {Bluhm}}, \bibinfo {author} {\bibfnamefont {R.}~\bibnamefont {Silva}}, \bibinfo {author} {\bibfnamefont {N.}~\bibnamefont {Friis}}, \bibinfo {author} {\bibfnamefont {M.~P.}\ \bibnamefont {Lock}}, \bibinfo {author} {\bibfnamefont {G.}~\bibnamefont {Vitagliano}}, \bibinfo {author} {\bibfnamefont {F.~C.}\ \bibnamefont {Binder}}, \bibinfo {author} {\bibfnamefont {T.}~\bibnamefont {Debarba}}, \bibinfo {author} {\bibfnamefont {E.}~\bibnamefont {Schwarzhans}}, \bibinfo {author} {\bibfnamefont {F.}~\bibnamefont {Clivaz}},\ and\ \bibinfo {author} {\bibfnamefont {M.}~\bibnamefont {Huber}},\ }\bibfield  {title} {\bibinfo {title} {{Landauer versus Nernst: What is the true cost of cooling a quantum system?}},\ }\href {https://doi.org/10.1103/PRXQuantum.4.010332} {\bibfield  {journal} {\bibinfo  {journal} {PRX Quantum}\ }\textbf {\bibinfo {volume} {4}},\ \bibinfo {pages} {010332} (\bibinfo {year} {2023})}\BibitemShut {NoStop}%
\bibitem [{\citenamefont {Hatano}\ and\ \citenamefont {Sasa}(2001)}]{Hatano.2001.PRL}%
  \BibitemOpen
  \bibfield  {author} {\bibinfo {author} {\bibfnamefont {T.}~\bibnamefont {Hatano}}\ and\ \bibinfo {author} {\bibfnamefont {S.-i.}\ \bibnamefont {Sasa}},\ }\bibfield  {title} {\bibinfo {title} {{Steady-state thermodynamics of Langevin systems}},\ }\href {https://doi.org/10.1103/PhysRevLett.86.3463} {\bibfield  {journal} {\bibinfo  {journal} {Phys. Rev. Lett.}\ }\textbf {\bibinfo {volume} {86}},\ \bibinfo {pages} {3463} (\bibinfo {year} {2001})}\BibitemShut {NoStop}%
\bibitem [{\citenamefont {Esposito}\ and\ \citenamefont {Van~den Broeck}(2010)}]{Esposito.2010.PRL}%
  \BibitemOpen
  \bibfield  {author} {\bibinfo {author} {\bibfnamefont {M.}~\bibnamefont {Esposito}}\ and\ \bibinfo {author} {\bibfnamefont {C.}~\bibnamefont {Van~den Broeck}},\ }\bibfield  {title} {\bibinfo {title} {{Three detailed fluctuation theorems}},\ }\href {https://doi.org/10.1103/PhysRevLett.104.090601} {\bibfield  {journal} {\bibinfo  {journal} {Phys. Rev. Lett.}\ }\textbf {\bibinfo {volume} {104}},\ \bibinfo {pages} {090601} (\bibinfo {year} {2010})}\BibitemShut {NoStop}%
\bibitem [{fnt({\natexlab{b}})}]{fnt2}%
  \BibitemOpen
  \href@noop {} {}Even though the error in the final probability distribution can be incorporated into the TSL, it still cannot establish the third law of thermodynamics.\BibitemShut {Stop}%
\bibitem [{\citenamefont {Casas-Vázquez}\ and\ \citenamefont {Jou}(2003)}]{Vazquez.2003.RPP}%
  \BibitemOpen
  \bibfield  {author} {\bibinfo {author} {\bibfnamefont {J.}~\bibnamefont {Casas-Vázquez}}\ and\ \bibinfo {author} {\bibfnamefont {D.}~\bibnamefont {Jou}},\ }\bibfield  {title} {\bibinfo {title} {Temperature in non-equilibrium states: a review of open problems and current proposals},\ }\href {https://doi.org/10.1088/0034-4885/66/11/r03} {\bibfield  {journal} {\bibinfo  {journal} {Rep. Prog. Phys.}\ }\textbf {\bibinfo {volume} {66}},\ \bibinfo {pages} {1937} (\bibinfo {year} {2003})}\BibitemShut {NoStop}%
\bibitem [{\citenamefont {Lipka-Bartosik}\ \emph {et~al.}(2023)\citenamefont {Lipka-Bartosik}, \citenamefont {Perarnau-Llobet},\ and\ \citenamefont {Brunner}}]{Bartosik.2023.PRL}%
  \BibitemOpen
  \bibfield  {author} {\bibinfo {author} {\bibfnamefont {P.}~\bibnamefont {Lipka-Bartosik}}, \bibinfo {author} {\bibfnamefont {M.}~\bibnamefont {Perarnau-Llobet}},\ and\ \bibinfo {author} {\bibfnamefont {N.}~\bibnamefont {Brunner}},\ }\bibfield  {title} {\bibinfo {title} {{Operational definition of the temperature of a quantum state}},\ }\href {https://doi.org/10.1103/PhysRevLett.130.040401} {\bibfield  {journal} {\bibinfo  {journal} {Phys. Rev. Lett.}\ }\textbf {\bibinfo {volume} {130}},\ \bibinfo {pages} {040401} (\bibinfo {year} {2023})}\BibitemShut {NoStop}%
\bibitem [{\citenamefont {Daffertshofer}\ \emph {et~al.}(2002)\citenamefont {Daffertshofer}, \citenamefont {Plastino},\ and\ \citenamefont {Plastino}}]{Daffertshofer.2002.PRL}%
  \BibitemOpen
  \bibfield  {author} {\bibinfo {author} {\bibfnamefont {A.}~\bibnamefont {Daffertshofer}}, \bibinfo {author} {\bibfnamefont {A.~R.}\ \bibnamefont {Plastino}},\ and\ \bibinfo {author} {\bibfnamefont {A.}~\bibnamefont {Plastino}},\ }\bibfield  {title} {\bibinfo {title} {{Classical no-cloning theorem}},\ }\href {https://doi.org/10.1103/PhysRevLett.88.210601} {\bibfield  {journal} {\bibinfo  {journal} {Phys. Rev. Lett.}\ }\textbf {\bibinfo {volume} {88}},\ \bibinfo {pages} {210601} (\bibinfo {year} {2002})}\BibitemShut {NoStop}%
\bibitem [{\citenamefont {{Mulder, Daan and Wolde, Pieter Rein ten and Ouldridge, Thomas E.}}(2025)}]{Mulder.2025.PRR}%
  \BibitemOpen
  \bibfield  {author} {\bibinfo {author} {\bibnamefont {{Mulder, Daan and Wolde, Pieter Rein ten and Ouldridge, Thomas E.}}},\ }\bibfield  {title} {\bibinfo {title} {Exploiting bias in optimal finite-time copying protocols},\ }\href {https://doi.org/10.1103/PhysRevResearch.7.L012026} {\bibfield  {journal} {\bibinfo  {journal} {Phys. Rev. Res.}\ }\textbf {\bibinfo {volume} {7}},\ \bibinfo {pages} {L012026} (\bibinfo {year} {2025})}\BibitemShut {NoStop}%
\bibitem [{\citenamefont {Horowitz}\ and\ \citenamefont {Esposito}(2014)}]{Horowitz.2014.PRX}%
  \BibitemOpen
  \bibfield  {author} {\bibinfo {author} {\bibfnamefont {J.~M.}\ \bibnamefont {Horowitz}}\ and\ \bibinfo {author} {\bibfnamefont {M.}~\bibnamefont {Esposito}},\ }\bibfield  {title} {\bibinfo {title} {{Thermodynamics with continuous information flow}},\ }\href {https://doi.org/10.1103/PhysRevX.4.031015} {\bibfield  {journal} {\bibinfo  {journal} {Phys. Rev. X}\ }\textbf {\bibinfo {volume} {4}},\ \bibinfo {pages} {031015} (\bibinfo {year} {2014})}\BibitemShut {NoStop}%
\bibitem [{fnt({\natexlab{c}})}]{fnt3}%
  \BibitemOpen
  \href@noop {} {}A similar result can also be derived for general $p_e(\tau) \in [0,1]$.\BibitemShut {Stop}%
\bibitem [{\citenamefont {Gorini}\ \emph {et~al.}(1976)\citenamefont {Gorini}, \citenamefont {Kossakowski},\ and\ \citenamefont {Sudarshan}}]{Gorini.1976.JMP}%
  \BibitemOpen
  \bibfield  {author} {\bibinfo {author} {\bibfnamefont {V.}~\bibnamefont {Gorini}}, \bibinfo {author} {\bibfnamefont {A.}~\bibnamefont {Kossakowski}},\ and\ \bibinfo {author} {\bibfnamefont {E.~C.~G.}\ \bibnamefont {Sudarshan}},\ }\bibfield  {title} {\bibinfo {title} {{Completely positive dynamical semigroups of N-level systems}},\ }\href {https://doi.org/10.1063/1.522979} {\bibfield  {journal} {\bibinfo  {journal} {J. Math. Phys.}\ }\textbf {\bibinfo {volume} {17}},\ \bibinfo {pages} {821} (\bibinfo {year} {1976})}\BibitemShut {NoStop}%
\bibitem [{\citenamefont {Lindblad}(1976)}]{Lindblad.1976.CMP}%
  \BibitemOpen
  \bibfield  {author} {\bibinfo {author} {\bibfnamefont {G.}~\bibnamefont {Lindblad}},\ }\bibfield  {title} {\bibinfo {title} {{On the generators of quantum dynamical semigroups}},\ }\href {https://doi.org/10.1007/BF01608499} {\bibfield  {journal} {\bibinfo  {journal} {Commun. Math. Phys.}\ }\textbf {\bibinfo {volume} {48}},\ \bibinfo {pages} {119} (\bibinfo {year} {1976})}\BibitemShut {NoStop}%
\bibitem [{\citenamefont {Horowitz}\ and\ \citenamefont {Parrondo}(2013)}]{Horowitz.2013.NJP}%
  \BibitemOpen
  \bibfield  {author} {\bibinfo {author} {\bibfnamefont {J.~M.}\ \bibnamefont {Horowitz}}\ and\ \bibinfo {author} {\bibfnamefont {J.~M.~R.}\ \bibnamefont {Parrondo}},\ }\bibfield  {title} {\bibinfo {title} {{Entropy production along nonequilibrium quantum jump trajectories}},\ }\href {https://doi.org/10.1088/1367-2630/15/8/085028} {\bibfield  {journal} {\bibinfo  {journal} {New J. Phys.}\ }\textbf {\bibinfo {volume} {15}},\ \bibinfo {pages} {085028} (\bibinfo {year} {2013})}\BibitemShut {NoStop}%
\bibitem [{\citenamefont {Manzano}\ and\ \citenamefont {Zambrini}(2022)}]{Manzano.2022.QS}%
  \BibitemOpen
  \bibfield  {author} {\bibinfo {author} {\bibfnamefont {G.}~\bibnamefont {Manzano}}\ and\ \bibinfo {author} {\bibfnamefont {R.}~\bibnamefont {Zambrini}},\ }\bibfield  {title} {\bibinfo {title} {{Quantum thermodynamics under continuous monitoring: A general framework}},\ }\href {https://doi.org/10.1116/5.0079886} {\bibfield  {journal} {\bibinfo  {journal} {{AVS} Quantum Sci.}\ }\textbf {\bibinfo {volume} {4}},\ \bibinfo {pages} {025302} (\bibinfo {year} {2022})}\BibitemShut {NoStop}%
\bibitem [{\citenamefont {Van~Vu}\ and\ \citenamefont {Saito}(2024)}]{Vu.2024.PRA}%
  \BibitemOpen
  \bibfield  {author} {\bibinfo {author} {\bibfnamefont {T.}~\bibnamefont {Van~Vu}}\ and\ \bibinfo {author} {\bibfnamefont {K.}~\bibnamefont {Saito}},\ }\bibfield  {title} {\bibinfo {title} {{Geometric characterization for cyclic heat engines far from equilibrium}},\ }\href {https://doi.org/10.1103/PhysRevA.109.042209} {\bibfield  {journal} {\bibinfo  {journal} {Phys. Rev. A}\ }\textbf {\bibinfo {volume} {109}},\ \bibinfo {pages} {042209} (\bibinfo {year} {2024})}\BibitemShut {NoStop}%
\bibitem [{fnt({\natexlab{d}})}]{fnt4}%
  \BibitemOpen
  \href@noop {} {}This follows directly from the inequality $\tr{AB}\le\|A\|\|B\|_1$, where $\|.\|_1$ denotes the trace norm.\BibitemShut {Stop}%
\bibitem [{\citenamefont {Esposito}\ \emph {et~al.}(2010)\citenamefont {Esposito}, \citenamefont {Lindenberg},\ and\ \citenamefont {den Broeck}}]{Esposito.2010.NJP}%
  \BibitemOpen
  \bibfield  {author} {\bibinfo {author} {\bibfnamefont {M.}~\bibnamefont {Esposito}}, \bibinfo {author} {\bibfnamefont {K.}~\bibnamefont {Lindenberg}},\ and\ \bibinfo {author} {\bibfnamefont {C.~V.}\ \bibnamefont {den Broeck}},\ }\bibfield  {title} {\bibinfo {title} {{Entropy production as correlation between system and reservoir}},\ }\href {https://doi.org/10.1088/1367-2630/12/1/013013} {\bibfield  {journal} {\bibinfo  {journal} {New J. Phys.}\ }\textbf {\bibinfo {volume} {12}},\ \bibinfo {pages} {013013} (\bibinfo {year} {2010})}\BibitemShut {NoStop}%
\bibitem [{\citenamefont {Leggett}\ \emph {et~al.}(1987)\citenamefont {Leggett}, \citenamefont {Chakravarty}, \citenamefont {Dorsey}, \citenamefont {Fisher}, \citenamefont {Garg},\ and\ \citenamefont {Zwerger}}]{Leggett.1987.RMP}%
  \BibitemOpen
  \bibfield  {author} {\bibinfo {author} {\bibfnamefont {A.~J.}\ \bibnamefont {Leggett}}, \bibinfo {author} {\bibfnamefont {S.}~\bibnamefont {Chakravarty}}, \bibinfo {author} {\bibfnamefont {A.~T.}\ \bibnamefont {Dorsey}}, \bibinfo {author} {\bibfnamefont {M.~P.~A.}\ \bibnamefont {Fisher}}, \bibinfo {author} {\bibfnamefont {A.}~\bibnamefont {Garg}},\ and\ \bibinfo {author} {\bibfnamefont {W.}~\bibnamefont {Zwerger}},\ }\bibfield  {title} {\bibinfo {title} {{Dynamics of the dissipative two-state system}},\ }\href {https://doi.org/10.1103/RevModPhys.59.1} {\bibfield  {journal} {\bibinfo  {journal} {Rev. Mod. Phys.}\ }\textbf {\bibinfo {volume} {59}},\ \bibinfo {pages} {1} (\bibinfo {year} {1987})}\BibitemShut {NoStop}%
\bibitem [{\citenamefont {Miller}\ \emph {et~al.}(2020)\citenamefont {Miller}, \citenamefont {Guarnieri}, \citenamefont {Mitchison},\ and\ \citenamefont {Goold}}]{Miller.2020.PRL.QLP}%
  \BibitemOpen
  \bibfield  {author} {\bibinfo {author} {\bibfnamefont {H.~J.~D.}\ \bibnamefont {Miller}}, \bibinfo {author} {\bibfnamefont {G.}~\bibnamefont {Guarnieri}}, \bibinfo {author} {\bibfnamefont {M.~T.}\ \bibnamefont {Mitchison}},\ and\ \bibinfo {author} {\bibfnamefont {J.}~\bibnamefont {Goold}},\ }\bibfield  {title} {\bibinfo {title} {{Quantum fluctuations hinder finite-time information erasure near the Landauer limit}},\ }\href {https://doi.org/10.1103/PhysRevLett.125.160602} {\bibfield  {journal} {\bibinfo  {journal} {Phys. Rev. Lett.}\ }\textbf {\bibinfo {volume} {125}},\ \bibinfo {pages} {160602} (\bibinfo {year} {2020})}\BibitemShut {NoStop}%
\bibitem [{\citenamefont {Annby-Andersson}\ \emph {et~al.}(2022)\citenamefont {Annby-Andersson}, \citenamefont {Bakhshinezhad}, \citenamefont {Bhattacharyya}, \citenamefont {De~Sousa}, \citenamefont {Jarzynski}, \citenamefont {Samuelsson},\ and\ \citenamefont {Potts}}]{Andersson.2022.PRL}%
  \BibitemOpen
  \bibfield  {author} {\bibinfo {author} {\bibfnamefont {B.}~\bibnamefont {Annby-Andersson}}, \bibinfo {author} {\bibfnamefont {F.}~\bibnamefont {Bakhshinezhad}}, \bibinfo {author} {\bibfnamefont {D.}~\bibnamefont {Bhattacharyya}}, \bibinfo {author} {\bibfnamefont {G.}~\bibnamefont {De~Sousa}}, \bibinfo {author} {\bibfnamefont {C.}~\bibnamefont {Jarzynski}}, \bibinfo {author} {\bibfnamefont {P.}~\bibnamefont {Samuelsson}},\ and\ \bibinfo {author} {\bibfnamefont {P.~P.}\ \bibnamefont {Potts}},\ }\bibfield  {title} {\bibinfo {title} {{Quantum Fokker-Planck master equation for continuous feedback control}},\ }\href {https://doi.org/10.1103/PhysRevLett.129.050401} {\bibfield  {journal} {\bibinfo  {journal} {Phys. Rev. Lett.}\ }\textbf {\bibinfo {volume} {129}},\ \bibinfo {pages} {050401} (\bibinfo {year} {2022})}\BibitemShut {NoStop}%
\bibitem [{\citenamefont {Zhang}\ \emph {et~al.}(2019)\citenamefont {Zhang}, \citenamefont {Cao}, \citenamefont {Ouyang},\ and\ \citenamefont {Tu}}]{Zhang.2019.NP}%
  \BibitemOpen
  \bibfield  {author} {\bibinfo {author} {\bibfnamefont {D.}~\bibnamefont {Zhang}}, \bibinfo {author} {\bibfnamefont {Y.}~\bibnamefont {Cao}}, \bibinfo {author} {\bibfnamefont {Q.}~\bibnamefont {Ouyang}},\ and\ \bibinfo {author} {\bibfnamefont {Y.}~\bibnamefont {Tu}},\ }\bibfield  {title} {\bibinfo {title} {{The energy cost and optimal design for synchronization of coupled molecular oscillators}},\ }\href {https://doi.org/10.1038/s41567-019-0701-7} {\bibfield  {journal} {\bibinfo  {journal} {Nat. Phys.}\ }\textbf {\bibinfo {volume} {16}},\ \bibinfo {pages} {95} (\bibinfo {year} {2019})}\BibitemShut {NoStop}%
\bibitem [{\citenamefont {Rao}\ and\ \citenamefont {Esposito}(2016)}]{Rao.2016.PRX}%
  \BibitemOpen
  \bibfield  {author} {\bibinfo {author} {\bibfnamefont {R.}~\bibnamefont {Rao}}\ and\ \bibinfo {author} {\bibfnamefont {M.}~\bibnamefont {Esposito}},\ }\bibfield  {title} {\bibinfo {title} {{Nonequilibrium thermodynamics of chemical reaction networks: Wisdom from stochastic thermodynamics}},\ }\href {https://doi.org/10.1103/PhysRevX.6.041064} {\bibfield  {journal} {\bibinfo  {journal} {Phys. Rev. X}\ }\textbf {\bibinfo {volume} {6}},\ \bibinfo {pages} {041064} (\bibinfo {year} {2016})}\BibitemShut {NoStop}%
\bibitem [{\citenamefont {Nakajima}\ and\ \citenamefont {Tajima}(2024)}]{Nakajima.2024.arxiv}%
  \BibitemOpen
  \bibfield  {author} {\bibinfo {author} {\bibfnamefont {S.}~\bibnamefont {Nakajima}}\ and\ \bibinfo {author} {\bibfnamefont {H.}~\bibnamefont {Tajima}},\ }\bibfield  {title} {\bibinfo {title} {{Speed-accuracy trade-off relations in quantum measurements and computations}},\ }\href {https://arxiv.org/abs/2405.15291} {\bibfield  {journal} {\bibinfo  {journal} {arXiv:2405.15291}\ } (\bibinfo {year} {2024})}\BibitemShut {NoStop}%
\bibitem [{\citenamefont {Guryanova}\ \emph {et~al.}(2020)\citenamefont {Guryanova}, \citenamefont {Friis},\ and\ \citenamefont {Huber}}]{Guryanova.2020.Q}%
  \BibitemOpen
  \bibfield  {author} {\bibinfo {author} {\bibfnamefont {Y.}~\bibnamefont {Guryanova}}, \bibinfo {author} {\bibfnamefont {N.}~\bibnamefont {Friis}},\ and\ \bibinfo {author} {\bibfnamefont {M.}~\bibnamefont {Huber}},\ }\bibfield  {title} {\bibinfo {title} {{Ideal projective measurements have infinite resource costs}},\ }\href {https://doi.org/10.22331/q-2020-01-13-222} {\bibfield  {journal} {\bibinfo  {journal} {Quantum}\ }\textbf {\bibinfo {volume} {4}},\ \bibinfo {pages} {222} (\bibinfo {year} {2020})}\BibitemShut {NoStop}%
\bibitem [{\citenamefont {Mohammady}\ and\ \citenamefont {Miyadera}(2023)}]{Mohammady.2023.PRA}%
  \BibitemOpen
  \bibfield  {author} {\bibinfo {author} {\bibfnamefont {M.~H.}\ \bibnamefont {Mohammady}}\ and\ \bibinfo {author} {\bibfnamefont {T.}~\bibnamefont {Miyadera}},\ }\bibfield  {title} {\bibinfo {title} {{Quantum measurements constrained by the third law of thermodynamics}},\ }\href {https://doi.org/10.1103/PhysRevA.107.022406} {\bibfield  {journal} {\bibinfo  {journal} {Phys. Rev. A}\ }\textbf {\bibinfo {volume} {107}},\ \bibinfo {pages} {022406} (\bibinfo {year} {2023})}\BibitemShut {NoStop}%
\bibitem [{\citenamefont {Cleuren}\ \emph {et~al.}(2012)\citenamefont {Cleuren}, \citenamefont {Rutten},\ and\ \citenamefont {Van~den Broeck}}]{Cleuren.2012.PRL}%
  \BibitemOpen
  \bibfield  {author} {\bibinfo {author} {\bibfnamefont {B.}~\bibnamefont {Cleuren}}, \bibinfo {author} {\bibfnamefont {B.}~\bibnamefont {Rutten}},\ and\ \bibinfo {author} {\bibfnamefont {C.}~\bibnamefont {Van~den Broeck}},\ }\bibfield  {title} {\bibinfo {title} {{Cooling by heating: Refrigeration powered by photons}},\ }\href {https://doi.org/10.1103/PhysRevLett.108.120603} {\bibfield  {journal} {\bibinfo  {journal} {Phys. Rev. Lett.}\ }\textbf {\bibinfo {volume} {108}},\ \bibinfo {pages} {120603} (\bibinfo {year} {2012})}\BibitemShut {NoStop}%
\bibitem [{\citenamefont {Kosloff}(2013)}]{Kosloff.2013.E}%
  \BibitemOpen
  \bibfield  {author} {\bibinfo {author} {\bibfnamefont {R.}~\bibnamefont {Kosloff}},\ }\bibfield  {title} {\bibinfo {title} {{Quantum thermodynamics: A dynamical viewpoint}},\ }\href {https://doi.org/10.3390/e15062100} {\bibfield  {journal} {\bibinfo  {journal} {Entropy}\ }\textbf {\bibinfo {volume} {15}},\ \bibinfo {pages} {2100} (\bibinfo {year} {2013})}\BibitemShut {NoStop}%
\bibitem [{\citenamefont {Shiraishi}\ and\ \citenamefont {Saito}(2019)}]{Shiraishi.2019.PRL}%
  \BibitemOpen
  \bibfield  {author} {\bibinfo {author} {\bibfnamefont {N.}~\bibnamefont {Shiraishi}}\ and\ \bibinfo {author} {\bibfnamefont {K.}~\bibnamefont {Saito}},\ }\bibfield  {title} {\bibinfo {title} {{Information-theoretical bound of the irreversibility in thermal relaxation processes}},\ }\href {https://doi.org/10.1103/PhysRevLett.123.110603} {\bibfield  {journal} {\bibinfo  {journal} {Phys. Rev. Lett.}\ }\textbf {\bibinfo {volume} {123}},\ \bibinfo {pages} {110603} (\bibinfo {year} {2019})}\BibitemShut {NoStop}%
\bibitem [{\citenamefont {Kolchinsky}\ \emph {et~al.}(2022)\citenamefont {Kolchinsky}, \citenamefont {Dechant}, \citenamefont {Yoshimura},\ and\ \citenamefont {Ito}}]{Kolchinsky.2022.arxiv}%
  \BibitemOpen
  \bibfield  {author} {\bibinfo {author} {\bibfnamefont {A.}~\bibnamefont {Kolchinsky}}, \bibinfo {author} {\bibfnamefont {A.}~\bibnamefont {Dechant}}, \bibinfo {author} {\bibfnamefont {K.}~\bibnamefont {Yoshimura}},\ and\ \bibinfo {author} {\bibfnamefont {S.}~\bibnamefont {Ito}},\ }\bibfield  {title} {\bibinfo {title} {{Information geometry of excess and housekeeping entropy production}},\ }\href {https://arxiv.org/abs/2206.14599} {\bibfield  {journal} {\bibinfo  {journal} {arXiv:2206.14599}\ } (\bibinfo {year} {2022})}\BibitemShut {NoStop}%
\end{thebibliography}
\end{document}